\documentclass[aps,prx,twocolumn,nofootinbib]{revtex4-2}
\pdfoutput=1
\usepackage{hyperref}
\usepackage{mathtools}
\usepackage{lipsum}
\usepackage{amsmath}
\usepackage{xr}
\usepackage{amsfonts}
\usepackage{layouts}
\usepackage{graphicx}
\usepackage{lipsum}
\usepackage{algorithm}
\usepackage{tabularx}
\usepackage{algpseudocode}
\usepackage{comment}
\usepackage{braket}
\usepackage{tikz}
\usetikzlibrary {arrows.meta}
\usepackage{amsthm}
\newtheorem{theorem}{Theorem}
\newtheorem{lemma}{Lemma}
\newtheorem{corollary}{Corollary}
\usepackage{ytableau}
\ytableausetup{boxsize=1.2em}
\ytableausetup{centertableaux}

\hypersetup{
    colorlinks=true,
    linkcolor=blue,
    filecolor=magenta,      
    urlcolor=blue,
    citecolor = red,
    }
\begin{document}

\title{Mitigating the barren plateau problem in linear optics}
\author{Matthew D. Horner}
\affiliation{Aegiq Ltd., Cooper Buildings, Arundel Street, Sheffield, S1 2NS, United Kingdom}

\begin{abstract}
We prove that variational quantum algorithms using discrete variable boson samplers suffer from barren plateaus and that by changing to fermionic statistics, which is classically efficient to simulate, the barren plateau scaling is less severe. Moreover, when using fermionic statistics the cost landscape boils down to one unit frequency Fourier mode per variational parameter, irrespective of circuit design, ansatz or variational problem, allowing for more efficient parameter shift rules. This simplified cost landscape allows for the gradient-free Rotosolve algorithm to be used to gain a considerable speed-up over the original bosonic algorithm. Motivated by this, we propose a non-linear parametrised phase shifter, the dual-valued phase shifter, that generates these single-harmonic cost landscapes for bosonic statistics. We propose three ways to achieve this by using either non-linear optics, measurement-induced non-linearities, or entangled resource states simulating fermionic statistics. The latter two require linear optics only, allowing for implementation with widely-available technology today. Despite this, the classical simulability of free fermions and the fact the fermionic algorithm can defeat the bosonic algorithm casts doubt on whether the quantum advantage of boson sampling results in any practical advantage in this setting.
\end{abstract}
\maketitle

\tableofcontents

\section{Introduction}
Variational quantum algorithms~\cite{cerezo2021variational,abbas2024challenges} have emerged as a common application of quantum computing as they do not require fault tolerant systems and are widely applicable to many real-world problems that are NP-hard. Example problems include the travelling salesman problem~\cite{qubo_list,tsp,CVRP,salehi2022unconstrained,RevModPhys.80.1061}, SAT solvers~\cite{qubo_list,sat,3sat,3sat_2,BIAN2020104609,glover2019tutorialformulatingusingqubo}, job shop scheduling~\cite{venturelli2016quantumannealingimplementationjobshop,10.1007/978-3-031-08011-1_10,DENKENA2021100}, maximum cut~\cite{glover2019tutorialformulatingusingqubo}, and graph colouring problems~\cite{glover2019tutorialformulatingusingqubo}, among others~\cite{qubo_list,glover2019tutorialformulatingusingqubo}. At the core of these algorithms lies an optimisation, where the task is to find the minimum of a cost function encoding the problem, typically done using gradient descent. 

A barrier to the success of these algorithms is the barren plateau problem~\cite{mcclean2018barren,larocca2024review,TILLY20221} which states that the variance of either the cost function or its gradient becomes exponentially small as the problem size scales up, reducing the chance of finding the minimum. This reduces the practical utility of these algorithms for solving large real-world problems. Various tricks to mitigate the barren plateau problem exist including a careful choice of ansatz~\cite{Holmes_2022,grimsley2023adaptive,Wada_2024,PRXQuantum.2.020310}, reducing the expressibility of the circuit and shallow circuits~\cite{Holmes_2022,cerezo2021cost,PhysRevLett.132.150603,leone2022practical}, initial parameter optimisation~\cite{sauvage2021flip,9951195}, modifying the cost function~\cite{Wu2021MitigatingNG,PhysRevResearch.3.033090}, identifying symmetries in the problem~\cite{Lyu2023symmetryenhanced,PhysRevLett.125.260505,PhysRevA.79.042335,PhysRevA.101.052340,PhysRevA.98.022322}, or by off-loading work to the classical optimiser~\cite{lerch2024efficientquantumenhancedclassicalsimulation,10636813}. However, these heavily rely upon classical pre-processing instead of anything inherently quantum, further increasing the reliance on classical methods and relegating the role of the quantum computer.

Here we focus on variational quantum algorithms implemented with discrete variable boson samplers. These algorithms work by sampling bit strings from the output of a parametrised linear optical interferometer and minimising a cost function~\cite{bradler2021certain}. We demonstrate both numerically and analytically that boson samplers suffer from barren plateaus and generate needlessly complicated cost landscapes for the problems they attempt to solve. We show that by using fermionic statistics instead, we can reduce the severity of barren plateaus by orders of magnitude and reduce the number of Fourier modes per variational parameter down to at most one of unit frequency, resulting in a far simpler cost landscape irrespective of the circuit, input ansatz or variational problem we are solving. Moreover, this allows for the gradient-free Rotosolve algorithm to be used to gain a considerable speed up~\cite{rotosolve1,rotosolve2,rotosolve3,rotosolve4,rotosolve5} and reduces the sample complexity of the algorithm. As fermion sampling is efficient to simulate classically, we have shown that bosonic linear optical VQAs can be defeated by a classical algorithm.

\begin{figure}
\begin{tikzpicture}[x=0.75pt,y=0.75pt,yscale=-1,xscale=1]

\draw [line width=0.75]    (74,70) -- (94,70) ;
\draw [line width=0.75]    (74,80) -- (94,80) ;
\draw [line width=0.75]    (74,90) -- (94,90) ;
\draw [line width=0.75]    (74,100) -- (94,100) ;
\draw [line width=0.75]    (74,110) -- (94,110) ;
\draw [line width=0.75]    (114,100) -- (124,100) ;
\draw [line width=0.75]    (114,90) -- (124,90) ;
\draw [line width=0.75]    (114,70) -- (124,70) ;
\draw [line width=0.75]    (74,60) -- (94,60) ;
\draw [line width=0.75]    (94,60) -- (104,70) ;
\draw [line width=0.75]    (94,70) -- (104,60) ;
\draw [line width=0.75]    (94,80) -- (104,90) ;
\draw [line width=0.75]    (94,90) -- (104,80) ;
\draw [line width=0.75]    (94,100) -- (104,110) ;
\draw [line width=0.75]    (94,110) -- (104,100) ;
\draw [line width=0.75]    (124,70) -- (134,80) ;
\draw [line width=0.75]    (124,80) -- (134,70) ;
\draw [line width=0.75]    (114,80) -- (124,80) ;
\draw [line width=0.75]    (124,90) -- (134,100) ;
\draw [line width=0.75]    (124,100) -- (134,90) ;
\draw [line width=0.75]    (114,60) -- (154,60) ;
\draw [line width=0.75]    (134,70) -- (154,70) ;
\draw [line width=0.75]    (134,80) -- (154,80) ;
\draw [line width=0.75]    (134,90) -- (154,90) ;
\draw [line width=0.75]    (134,100) -- (154,100) ;
\draw [line width=0.75]    (114,110) -- (154,110) ;
\draw  [fill={rgb, 255:red, 0; green, 0; blue, 0 }  ,fill opacity=1, line width=1.5]  (154,68) -- (156,68) .. controls (157.1,68) and (158,68.9) .. (158,70) .. controls (158,71.1) and (157.1,72) .. (156,72) -- (154,72) -- cycle ;
\draw  [fill={rgb, 255:red, 0; green, 0; blue, 0 }  ,fill opacity=1, line width=1.5]  (154,88) -- (156,88) .. controls (157.1,88) and (158,88.9) .. (158,90) .. controls (158,91.1) and (157.1,92) .. (156,92) -- (154,92) -- cycle ;
\draw  [fill={rgb, 255:red, 0; green, 0; blue, 0 }  ,fill opacity=1, line width=1.5]  (154,98) -- (156,98) .. controls (157.1,98) and (158,98.9) .. (158,100) .. controls (158,101.1) and (157.1,102) .. (156,102) -- (154,102) -- cycle ;
\draw  [fill={rgb, 255:red, 0; green, 0; blue, 0 }  ,fill opacity=1, line width=1.5]  (154,108) -- (156,108) .. controls (157.1,108) and (158,108.9) .. (158,110) .. controls (158,111.1) and (157.1,112) .. (156,112) -- (154,112) -- cycle ;
\draw   (258,66) -- (330,66) -- (330,102) -- (258,102) -- cycle ;
\draw    (114,28.4) -- (294,28.4) ;
\draw    (294,28) -- (294,66) ;
\draw  [fill={rgb, 255:red, 0; green, 0; blue, 0 }  ,fill opacity=1 ] (72,60) .. controls (72,58.9) and (72.9,58) .. (74,58) .. controls (75.1,58) and (76,58.9) .. (76,60) .. controls (76,61.1) and (75.1,62) .. (74,62) .. controls (72.9,62) and (72,61.1) .. (72,60) -- cycle ;
\draw  [fill={rgb, 255:red, 0; green, 0; blue, 0 }  ,fill opacity=1 ] (72,70) .. controls (72,68.9) and (72.9,68) .. (74,68) .. controls (75.1,68) and (76,68.9) .. (76,70) .. controls (76,71.1) and (75.1,72) .. (74,72) .. controls (72.9,72) and (72,71.1) .. (72,70) -- cycle ;
\draw  [fill={rgb, 255:red, 245; green, 166; blue, 35 }  ,fill opacity=1 ] (156.34,74.58) -- (157.78,77.38) -- (161.02,77.83) -- (158.68,80) -- (159.23,83.07) -- (156.34,81.62) -- (153.45,83.07) -- (154,80) -- (151.66,77.83) -- (154.89,77.38) -- cycle ;
\draw  [fill={rgb, 255:red, 245; green, 166; blue, 35 }  ,fill opacity=1 ] (156.34,54.58) -- (157.78,57.38) -- (161.02,57.83) -- (158.68,60) -- (159.23,63.07) -- (156.34,61.62) -- (153.45,63.07) -- (154,60) -- (151.66,57.83) -- (154.89,57.38) -- cycle ;
\draw [line width=0.75]  [dash pattern={on 1pt off 1pt}]  (104,100) -- (114,100) ;
\draw [line width=0.75]  [dash pattern={on 1pt off 1pt}]  (104,110) -- (114,110) ;
\draw [line width=0.75]  [dash pattern={on 1pt off 1pt}]  (104,90) -- (114,90) ;
\draw [line width=0.75]  [dash pattern={on 1pt off 1pt}]  (104,80) -- (114,80) ;
\draw [line width=0.75]  [dash pattern={on 1pt off 1pt}]  (104,70) -- (114,70) ;
\draw [line width=0.75]  [dash pattern={on 1pt off 1pt}]  (104,60) -- (114,60) ;
\draw   [line width=0.08] (114,28.4) -- (114,47.4) ;
\draw [shift={(114,50.4)}, rotate = 270] [fill={rgb, 255:red, 0; green, 0; blue, 0 }  ][line width=0.08]  [draw opacity=0] (10.72,-5.15) -- (0,0) -- (10.72,5.15) -- (7.12,0) -- cycle    ;
\draw    (175,84.75) -- (189,84.75) ;
\draw [shift={(192,84.75)}, rotate = 180] [fill={rgb, 255:red, 0; green, 0; blue, 0 }  ][line width=0.08]  [draw opacity=0] (10.72,-5.15) -- (0,0) -- (10.72,5.15) -- (7.12,0) -- cycle    ;
\draw  (228.22,84.5) -- (250.84,84.5) ;
\draw [shift={(253.84,84.75)}, rotate = 179.86] [fill={rgb, 255:red, 0; green, 0; blue, 0 }  ][line width=0.08]  [draw opacity=0] (10.72,-5.15) -- (0,0) -- (10.72,5.15) -- (7.12,0) -- cycle    ;
\draw [line width=0.75]    (56,158) -- (76,158) ;
\draw [line width=0.75]    (56,182) -- (76,182) ;
\draw [line width=0.75]    (76,158) -- (100,182) ;
\draw [line width=0.75]    (76,182) -- (100,158) ;
\draw [line width=0.75]    (100,158) -- (120,158) ;
\draw [line width=0.75]    (100,182) -- (120,182) ;
\draw [line width=0.75]    (188,158) -- (208,158) ;
\draw [line width=0.75]    (208,158) -- (232,182) ;
\draw [line width=0.75]    (208,182) -- (232,158) ;
\draw [line width=0.75]    (232,158) -- (252,158) ;
\draw [line width=0.75]    (232,182) -- (288,182) ;
\draw [line width=0.75]    (268,158) -- (288,158) ;
\draw [line width=0.75]    (288,158) -- (312,182) ;
\draw [line width=0.75]    (288,182) -- (312,158) ;
\draw [line width=0.75]    (312,158) -- (332,158) ;
\draw [line width=0.75]    (312,182) -- (332,182) ;
\draw   (252,150) -- (268,150) -- (268,166) -- (252,166) -- cycle ;
\draw   (172,150) -- (188,150) -- (188,166) -- (172,166) -- cycle ;
\draw [line width=0.75]    (152,182) -- (208,182) ;
\draw [line width=0.75]    (152,158) -- (172,158) ;
\draw  [fill={rgb, 255:red, 0; green, 0; blue, 0 }  ,fill opacity=1 ] (208,168) -- (232,168) -- (232,172) -- (208,172) -- cycle ;
\draw  [fill={rgb, 255:red, 0; green, 0; blue, 0 }  ,fill opacity=1 ] (288,168) -- (312,168) -- (312,172) -- (288,172) -- cycle ;

\draw (96,113.4) node [anchor=north west][inner sep=0.75pt]    {$U(\boldsymbol{\theta } )$};
\draw (26,50) node [anchor=north west][inner sep=0.75pt]    {$|\psi _{\text{in}} \rangle \begin{cases}
 & \\
 & \\
 & \\
 &
\end{cases}$};
\draw (194,46.4) node [anchor=north west][inner sep=0.75pt]    {$\begin{pmatrix}
1\\
0\\
1\\
\vdots 
\end{pmatrix}$};
\draw (279,76.4) node [anchor=north west][inner sep=0.75pt]    {$E(\boldsymbol{\theta } )$};
\draw (171,12) node [anchor=north west][inner sep=0.75pt]   [align=left] {Update $\displaystyle \boldsymbol{\theta }$};
\draw (268,106) node [anchor=north west][inner sep=0.75pt]   [align=left] {Classical \\ computer};
\draw (26,14) node [anchor=north west][inner sep=0.75pt]   [align=left] {(a)};
\draw (26,146) node [anchor=north west][inner sep=0.75pt]   [align=left] {(b)};
\draw (125.6,165.8) node [anchor=north west][inner sep=0.75pt]    {$=$};
\draw (174.5,151) node [anchor=north west][inner sep=0.75pt]    {$\phi $};
\draw (252,152) node [anchor=north west][inner sep=0.75pt]   {$2\theta $};
\draw (204.4,184.8) node [anchor=north west][inner sep=0.75pt]   [align=left] {50:50};
\draw (284.4,184.8) node [anchor=north west][inner sep=0.75pt]   [align=left] {50:50};
\draw (179,50) node [anchor=north east] [inner sep=0.75pt]    {$\begin{drcases}
 & \\
 & \\
 & \\
 &
\end{drcases}$};

\end{tikzpicture}
\caption{(a) The bosonic linear optics variational quantum algorithm of Ref.~\cite{bradler2021certain} consists of a linear optical interferometer encoding a parametrised unitary $U(\boldsymbol{\theta})$ and photo-detectors. Here we show an example with a photon inserted into the top two modes on the left, represented by the solid circles, where the stars represent detector clicks which maps to a bit string. The classical computer calculates the cost function $E(\boldsymbol{\theta})$ from multiple shots of this and then updates the parameters in order to optimise this. (b) Each cross-over point corresponds to a Mach-Zehnder interferometer consisting of two phase shifters $\theta,\phi \in [0,2\pi)$ and two fixed 50:50 beamsplitters. \label{fig:variational_quantum_solver}}
\end{figure}
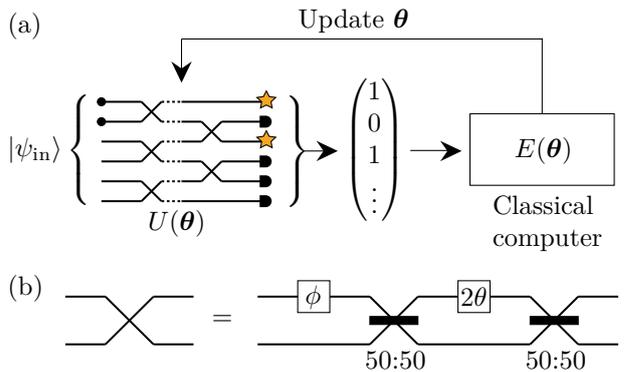

Inspired by the effects of fermionic statistics, we devise a non-linear phase shifter, the dual-valued phase shifter, that is a phase shifter with two distinct eigenvalues resulting in the same form of cost landscape as fermionic statistics. We provide three ways to realise this non-linear phase shifter. First by direct implementation using non-linear optical components, second by measurement-induced non-linearities with linear optics; and third by performing fermion sampling with an entangled resource state. The second and third ways require single photon sources, linear optics and single photon detectors, which is widely available technology today. As these three methods result in cost landscapes of the same form, for numerical convenience we focus on the performance of the fermionic resource state, thereby comparing the performance of fermionic linear optics to bosonic linear optics for solving quantum variational problems. 

The paper is structured as follows. First we review the variational quantum algorithm for solving QUBO problems with linear optics. We then present a theorem for the barren plateau scaling of linear optics. Then we study the difference between the cost landscapes of fermionic and bosonic linear optics, which then motivates us to introduce a non-linear phase shifter, the dual-valued phase shifter (DVPS), and study its resulting cost function. Then we present three different realisations of the DVPS and compare them and their resources. Then the remainder of the paper presents numerical results comparing the performance of fermion and bosonic linear optics for solving random QUBO problems, and testing the application of the gradient-free Rotosolve algorithm. We conclude the paper with open problems that arose during this work.

\section{The barren plateau problem in linear optics \label{sec:simplifying_cost_landscape}}
\subsection{Variational quantum algorithms using linear optics \label{sec:vqe}}

An $N$-mode linear optical interferometer is a network of $N$ waveguides, parametrised phase shifters, 50:50 beamsplitters and single-particle detection, as shown in Fig.~\ref{fig:variational_quantum_solver}. The particles inserted into these devices are indistinguishable and can be either fermionic or bosonic, so the only degree of freedom available is the waveguide degree of freedom which we refer to as the modes. 


The states of this system are described by a Fock space $\mathcal{F}$ spanned by the number states $|\mathbf{n}\rangle = |n_1,n_2,\ldots,n_N\rangle$, where $n_i$ is the number of particles in the $i$th mode. For bosons $n_i \in \mathbb{N}$, whereas for fermions $n_i \in \{ 0 , 1\}$ due to the Pauli exclusion principle. Acting on the Fock space is a set of ladder operators, $a_i^\dagger$ and $a_i$, that create and annihilate particles in the $i$th mode and obey the commutation relations $[a_i, a_j^\dagger]_\pm  = \delta_{ij}$ and $[a_i, a_j]_\pm = [a_i^\dagger, a_j^\dagger]_\pm  = 0$, where $\pm$ corresponds to the anti-commutator for fermions or the commutator for bosons, respectively. 

A linear optical interferometer applies a non-interacting and particle conserving unitary transformation $U: \mathcal{F} \to \mathcal{F}$ whose action in the Heisenberg picture is a linear transformation on the space of ladder operators as
\begin{equation}
U a_i^\dagger U^\dagger= \sum_{j = 1}^N u_{ji} a_j^\dagger, \label{eq:linear_unitary}
\end{equation}
where $u_{ji}$ are the components of a matrix $u \in \mathrm{U}(N)$. The unitaries $u$ which are programmed into the interferometer are parametrised by a vector of phases $\boldsymbol{\theta}=(\theta_1,\theta_2,\ldots)$, where $\theta_a \in [0,2\pi)$ is the phase shift of the $a$th phase shifter in the interferometer. There exists universal interferometers that can encode any $u \in \mathrm{U}(N)$~\cite{PhysRevLett.73.58,Clements:16,Fldzhyan:20}, however this does not mean that the interferometer is a universal quantum computer as only unitaries $U$ which transform ladder operators linearly as in Eq.~\eqref{eq:linear_unitary} can be encoded on the Fock space. In other words, this can only encode unitaries generated by \textit{quadratic} particle-conserving Hamiltonians. For universal quantum computing with linear optics then the KLM protocol is required~\cite{KLM,Kok_Lovett_2010}. 

Inserting a number state $|\mathbf{n}\rangle$ into a linear optical interferometer and sampling from the output is called boson sampling or fermion sampling depending upon the statistics of the particles~\cite{arkhipov2012bosonic,scheel2004permanents,valiant2001quantum,PhysRevA.65.032325,knill2001fermioniclinearopticsmatchgates,PRXQuantum.3.020328}. The transition amplitudes between $n$-particle input and output states of one of these interferometers is given by
\begin{equation}
\langle \mathbf{m} | U |\mathbf{n} \rangle = \frac{1}{\sqrt{\mathbf{m}! \mathbf{n}!}} \begin{cases}
\operatorname{per} u[\mathbf{m}|\mathbf{n}] & \text{bosons} \\
\det u[\mathbf{m}|\mathbf{n}] & \text{fermions}
\end{cases}, \label{eq:transition_amplitude}
\end{equation}
where per is the permanent, det is the determinant, $\mathbf{n}! = n_1! n_2! \ldots n_N!$ and $u[\mathbf{m}|\mathbf{n}]$ is an $n$-dimensional matrix constructed from the elements of $u$ from Eq.~\eqref{eq:linear_unitary} by repeating the $i$th row index $m_i$ times and the $j$th column index $n_j$ times. Calculating the permanent is a \#P problem which gives boson sampling its quantum advantage~\cite{doi:10.1126/science.abe8770,10.1145/1993636.1993682,PhysRevLett.123.250503,VALIANT1979189,madsen2022quantum}, whilst calculating the determinant is a P problem meaning fermion sampling is efficiently classically simulable~\cite{valiant2001quantum,PhysRevA.65.032325,knill2001fermioniclinearopticsmatchgates,PRXQuantum.3.020328}. See Sec.~\ref{app:boson_fermion_sampling} of the appendix for an overview.

Many variational quantum algorithms encode a problem we wish to solve into an observable $H$ such that the ground state corresponds to the solution to the problem. To find this ground state, the quantum computer is prepared in the state $\rho $ and is evolved by the parametrised unitary $U(\boldsymbol{\theta})$. The expectation value, or cost function, of the observable at the output is
\begin{equation}
E(\boldsymbol{\theta}) = \mathrm{Tr} \left[ U(\boldsymbol{\theta}) \rho U^\dagger(\boldsymbol{\theta}) H \right], \label{eq:general_cost_function}
\end{equation}
which is calculated by a classical computer using measurement data from multiple shots. The classical computer then attempts to find the minimum of the cost function by optimising over the parameters $\boldsymbol{\theta}$. With the optimal parameters, $\boldsymbol{\theta}_0$, the output state $\rho(\boldsymbol{\theta}_0) = U(\boldsymbol{\theta}_0) \rho U^\dagger(\boldsymbol{\theta}_0)$ is returned as the solution to the problem. In Fig.~\ref{fig:variational_quantum_solver} we show how a variational quantum algorithm can be performed with a linear optical interferometer and particle detectors, where the variational parameters are the phases of the phase shifters.

One application of this is to solve quadratic unconstrained binary optimisation (QUBO) problems. Given an $N$-bit QUBO problem, the goal is to find the $N$-bit string $\mathbf{x} = (x_1,x_2,\ldots,x_N)$ that minimises the quadratic cost function
\begin{equation}
C(\mathbf{x}) = \sum_{i,j = 1}^N Q_{ij} x_i  x_j, \label{eq:QUBO_cost}
\end{equation}
where $Q$ is an $N \times N$ real symmetric matrix~\cite{qubo_list}. This can be encoded into a linear optical variational quantum algorithm by introducing the Hamiltonian
\begin{equation}
H = \sum_{i,j = 1}^N Q_{ij} \Theta(\hat{n}_i)  \Theta(\hat{n}_j), \label{eq:qubit_hamiltonian}
\end{equation}
where $\hat{n}_i = a^\dagger_i a_i$ is the number operator for the $i$th mode and $\Theta$ is the Heaviside step function (using the convention that $\Theta(0) = 0$) which acts on the number operator as $\Theta(\hat{n}_i) |\mathbf{n}\rangle = \Theta(n_i) |\mathbf{n}\rangle$. This step function models the effect of a threshold photo-detector that can count only whether there was at least one particle or not and naturally maps the outputs to bit strings. As multiple bosons can occupy the same mode, multiple outputs from a boson sampler will be indistinguishable after mapping to bit strings which introduces redundancy, however fermions do not suffer from this as all fermionic states are fixed Hamming weight bit strings due to the Pauli exclusion principle, i.e., $\Theta(\hat{n}_i) = \hat{n}_i$. 

Inserting the observable of Eq.~\eqref{eq:qubit_hamiltonian} into the general expression for the cost function of Eq.~\eqref{eq:general_cost_function} yields 
\begin{equation}
E(\boldsymbol{\theta})  =  \sum_\mathbf{x} p(\mathbf{x}|\boldsymbol{\theta}) C(\mathbf{x}), \label{eq:quantum_cost}
\end{equation}
which is the expectation value of the classical cost function from Eq.~\eqref{eq:QUBO_cost}, where $p(\mathbf{x}|\boldsymbol{\theta})$ is the probability for the output state to yield the bit string $\mathbf{x} = \Theta(\mathbf{n})$ after thresholding given the parameters $\boldsymbol{\theta}$ of the interferometer. This quantity is then minimised.

This is the best known method for solving QUBO problems with a boson sampler which was first presented in Ref.~\cite{bradler2021certain} and can be easily used for any binary cost function $C(\mathbf{x})$ beyond QUBO. Building upon this algorithm is the focus of this study.

\subsection{Barren plateaus \label{sec:barren_plateau}}
A huge barrier to the success of gradient-based optimisers is the barren plateau problem. For qubit-based systems, this means the variance of the cost function decays exponentially as
\begin{equation}
\mathrm{Var}_{\boldsymbol{\theta}}[E(\boldsymbol{\theta})] = O(1/b^M)
\end{equation}
where $b > 1$ and $M$ is the number of qubits~\cite{ragone2024lie,larocca2024review}. An alternative definition in terms of the variance of the gradient is also used~\cite{mcclean2018barren, fontanacharacterizing} but we use the former. In this section we derive the scaling behaviour of the variance of the cost function for linear optics, where now we expect a dependence on the number of modes $N$, the number of particles $n$, and the statistics of the particles.

From the definition of the linear Fock space unitaries $U$ in Eq.~\eqref{eq:linear_unitary}, we see that they form a faithful representation of $\mathrm{U}(N)$ on the Fock space: by interpreting the Fock space unitary as a function of $u_{ij}$ as $U = U(u)$, we find that $U(u_1)U(u_2) = U(u_1 u_2)$. As linear unitaries conserve particle number, each $n$-particle subspace is an invariant subspace under the action of $U$ and these subspaces additionally form irreducible representations (irreps), see Sec.~\ref{app:barren_plateaus} of the appendix for the proof. The Fock space decomposes as
\begin{equation}
\mathcal{F} = \bigoplus_{n} \mathcal{H}_n,
\end{equation}
where $\mathcal{H}_n$ is the $n$-particle subspace forming the irrep. If we let $V = \mathbb{C}^N$ be the vector space of a single particle in an $N$-mode interferometer, then these subspaces depend upon the statistics of the particles and are given by
\begin{equation}
\mathcal{H}_n = \begin{cases}
\mathrm{Sym}^n(V) & \text{bosons} \\
\wedge^n(V) & \text{fermions}
\end{cases}, \label{eq:irrep}
\end{equation}
where ($\wedge^n$) $\mathrm{Sym}^n$ is the (anti-)symmetric $n$-fold tensor product space. Therefore, a linear optical interferometer with a fixed number of particles hosts an irreducible representation of $\mathrm{U}(N)$. 

To calculate the variance of the cost function we typically assume that circuit unitaries form a $t$-design, meaning that moments of the cost function with respect to the distribution of circuit parameters is equal to moments with respect to the uniform Haar measure distribution on the Lie group of the circuit~\cite{mcclean2018barren,ragone2024lie,larocca2024review}. For linear optics, we assume that we have a $t$-design with respect to \textit{linear} unitaries of the form of Eq.~\eqref{eq:linear_unitary} only, and not all possible Fock space unitaries which means that for any operator $T$
\begin{equation}
\int d \boldsymbol{\theta} p(\boldsymbol{\theta}) U^{\otimes t}(\boldsymbol{\theta}) T U^{\dagger \otimes t}(\boldsymbol{\theta})    = \int_{\mathrm{U}(N)} du U^{\otimes t}(u) T U^{\dagger \otimes t}(u),
\end{equation}
where $p(\boldsymbol{\theta})$ is the distribution of parameters we sample from and $du$ is the Haar measure. This is a realistic assumption to make as universal interferometers are widely available and can encode any $U = U(u)$, so we can sample $u \in \mathrm{U}(N)$ with respect to the Haar measure. The result is the following theorem:
\begin{theorem}
\label{thm:barren_plateaus}
Consider an $N$ mode linear optical interferometer containing $n$ particles of either bosonic or fermionic statistics forming the $\mathrm{U}(N)$ irrep $\mathcal{H}_n$, c.f. Eq.~\eqref{eq:irrep}, whose tensor product space decomposes into irreps as
$ \mathcal{H}_n \otimes \mathcal{H}_n = \bigoplus_\alpha \mathcal{H}_\alpha $
indexed by $\alpha$, where each irrep has unit multiplicity with orthogonal projectors $P_\alpha$. Then given an input state $\rho$ and an observable $H$, the variance of the cost function of Eq.~\eqref{eq:general_cost_function} with respect to the Haar measure is given by
\begin{equation}
\mathrm{Var}_{\boldsymbol{\theta}}[ E(\boldsymbol{\theta})] 
 = \sum_{\alpha =0}^{n} \frac{1}{d_\alpha} \mathrm{Tr}(P_\alpha \rho^{\otimes 2}) \mathrm{Tr}(P_\alpha H^{\otimes 2}) - \frac{\mathrm{Tr}(H)^2}{D_n^2}, \label{eq:variance_thm}
\end{equation}
where $D_n = \mathrm{dim}(\mathcal{H}_n)$ and $d_\alpha = \mathrm{dim}(\mathcal{H}_\alpha)$.
\end{theorem}
\begin{proof}
Here we sketch the proof. The variance is given by $\mathrm{Var}_{\boldsymbol{\theta}}[E(\boldsymbol{\theta})] = \mathbb{E}_{\boldsymbol{\theta}} [E(\boldsymbol{\theta})^2] - \mathbb{E}_{\boldsymbol{\theta}} [E(\boldsymbol{\theta})]^2$.
Using the $t$-design assumption, the first moment can be written as an integral over $\mathrm{U}(N)$ as 
\begin{equation}
\mathbb{E}_{\boldsymbol{\theta}}[E(\boldsymbol{\theta})] =  \int_{\mathrm{U}(N)} du \mathrm{Tr} \left[  U(u) \rho U^\dagger(u)  H \right]  = \frac{1}{D_n} \mathrm{Tr}(H).
\end{equation}
The second moment is 
\begin{equation}
\begin{aligned}
\mathbb{E}_{\boldsymbol{\theta}}[E^2(\boldsymbol{\theta})]
& = \int_{\mathrm{U}(N)} d u  \mathrm{Tr}\left[U^{\otimes 2}(u) \rho^{\otimes 2} U^{\dagger \otimes 2} (u) H^{\otimes 2} \right] \\
& = \sum_{\alpha = 0}^n \frac{1}{d_\alpha} \mathrm{Tr}(P_\alpha \rho^{\otimes 2}) \mathrm{Tr}( P_\alpha H^{\otimes 2})
\end{aligned}
\end{equation}
which is a sum over irreps $\mathcal{H}_\alpha$, where we used the fact that $E^2(\boldsymbol{\theta}) = (\mathrm{Tr}[\rho(\boldsymbol{\theta}) H])^2 = \mathrm{Tr}[\rho^{\otimes 2}(\boldsymbol{\theta}) H^{\otimes 2}]$ for $\rho(\boldsymbol{\theta}) = U(\boldsymbol{\theta}) \rho U^\dagger(\boldsymbol{\theta})$ and the $t$-design assumption to bring it to an integral over the tensor product representation $U^{\otimes 2}(u)$. For both moments, the integrals are evaluated using Schur othogonality of the matrix elements of the irreps. The full detailed proof can be found in the Sec.~\ref{app:barren_plateaus} of the appendix. 
\end{proof}

An alternative expression expression for the moments can be calculated by using the adjoint representation of operators instead, as done in Ref.~\cite{mhiri2026bosonsamplingdiluteregime}. 

In practice, the input states to a linear optical interferometer are pure number states prepared by a single photon source with a demultiplexer to route the photons into the same time bin, so we take $\rho = |\psi_\text{in}\rangle \langle \psi_\text{in}|$ to be a number state. Moreover, the measurement observables accessible to an experiment are projective measurements onto number states, so we choose $H$ to be a projector onto a number state. Observables that are diagonal in the number basis, such as the QUBO observable of Eq.~\eqref{eq:qubit_hamiltonian}, can be written as sums of terms of this form.

With this information we can now evaluate the variances. Evaluating the traces in Eq.~\eqref{eq:variance_thm} is challenging in practice as we need to understand how the tensor product space $\mathcal{H}_n^{\otimes 2}$ decomposes into irreps and calculate projections, so as a first step we can construct a loose upper bound by noting that the operators $\rho$ and $H$ are both rank-1 projectors, so $\mathrm{Tr}(P_\alpha \rho^{\otimes 2}) \leq 1$ and $\mathrm{Tr}(P_\alpha H^{\otimes 2}) \leq 1$. This gives us
\begin{equation}
\mathrm{Var}_{\boldsymbol{\theta}} [E(\boldsymbol{\theta})] < \sum_{\alpha = 0}^n \frac{1}{d_\alpha}, \label{eq:bosonic_upper_bound}
\end{equation}
providing enough information to establish the existence of barren plateaus alone as this can decay exponentially, as seen in Fig.~\ref{fig:barren_plateau_scaling}.

Under certain circumstances we can improve upon this. Let us take the observable to be the projector onto the input state as $H = |\psi_\text{in}\rangle \langle \psi_\text{in}|$. The variance is invariant under $H \rightarrow U(u) H U^\dagger(u)$ for all $u \in \mathrm{U}(N)$, hence projecting onto any output state differing to the input state up to a linear unitary, $U(u)|\psi_\text{in}\rangle$, will yield the same variance, which includes all permutations. This means that projecting onto any fermionic number state will yield the same variance, as all $n$-fermion number states differ by a permutation; projecting onto bosonic states will depend on the state as not all bosonic number states can be connected with a linear unitary.

One experimentally-relevant input state is the $n$-particle collision-free state $|\psi_\text{in}\rangle = |1\rangle^{ n} |0\rangle^{N-n}$ consisting of one particle in each of the first $n$ modes. If the particles are bosonic, the variance is given by
\begin{equation}
\begin{aligned}
& \mathrm{Var}_{\boldsymbol{\theta}}[E(\boldsymbol{\theta})] = \\
&  \sum_{\alpha = 0}^{\lfloor n/2 \rfloor} \frac{4^{n - 2\alpha}}{d_{2\alpha}}   \left(\frac{2n - 4\alpha + 1}{2n - 2\alpha + 1}  \frac{ { n \choose \alpha}}{ { 2n - 2\alpha \choose n -\alpha}}\right)^2 - \frac{1}{D_n^2},
\end{aligned} \label{eq:var_1}
\end{equation}
whereas if the particles are fermionic the expression is much simpler and is given by
\begin{equation}
\mathrm{Var}_{\boldsymbol{\theta}}[E(\boldsymbol{\theta})] = \frac{1}{D_n^2} \frac{n(N-n)}{N+1}, \label{eq:var_2}
\end{equation}
which holds for projection onto any number state.

Another example of interest is the $n$-boson input state consisting of $n$ particles on the first mode $|\psi_\text{in}\rangle = |n\rangle |0\rangle^{N-1}$. In this case, the variance is given by
\begin{equation}
\mathrm{Var}_{\boldsymbol{\theta}} [E(\boldsymbol{\theta})] =
\frac{1}{D_{2n}} - \frac{1}{D_n^2}. \label{eq:var_3}
\end{equation}
There is no fermionic analogue of this due to the Pauli exclusion principle. 

In these formulae, the Hilbert space dimensions $D_n = \mathrm{dim}(\mathcal{H}_n)$ are 
\begin{equation}
D_n  = \begin{cases} 
{ N + n - 1 \choose n} & \text{bosons} \\  
{ N \choose n } & \text{fermions}
\end{cases}
\end{equation}
and the dimension of the irreps $d_\alpha = \mathrm{dim}(\mathcal{H}_\alpha)$ are given by
\begin{equation}
d_\alpha   = 
\begin{cases}
{ N + 2n - \alpha - 1 \choose 2n - \alpha} {N + \alpha - 2 \choose \alpha} \frac{2n - 2\alpha + 1}{2n - \alpha + 1} & \text{bosons} \\
{N \choose  2n - \alpha} {N+1 \choose \alpha} \frac{2n - 2\alpha + 1}{2n - \alpha + 1} & \text{fermions}
\end{cases}
\end{equation}
which are derived using Young's tableaux. See Sec.~\ref{app:barren_plateaus} of the appendix for detailed proofs.

In Fig.~\ref{fig:barren_plateau_scaling} we test these cases numerically using the Lightworks python package~\cite{lightworks} using exact wavefunction simulation and see a good agreement with the analytics. If we fix $n$ and scale $N$ only, then the variance decays polynomially regardless of the statistics implying there is no barren plateau here. If we fix $N$ and scale $n$ only then the decay is faster, with fermions reaching minimum variance at $n = N/2$ before increasing again. If we scale both $N$ and $n$ simultaneously (as this is typically what we mean when we scale up linear optics) with the relationship $N = 2n$  then the decay is exponential for both statistics, implying the existence of a barren plateau. Despite this, bosonic statistics decays orders of magnitude faster than fermionic statistics, especially for the collision-free input state which is the input state in real experiments. This implies that fermionic linear optics is less susceptible to the barren plateau problem for experimentally-relevant cases. In all three plots, we give the bosonic upper bound of Eq.~\eqref{eq:bosonic_upper_bound} which applies to projections onto all possible output states; there is no upper bound shown for fermionic statistics as Eq.~\eqref{eq:var_2} holds for all outputs.

\begin{figure}[t]
\begin{center}
\includegraphics[width = \columnwidth]{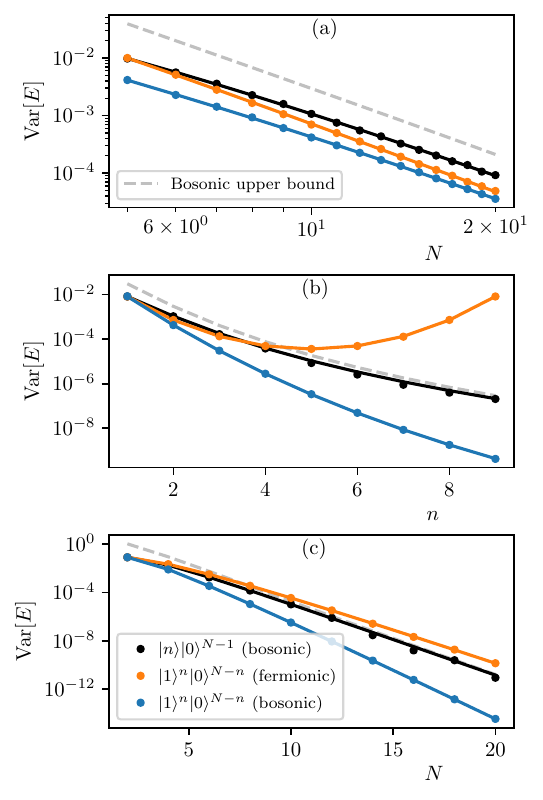}
\caption{The variance of the cost function for an $N$-mode interferometer containing $n$ particles of either bosonic or fermionic statistics, where linear unitaries are sampled with respect to the Haar measure for $3 \times 10^4$ repeats. In each subplot, markers correspond to numeric data, whilst solid lines correspond to analytic formulae of Eqs.~\eqref{eq:var_1}-\eqref{eq:var_3}. The dashed line corresponds to the upper bound of Eq.~\eqref{eq:bosonic_upper_bound}. (a) We fix $n = 2$ and vary $N$. We see a polynomial decay in $N$ as this is a log-log plot, hence scaling $N$ alone does not result in a barren plateau. (b) We fix $N = 10$ as we vary $n$. We see a decay slower than, but approaching, exponential decay. In particular, when $n > N/2$ the fermionic variance increases as the Hilbert space dimension begins to decrease beyond this point. (c) We let $N = 2n$ and vary $N$. We see exponential decay, implying that we have a barren plateau here. Despite this we see that the collision-free input state for bosons, which is the most relevant to a real experiment, decays orders of magnitudes faster than the same state with fermionic statistics.}
\label{fig:barren_plateau_scaling}
\end{center}
\end{figure}

\subsection{Bosonic cost landscapes}

The previous section has demonstrated that if our particles have fermionic statistics, the barren plateau scaling can be orders of magnitude slower than the same hardware with bosonic statistics. In order to further investigate the difference between bosonic and fermionic linear optics, let us now study the functional form of the cost function of Eq.~\eqref{eq:general_cost_function} for an arbitrary observable $H$. If we insert $n$ bosons into an $N$-mode interferometer and vary a single phase shifter, say the $\mu$th one, whilst keeping the rest fixed, then the cost function is given by the real trigonometric polynomial
\begin{equation}
f(x) := E(\theta_\mu = x) = \sum_{k = -n}^n c_k e^{ikx}, \label{eq:bosonic_cost_landscape}
\end{equation}
where $c_k$ are coefficients that depend upon details of the rest of the interferometer and $H$. See Sec.~\ref{app:cost_landscape} of the appendix or Ref.~\cite{gan2022fock} for the proof. We refer to this as the \textit{bosonic cost landscape}.

Using trigonometric interpolation, the gradient of an unknown cost function with respect to a given parameter can be evaluated exactly given a set of $2n$ samples of the cost function as
\begin{equation}
f'(x)   = \sum_{k = 1}^{2n} f(x + x_k) \frac{(-1)^{k + 1}}{4n \sin^2 ( x_k/2)} , \label{eq:parameter_shift_rule}
\end{equation}
where $x_k = \frac{(2k-1)\pi}{2n}$ are the equally-spaced sample parameters~\cite{rotosolve3,facelli2024exactgradientslinearoptics,atkinson1991introduction}. This is a generalised parameter shift rule for cost functions with $n$ harmonics that can be used for gradient descent. As $n$ is the number of particles in the interferometer, this scales poorly in general which will impact the speed of the optimiser. Moreover, due to how expressive Fourier series is, with larger $n$ the cost function of Eq.~\eqref{eq:bosonic_cost_landscape} can result in extremely complicated functions, including functions with approximate discontinuities and barren plateaus, and it contains up to $n$ local minima per parameter. All of these issues will impede a gradient-based optimiser and in this paper we ask how one can mitigate them.

To obtain a cost landscape of the form of Eq.~\eqref{eq:bosonic_cost_landscape} we require the parametrised unitaries to be phase shifters. In other variational algorithms, such as QAOA~\cite{farhi2014quantumapproximateoptimizationalgorithm} or the variational quantum eigensolver~\cite{peruzzo2014variational}, the parametrised unitaries may not be phase shifters in general. However, if the parametrised unitaries are generated by quadratric Hamiltonians then one can decompose them exactly into a product of parametrised phase shifters and fixed 50:50 beamsplitters using known algorithms~\cite{PhysRevLett.73.58,Clements:16,Fldzhyan:20} which we can then vary independently.

\subsection{Fermionic cost landscapes and Rotosolve \label{sec:fermionic_cost_landscapes}}

The large number of harmonics in the bosonic cost landscape of Eq.~\eqref{eq:bosonic_cost_landscape} can be traced back to the unitary implementing the parametrised phase shifter, given by $U_\text{PS}(x) = \exp(i \hat{n} x)$, where $\hat{n}$ is the number operator for the mode it acts upon and $x \in [0,2\pi)$ is the phase.  Given an $n$-particle state $|n\rangle$, this acts as
\begin{equation}
U_\text{PS}(x) |n\rangle  =  e^{i  n x }|n\rangle, \label{eq:standard_phase_shifters}
\end{equation}
where $n \in \mathbb{N} $ are the eigenvalues of $\hat{n}$. These complex phases give rise to the harmonics of the cost function of Eq.~\eqref{eq:bosonic_cost_landscape} with the set frequencies equal to the set of differences of eigenvalues of $\hat{n}$, see Sec.~\ref{app:cost_landscape} of the appendix. 

Suppose that we performed the same experiment with fermions instead, then the number operator has only two eigenvalues of $ n \in \{0, 1\}$ due to the Pauli exclusion principle and the cost landscape of Eq.~\eqref{eq:bosonic_cost_landscape} reduces to the simple form
\begin{equation}
f(x) = A \sin ( x - \phi) + B, \label{eq:fermionic_cost}
\end{equation}
where $A, B, \phi$ are constants depending upon details of the rest of the interferometer and the observable, see Sec.~\ref{app:cost_landscape} of the appendix. This holds regardless of the observable $H$, the input state, the number of particles or the size and layout of the interferometer. We refer to this as a \textit{fermionic cost landscape}. See Fig.~\ref{fig:rotosolve_minimum} for an example produced using an exact wavefunction simulation. This shows that bosonic cost landscapes are needlessly complicated.

We can immediately apply the parameter shift rule of Eq.~\eqref{eq:parameter_shift_rule} to this simpler cost landscape to get
\begin{equation}
f'(x) = \frac{1}{2} \left[ f\left(x + \frac{\pi}{2} \right) - f \left(x- \frac{\pi}{2}  \right) \right], \label{eq:fermionic_parameter_shift}
\end{equation}
which is the original parameter shift rule of Ref.~\cite{PhysRevA.98.032309} first applied to qubit-based systems. This means only two evaluations of the cost function are required for each component of the gradient as it no longer scales with the number of particles inserted, $n$, which is a considerable improvement over the requirements for the gradient of Eq.~\eqref{eq:parameter_shift_rule}. 

\begin{figure}
\begin{algorithm}[H]
  \caption{Rotosolve \cite{rotosolve1,rotosolve2,rotosolve3,rotosolve4}}
  \label{alg:rotosolve}
   \begin{algorithmic}
   \Require Cost function $E(\boldsymbol{\theta}) = \langle \psi_\text{out}(\boldsymbol{\theta})|H |\psi_\text{out}(\boldsymbol{\theta}) \rangle$
   \While{termination criteria not met}
   \For{$\mu= 1,2,\ldots,N_\text{p}$}
   		\State Let $f(x) = E(\theta_\mu = x)$
   		\State Estimate $f(0)$, $f(\pi/2)$ and $f(-\pi/2)$
   		\State $X = f(\pi/2) - f(-\pi/2)$
   		\State $Y = 2f(0) - f(\pi/2) - f(-\pi/2)$
   		\State $\theta_\mu \to -\pi/2 - \operatorname{atan2}(Y,X)$
   \EndFor
   \EndWhile
   \end{algorithmic}
\end{algorithm}
\includegraphics[width=\columnwidth]{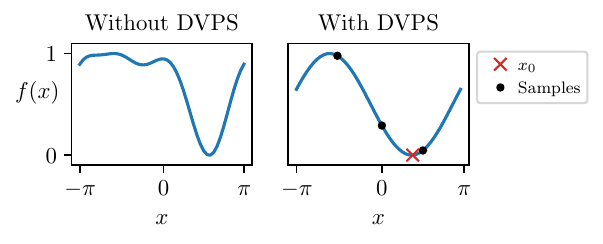}
\caption{The cost function for $n = 7$ particles with a random $H$.~Fermionic statistics simplifies the cost landscape slice for this example. The minimum $x_0$ can be found by sampling the cost function at $x=0,\pm\pi/2$ alone. Rotosolve exploits this by applying this to each variable iteratively. \label{fig:rotosolve_minimum}}
\end{figure}

On the other hand, the fact the fermionic cost landscape contains only one harmonic means we can avoid using gradient-based optimisation entirely. Due to the fact that there is a single minimum as one parameter is varied, we can solve for it using trigonometric interpolation. The minimum of $f(x)$ is given by $x_0$, where
\begin{subequations}
\label{eq:rotosolve_iteration}
\begin{align}
x_0 & = - \pi/2 - \operatorname{atan2}(Y,X), \\
X & = f(\pi/2) - f(-\pi/2),  \\
Y & = 2f(0) - f(\pi/2) - f(-\pi/2), 
\end{align}
\end{subequations}
as shown in Fig.~\ref{fig:rotosolve_minimum}, where $\operatorname{atan2} \in [-\pi,\pi]$ is the two-argument arctangent. This forms the basis of the Rotosolve algorithm of Refs.~\cite{rotosolve1,rotosolve2,rotosolve3,rotosolve4,rotosolve5}, whereby each phase shifter is iteratively optimised whilst holding the rest fixed, as shown in Algorithm~\ref{alg:rotosolve}. This significantly reduces the number of cost function samples, and hence shots, required to optimise. If the interferometer has $N_\text{PS}$ phase shifters and $n$ photons, then the number of cost evaluations per iteration of gradient descent for bosonic and fermionic cost landscapes using the parameter shift rule is $2nN_\text{PS}$ and $2N_\text{PS}$ respectively, and for Rotosolve it is $3$. 

In Fig.~\ref{fig:barren_plateau_scaling} of the previous section, we demonstrated that both bosonic and fermionic linear optics suffer from the barren plateau problem. This means that gradient-based optimisers will suffer as the problem scales up. Whilst Fig.~\ref{fig:barren_plateau_scaling} demonstrates that the barren plateau scaling of fermionic linear optics is better than its bosonic counterpart, the improvement is modest and it still exhibits the exponential decay. However, fermionic linear optics allows us to bypass gradient-based optimisers entirely and gives us a considerable speedup due to the simpler analytic form of the cost landscape. In Sec.~\ref{sec:results} we present results comparing the performance of the various methods and it is seen that Rotosolve significantly outperforms gradient descent for all tests we performed. Before presenting the results, in the next section we discuss three possible ways we can generate fermionic cost landscapes using bosons.

\section{The dual-valued phase shifter \label{sec:DVPS}} 

\subsection{Motivation}

Bosonic linear optics is far more experimentally accessible than fermionic linear optics, with large on-chip interferometers and multi-photon inputs available today. For this reason, we now ask whether it is possible to achieve a fermionic cost landscape of Eq.~\eqref{eq:fermionic_cost} with photonics. The key feature of a fermionic cost landscape is that it contains a single unit-frequency harmonic at most per parameter, which is due to the fact the number operator generating the parametrised phase shifters has only two distinct eigenvalues. For this reason, we replace the number operator with an operator $\hat{q}$ diagonal in the number basis with only two distinct eigenvalues $a,b \in \mathbb{R}$. This defines a \textit{dual-valued phase shifter} (DVPS) as $U_\text{DVPS}(x) := \exp(i \hat{q} x)$ whose action is given by
\begin{equation}
U_\text{DVPS}(x)|n\rangle  = e^{i  q(n) x }|n\rangle, \label{eq:dvps}
\end{equation}
where $q : \mathbb{N} \to \{ a , b \}$. Unlike a standard phase shifter, the DVPS is a non-linear device and cannot be constructed deterministically with linear optics. For this reason, the boson sampler using these is a non-linear boson sampler~\cite{spagnolo2023non}. It is simple to show that this will result in a fermionic cost landscape with a frequency of $\omega = |a-b|$. We now discuss three potential ways to realise this with experiments.

\subsection{Deterministic design with non-linear optics \label{sec:deterministic_DVPS}}

\begin{figure}[t]
\includegraphics[width = \columnwidth]{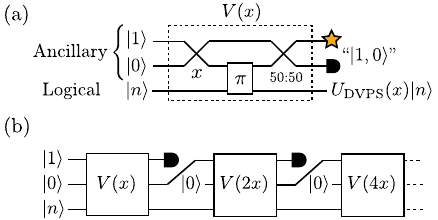}
\caption{(a) A non-deterministic DVPS consists of one logical mode and two ancillary modes. The logical mode interacts with the lower ancillary mode via a cross-Kerr non-linearity of $\pi$, shown by the square. The ancillary modes interact amongst themselves with a phase shifter of phase $x/2$, a tuneable beamsplitter of angle $x \in [0,2\pi)$, and a second 50:50 beamsplitter.  The success of the gate is heralded by the ancillary output state of $|1,0\rangle$. This gate has a success rate of $1/2$. (b) If the non-deterministic gate fails, as heralded by no photon in the upper ancillary mode, then the ancillary photon exits the lower ancillary mode and is rerouted into the input of a second iteration of the non-deterministic gate with double the parameter. We repeat until success. \label{fig:dvps}}
\end{figure}

One choice for the generator of the DVPS in Eq.~\eqref{eq:dvps} is given by $\hat{q} = \frac{1}{2}(1 - \hat{\pi})$, where $\hat{\pi} = \exp(i \pi \hat{n})$ is the number parity operator which is equivalent to a $\pi$-phase shifter. This operator has two distinct eigenvalues of $q(n) \in \{0,1\}$, where
\begin{equation}
q(n) = \frac{1}{2} ( 1 - (-1)^n). \label{eq:dvps_phase}
\end{equation}
This will yield the fermionic cost landscape of Eq.~\eqref{eq:fermionic_cost} with a unit frequency. 

Implementing this directly with optical elements would be impossible in practice as this unitary requires tuneable non-linear interactions. This is because the generator $\hat{q}$, which plays the role of the Hamiltonian of this device, is expanded out explicitly as
\begin{equation}
\hat{q} = \frac{1}{2}\sum_{m = 1}^\infty \frac{(-1)^{m+1} (\pi \hat{n})^{2m}}{(2m)!}, 
\end{equation}
which contains high-order interaction terms.

To circumvent this we off-load the tuneable parts of the device to linear components that we can control. We construct the circuit as shown in Fig.~\ref{fig:dvps}(a), consisting of three modes: one logical mode carrying the input and output state, and two ancillary modes. We have a single phase shifter of phase $x/2$, one tuneable beamsplitter with parameter $x$ described by the unitary
\begin{equation}
u(x) = \begin{pmatrix}
\cos \frac{x}{2} & -i \sin \frac{x}{2} \\
-i \sin \frac{x}{2} & \cos \frac{x}{2}
\end{pmatrix},
\end{equation}
and one 50:50 beamsplitter using the Hadamard convention. We also have a cross Kerr non-linearity between the second and third modes as $U_\text{K} = \exp \left( i \pi \hat{n}_2 \hat{n}_3 \right)$ with a magnitude of $\pi$. If we take the input states of the logical mode and ancillary modes to be $|\psi_\text{in}\rangle = |n\rangle$ and $|1,0\rangle$ respectively, and perform a projective measurement on the ancillary modes, then the output of the logical mode is given by
\begin{equation}
|\psi_\text{out}\rangle = \begin{cases}
 e^{iq(n)x}|n\rangle & \text{Ancillary $= | 1,0\rangle$} \\
 e^{-iq(n)x}|n\rangle & \text{Ancillary $=|0,1\rangle$} \label{eq:non_deterministic_dvps}
\end{cases},
\end{equation}
where each output has a probability of $1/2$, see Sec.~\ref{app:non_linear_DVPS} of the appendix. If we measure $|1,0\rangle$ then the gate is a success as it implements the desired phase, therefore this gate is currently non-deterministic with a probability of success of $1/2$.

This non-deterministic gate can be upgraded to an asymptotically deterministic gate by repeating until success. From Eq.~\eqref{eq:non_deterministic_dvps}, we see that if the ancillary photon is measured in the state $|0,1\rangle$ then the gate has failed and the phase applied to the the logical mode is the complex conjugate of what we want. We correct this by using the martingale strategy\footnote{The martingale strategy is a gambling strategy for betting on a game with two outcomes, where the player doubles their bet after each loss until an eventual win recovers all previous losses. For example when betting on red or black in roulette. Overall profit is guaranteed only if the player has infinite wealth and there is no upper limit to the allowed bets (and if the casino doesn't ask you to leave!)} by feeding the failed output into the input of a second DVPS whose phase is $2x$ instead. If the second attempt succeeds then it will correct the incorrect phase of the previous failed attempt to give us the correct output. In general, if we repeat this process $m$ times the phase of the $m$th DVPS is $x_m = 2^{m-1}x \mod 2\pi$. The probability of success after at most $m$ attempts is given by $1 - 1/2^m$ which is asymptotically deterministic. For example, $m = 7$ repeats gives a success probability of 99.2\%. 

In Fig.~\ref{fig:dvps}(b) we provide a photonic circuit that performs this feedforward process automatically without need of a classical computer: if the gate is a success the ancillary photon is consumed and the remaining gates reduce to the identity as the Kerr interaction does not activate; whereas if the gate fails the ancillary photon is automatically routed into the next attempt.

\subsection{Non-deterministic design with linear optics \label{sec:non_deterministic_DVPS}}

The design introduced in the previous section requires strong non-linearities rendering it infeasible with today's technology, which motivates us to search for a realisation with linear optics alone. By using measurement-induced non-linearities~\cite{PhysRevA.68.032310,PhysRevLett.95.040502,Scheel_2004,Scheel_2005,Scheel2006} we can achieve such at the expense of it becoming non-deterministic. 

To implement a DVPS acting on the subspace of at most $N$ photons, we introduce $N+1$ modes where the first mode is the logical mode and the remaining $N$ modes are ancillary modes. We interfere these modes with an $(N+1)$-mode linear optical interferometer encoding a unitary $u \in \mathrm{U}(N+1)$. Let us prepare the logical mode in the most general state
\begin{equation}
|\psi_\text{in}\rangle = \sum_{n = 0 }^N c_n |n\rangle,
\end{equation}
where $|n\rangle$ is the state of $n$ photons, and let us prepare the $N$ ancillary modes in the number state $|\mathbf{a}\rangle = |a_1,a_2,\ldots,a_N\rangle$, where $a_i \in \{ 0, 1\}$. If we insert the total state $|\psi_\text{in}\rangle |\mathbf{a}\rangle$ into the interferometer and project the output ancillary modes onto the state $|\mathbf{\mathbf{a}}\rangle$, as shown in Fig.~\ref{fig:non_deterministic}(a), then the (unnormalised) state of the logical mode is
\begin{equation}
|\psi_\text{out} \rangle  = \sum_{n = 0}^N \frac{c_n}{n!} \operatorname{per} (u[\mathbf{v}_n | \mathbf{v}_n]) |n \rangle , \label{eq:measurement_induced_output}
\end{equation}
where we used the transition amplitude of Eq.~\eqref{eq:transition_amplitude} and $\mathbf{v}_n = (n,a_1,\ldots,a_N)$. This output state will in general be equal to a non-linear transformation applied to the input state. The success probability of this transformation is $p =\langle \psi_\text{out}|\psi_\text{out}\rangle$.

\begin{figure}
\includegraphics[width=\columnwidth]{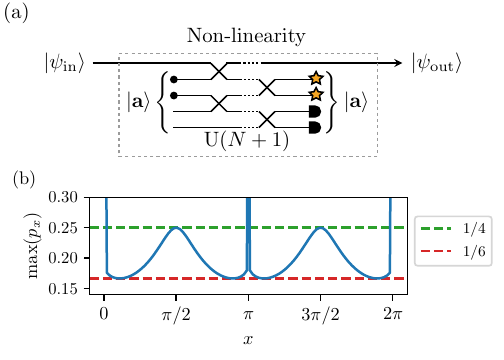}
\caption{(a) A measurement-induced non-linear mapping $|\psi_\text{in}\rangle \mapsto |\psi_\text{out}\rangle$ is obtained by interfering the input state with an ancillary state $|\mathbf{a}\rangle$ in an $(N+1)$-mode linear interferometer and postselecting on the ancillary output. (b) The maximum probability of the non-deterministic dual-valued phase shifter obtained by solving Eq.~\eqref{eq:non_linearity_condition} for the subspace of at most two photons. This uses the circuit of (a) with $N=  2$ ancillary modes and $|\mathbf{a}\rangle = |1,0\rangle$. For $x = 0,\pi$ and $2\pi$, $\mathrm{max}(p_x) = 1$. \label{fig:non_deterministic}}
\end{figure}

For this operation to give us the action of a DVPS, comparing Eqs.~\eqref{eq:dvps} and \eqref{eq:measurement_induced_output} implies that for each $x \in [0,2\pi)$ we must solve for the matrix $u_x \in \mathrm{U}(N+1)$ that obeys
\begin{equation}
\frac{1}{n!}\operatorname{per} (u_x[\mathbf{v}_n|\mathbf{v}_n]) =  \sqrt{p_x} e^{i \alpha_x}  e^{iq(n)x}, \label{eq:non_linearity_condition}
\end{equation}
for all $n = 0,1,\ldots,N$, where $p_x \in [0,1]$ is the success probability of this non-deterministic gate, $\alpha_x$ is a global phase, and $q(n)$ is chosen to be  Eq.~\eqref{eq:dvps_phase}. Solving this is a formidable task, so numerical methods can be employed such as given in Refs.~\cite{PhysRevLett.95.040502,wu2007optimizingopticalquantumlogic,sparrow2018simulating,PhysRevA.79.042326} and Sec.~\ref{app:numerical_methods} of the appendix.

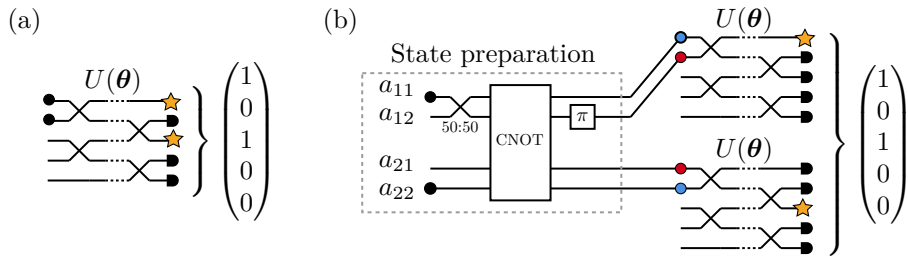
\begin{figure*}
\begin{center}
\begin{tikzpicture}[x=0.75pt,y=0.75pt,yscale=-1,xscale=1]

\draw  [line width=0.75]    (345,26) -- (320,56) ;
\draw  [line width=0.75]    (320,66) -- (345,36) ;
\draw [line width=0.75]    (345,36) -- (355,36) ;
\draw [line width=0.75]    (345,46) -- (355,46) ;
\draw [line width=0.75]    (345,56) -- (355,56) ;
\draw [line width=0.75]    (345,66) -- (365,66) ;
\draw [line width=0.75]    (345,26) -- (355,26) ;
\draw [line width=0.75]    (355,26) -- (365,36) ;
\draw [line width=0.75]    (355,36) -- (365,26) ;
\draw [line width=0.75]    (355,46) -- (365,56) ;
\draw [line width=0.75]    (355,56) -- (365,46) ;
\draw [line width=0.75]    (385,36) -- (395,46) ;
\draw [line width=0.75]    (385,46) -- (395,36) ;
\draw [line width=0.75]    (385,56) -- (395,66) ;
\draw [line width=0.75]    (385,66) -- (395,56) ;
\draw [line width=0.75]    (385,26) -- (405,26) ;
\draw [line width=0.75]    (395,36) -- (405,36) ;
\draw [line width=0.75]    (395,46) -- (405,46) ;
\draw [line width=0.75]    (395,56) -- (405,56) ;
\draw [line width=0.75]    (395,66) -- (405,66) ;
\draw  [fill={rgb, 255:red, 0; green, 0; blue, 0 }  ,fill opacity=1 ][line width=1.5]  (405.92,34) -- (407.92,34) .. controls (409.02,34) and (409.92,34.9) .. (409.92,36) .. controls (409.92,37.1) and (409.02,38) .. (407.92,38) -- (405.92,38) -- cycle ;
\draw  [fill={rgb, 255:red, 0; green, 0; blue, 0 }  ,fill opacity=1 ][line width=1.5]  (405.92,44) -- (407.92,44) .. controls (409.02,44) and (409.92,44.9) .. (409.92,46) .. controls (409.92,47.1) and (409.02,48) .. (407.92,48) -- (405.92,48) -- cycle ;
\draw  [fill={rgb, 255:red, 0; green, 0; blue, 0 }  ,fill opacity=1 ][line width=1.5]  (405.92,54) -- (407.92,54) .. controls (409.02,54) and (409.92,54.9) .. (409.92,56) .. controls (409.92,57.1) and (409.02,58) .. (407.92,58) -- (405.92,58) -- cycle ;
\draw  [fill={rgb, 255:red, 0; green, 0; blue, 0 }  ,fill opacity=1 ][line width=1.5]  (405.92,64) -- (407.92,64) .. controls (409.02,64) and (409.92,64.9) .. (409.92,66) .. controls (409.92,67.1) and (409.02,68) .. (407.92,68) -- (405.92,68) -- cycle ;
\draw [line width=0.75]    (280,102) -- (355,102) ;
\draw [line width=0.75]    (345,112) -- (355,112) ;
\draw [line width=0.75]    (345,122) -- (355,122) ;
\draw [line width=0.75]    (345,132) -- (365,132) ;
\draw [line width=0.75]    (280,92) -- (355,92) ;
\draw [line width=0.75]    (355,92) -- (365,102) ;
\draw [line width=0.75]    (355,102) -- (365,92) ;
\draw [line width=0.75]    (355,112) -- (365,122) ;
\draw [line width=0.75]    (355,122) -- (365,112) ;
\draw [line width=0.75]    (385,102) -- (395,112) ;
\draw [line width=0.75]    (385,112) -- (395,102) ;
\draw [line width=0.75]    (385,122) -- (395,132) ;
\draw [line width=0.75]    (385,132) -- (395,122) ;
\draw [line width=0.75]    (385,92) -- (405,92) ;
\draw [line width=0.75]    (395,102) -- (405,102) ;
\draw [line width=0.75]    (395,112) -- (405,112) ;
\draw [line width=0.75]    (395,122) -- (405,122) ;
\draw [line width=0.75]    (395,132) -- (405,132) ;
\draw  [fill={rgb, 255:red, 0; green, 0; blue, 0 }  ,fill opacity=1 ][line width=1.5]  (405.92,90) -- (407.92,90) .. controls (409.02,90) and (409.92,90.9) .. (409.92,92) .. controls (409.92,93.1) and (409.02,94) .. (407.92,94) -- (405.92,94) -- cycle ;
\draw  [fill={rgb, 255:red, 0; green, 0; blue, 0 }  ,fill opacity=1 ][line width=1.5]  (405.92,100) -- (407.92,100) .. controls (409.02,100) and (409.92,100.9) .. (409.92,102) .. controls (409.92,103.1) and (409.02,104) .. (407.92,104) -- (405.92,104) -- cycle ;
\draw  [fill={rgb, 255:red, 0; green, 0; blue, 0 }  ,fill opacity=1 ][line width=1.5]  (405.92,120) -- (407.92,120) .. controls (409.02,120) and (409.92,120.9) .. (409.92,122) .. controls (409.92,123.1) and (409.02,124) .. (407.92,124) -- (405.92,124) -- cycle ;
\draw  [fill={rgb, 255:red, 0; green, 0; blue, 0 }  ,fill opacity=1 ][line width=1.5]  (405.92,130) -- (407.92,130) .. controls (409.02,130) and (409.92,130.9) .. (409.92,132) .. controls (409.92,133.1) and (409.02,134) .. (407.92,134) -- (405.92,134) -- cycle ;
\draw  [fill={rgb, 255:red, 245; green, 166; blue, 35 }  ,fill opacity=1 ] (407,22) -- (408.45,24.79) -- (411.68,25.24) -- (409.34,27.42) -- (409.89,30.49) -- (407,29.04) -- (404.11,30.49) -- (404.66,27.42) -- (402.32,25.24) -- (405.55,24.79) -- cycle ;
\draw  [fill={rgb, 255:red, 245; green, 166; blue, 35 }  ,fill opacity=1 ] (405.92,107.31) -- (407.36,110.1) -- (410.59,110.55) -- (408.26,112.72) -- (408.81,115.8) -- (405.92,114.35) -- (403.03,115.8) -- (403.58,112.72) -- (401.24,110.55) -- (404.47,110.1) -- cycle ;
\draw  [line width=0.75]  (365,46) -- (373,46) ;
\draw  [line width=0.75]  (365,56) -- (373,56) ;
\draw  [line width=0.75][dash pattern={on 1pt off 1pt}]  (373,36) -- (385,36) ;
\draw [line width=0.75] [dash pattern={on 1pt off 1pt}]  (373,26) -- (385,26) ;
\draw  [line width=0.75]  (365,36) -- (373,36) ;
\draw[line width=0.75]    (365,26) -- (373,26) ;
\draw [line width=0.75] [dash pattern={on 1pt off 1pt}]  (373,46) -- (385,46) ;
\draw [line width=0.75] [dash pattern={on 1pt off 1pt}]  (373,56) -- (385,56) ;
\draw [line width=0.75] [dash pattern={on 1pt off 1pt}]  (373,66) -- (385,66) ;
\draw   [line width=0.75]   (365,66) -- (373,66) ;
\draw [line width=0.75]   (365,92) -- (373,92) ;
\draw [line width=0.75] [dash pattern={on 1pt off 1pt}]  (373,92) -- (385,92) ;
\draw [line width=0.75]   (365,102) -- (373,102) ;
\draw [line width=0.75] [dash pattern={on 1pt off 1pt}]  (373,102) -- (385,102) ;
\draw [line width=0.75] [dash pattern={on 1pt off 1pt}]  (373,112) -- (385,112) ;
\draw  [line width=0.75]  (365,112) -- (373,112) ;
\draw [line width=0.75]   (365,122) -- (373,122) ;
\draw [line width=0.75] [dash pattern={on 1pt off 1pt}]  (373,122) -- (385,122) ;
\draw  [line width=0.75]  (365,132) -- (373,132) ;
\draw [line width=0.75]   [dash pattern={on 1pt off 1pt}]  (373,132) -- (385,132) ;
\draw [line width=0.75]    (29,68) -- (39.74,68) ;
\draw [line width=0.75]    (29,78) -- (39.74,78) ;
\draw [line width=0.75]    (29,88) -- (39.74,88) ;
\draw [line width=0.75]    (29,98) -- (49.74,98) ;
\draw [line width=0.75]    (29,58) -- (39.74,58) ;
\draw [line width=0.75]    (39.74,58) -- (49.74,68) ;
\draw [line width=0.75]    (39.74,68) -- (49.74,58) ;
\draw [line width=0.75]    (39.74,78) -- (49.74,88) ;
\draw [line width=0.75]    (39.74,88) -- (49.74,78) ;
\draw [line width=0.75]    (69.74,68) -- (79.74,78) ;
\draw [line width=0.75]    (69.74,78) -- (79.74,68) ;
\draw [line width=0.75]    (69.74,88) -- (79.74,98) ;
\draw [line width=0.75]    (69.74,98) -- (79.74,88) ;
\draw [line width=0.75]    (69.74,58) -- (89,58) ;
\draw [line width=0.75]    (79.74,68) -- (89,68) ;
\draw [line width=0.75]    (79.74,78) -- (89,78) ;
\draw [line width=0.75]    (79.74,88) -- (89,88) ;
\draw [line width=0.75]    (79.74,98) -- (89,98) ;
\draw  [fill={rgb, 255:red, 0; green, 0; blue, 0 }  ,fill opacity=1 ][line width=1.5]  (89.41,66) -- (91.41,66) .. controls (92.52,66) and (93.41,66.9) .. (93.41,68) .. controls (93.41,69.1) and (92.52,70) .. (91.41,70) -- (89.41,70) -- cycle ;
\draw  [fill={rgb, 255:red, 0; green, 0; blue, 0 }  ,fill opacity=1 ][line width=1.5]  (89.41,86) -- (91.41,86) .. controls (92.52,86) and (93.41,86.9) .. (93.41,88) .. controls (93.41,89.1) and (92.52,90) .. (91.41,90) -- (89.41,90) -- cycle ;
\draw  [fill={rgb, 255:red, 0; green, 0; blue, 0 }  ,fill opacity=1 ][line width=1.5]  (89.41,96) -- (91.41,96) .. controls (92.52,96) and (93.41,96.9) .. (93.41,98) .. controls (93.41,99.1) and (92.52,100) .. (91.41,100) -- (89.41,100) -- cycle ;
\draw  [fill={rgb, 255:red, 245; green, 166; blue, 35 }  ,fill opacity=1 ] (90.49,54) -- (91.94,56.79) -- (95.17,57.24) -- (92.83,59.42) -- (93.38,62.49) -- (90.49,61.04) -- (87.6,62.49) -- (88.15,59.42) -- (85.82,57.24) -- (89.05,56.79) -- cycle ;
\draw  [fill={rgb, 255:red, 245; green, 166; blue, 35 }  ,fill opacity=1 ] (91.75,72.58) -- (93.19,75.38) -- (96.43,75.83) -- (94.09,78) -- (94.64,81.07) -- (91.75,79.62) -- (88.86,81.07) -- (89.41,78) -- (87.07,75.83) -- (90.3,75.38) -- cycle ;
\draw   [line width=0.75] (49.74,78) -- (57.74,78) ;
\draw  [line width=0.75]  (49.74,88) -- (57.74,88) ;
\draw  [line width=0.75] [dash pattern={on 1pt off 1pt}]  (57.74,68) -- (69.74,68) ;
\draw  [line width=0.75] [dash pattern={on 1pt off 1pt}]  (57.74,58) -- (69.74,58) ;
\draw  [line width=0.75]  (49.74,68) -- (57.74,68) ;
\draw [line width=0.75]   (49.74,58) -- (57.74,58) ;
\draw [line width=0.75] [dash pattern={on 1pt off 1pt}]  (57.74,78) -- (69.74,78) ;
\draw [line width=0.75]  [dash pattern={on 1pt off 1pt}]  (57.74,88) -- (69.74,88) ;
\draw [line width=0.75] [dash pattern={on 1pt off 1pt}]  (57.74,98) -- (69.74,98) ;
\draw  [line width=0.75]  (49.74,98) -- (57.74,98) ;
\draw [line width=0.75]   [fill={rgb, 255:red, 74; green, 144; blue, 226 }  ,fill opacity=1 ] (342.33,26) .. controls (342.33,24.53) and (343.53,23.33) .. (345,23.33) .. controls (346.47,23.33) and (347.67,24.53) .. (347.67,26) .. controls (347.67,27.47) and (346.47,28.67) .. (345,28.67) .. controls (343.53,28.67) and (342.33,27.47) .. (342.33,26) -- cycle ;
\draw  [fill={rgb, 255:red, 208; green, 2; blue, 27 }  ,fill opacity=1 ] (342.33,36) .. controls (342.33,34.53) and (343.53,33.33) .. (345,33.33) .. controls (346.47,33.33) and (347.67,34.53) .. (347.67,36) .. controls (347.67,37.47) and (346.47,38.67) .. (345,38.67) .. controls (343.53,38.67) and (342.33,37.47) .. (342.33,36) -- cycle ;
\draw  [color={rgb, 255:red, 0; green, 0; blue, 0 }  ,draw opacity=1 ][fill={rgb, 255:red, 208; green, 2; blue, 27 }  ,fill opacity=1 ] (342.33,92) .. controls (342.33,90.53) and (343.53,89.33) .. (345,89.33) .. controls (346.47,89.33) and (347.67,90.53) .. (347.67,92) .. controls (347.67,93.47) and (346.47,94.67) .. (345,94.67) .. controls (343.53,94.67) and (342.33,93.47) .. (342.33,92) -- cycle ;
\draw  [fill={rgb, 255:red, 74; green, 144; blue, 226 }  ,fill opacity=1 ] (342.33,102) .. controls (342.33,100.53) and (343.53,99.33) .. (345,99.33) .. controls (346.47,99.33) and (347.67,100.53) .. (347.67,102) .. controls (347.67,103.47) and (346.47,104.67) .. (345,104.67) .. controls (343.53,104.67) and (342.33,103.47) .. (342.33,102) -- cycle ;
\draw  [fill={rgb, 255:red, 0; green, 0; blue, 0 }  ,fill opacity=1 ] (27,57.33) .. controls (27,55.86) and (28.19,54.67) .. (29.67,54.67) .. controls (31.14,54.67) and (32.33,55.86) .. (32.33,57.33) .. controls (32.33,58.81) and (31.14,60) .. (29.67,60) .. controls (28.19,60) and (27,58.81) .. (27,57.33) -- cycle ;
\draw  [fill={rgb, 255:red, 0; green, 0; blue, 0 }  ,fill opacity=1 ] (27,67.33) .. controls (27,65.86) and (28.19,64.67) .. (29.67,64.67) .. controls (31.14,64.67) and (32.33,65.86) .. (32.33,67.33) .. controls (32.33,68.81) and (31.14,70) .. (29.67,70) .. controls (28.19,70) and (27,68.81) .. (27,67.33) -- cycle ;
\draw  [line width=0.75]   (250,50) -- (280,50) -- (280,108) -- (250,108) -- cycle ;
\draw [line width=0.75]    (280,66) -- (290,66) ;
\draw  [line width=0.75]    (280,56) -- (320,56) ;
\draw [line width=0.75]    (290,60) -- (302,60) -- (302,72) -- (290,72) -- cycle ;
\draw  [line width=0.75]  (302,66) -- (320,66) ;
\draw [line width=0.75]    (220,102) -- (250,102) ;
\draw [line width=0.75]    (220,92) -- (250,92) ;
\draw [line width=0.75]    (240,66) -- (250,66) ;
\draw [line width=0.75]    (220,66) -- (230,66) ;
\draw [line width=0.75]    (240,56) -- (250,56) ;
\draw [line width=0.75]    (220,56) -- (230,56) ;
\draw [line width=0.75]   (230,56) -- (240,66) ;
\draw [line width=0.75]   (230,66) -- (240,56) ;
\draw  [fill={rgb, 255:red, 0; green, 0; blue, 0 }  ,fill opacity=1 ] (217.33,56) .. controls (217.33,54.53) and (218.53,53.33) .. (220,53.33) .. controls (221.47,53.33) and (222.67,54.53) .. (222.67,56) .. controls (222.67,57.47) and (221.47,58.67) .. (220,58.67) .. controls (218.53,58.67) and (217.33,57.47) .. (217.33,56) -- cycle ;
\draw [line width=0.75]   [color={rgb, 255:red, 0; green, 0; blue, 0 }  ,draw opacity=1 ][fill={rgb, 255:red, 0; green, 0; blue, 0 }  ,fill opacity=1 ] (217.33,102) .. controls (217.33,100.53) and (218.53,99.33) .. (220,99.33) .. controls (221.47,99.33) and (222.67,100.53) .. (222.67,102) .. controls (222.67,103.47) and (221.47,104.67) .. (220,104.67) .. controls (218.53,104.67) and (217.33,103.47) .. (217.33,102) -- cycle ;
\draw [line width=0.75]   [color={rgb, 255:red, 155; green, 155; blue, 155 }  ,draw opacity=1 ][dash pattern={on 1.5pt off 1.5pt}] (186,44) -- (315.33,44) -- (315.33,114) -- (186,114) -- cycle ;

\draw (361,9.4) node [anchor=north west][inner sep=0.75pt]    {$U(\boldsymbol{\theta } )$};
\draw (361,74.9) node [anchor=north west][inner sep=0.75pt]    {$U(\boldsymbol{\theta } )$};
\draw (46,40.4) node [anchor=north west][inner sep=0.75pt]    {$U(\boldsymbol{\theta } )$};
\draw (7.6,10) node [anchor=north west][inner sep=0.75pt]   [align=left] {(a)};
\draw (165,10) node [anchor=north west][inner sep=0.75pt]   [align=left] {(b)};
\draw (83,40) node [anchor=north west][inner sep=0.75pt]   {$\begin{drcases}
 & \\
 & \\
 & 
\end{drcases}\begin{pmatrix}
1\\
0\\
1\\
0\\
0
\end{pmatrix}$};
\draw (401,18) node [anchor=north west][inner sep=0.75pt]   {$\begin{drcases}
 & \\
 & \\
 & \\
 & \\
 & \\
 & \\
 &
\end{drcases}\begin{pmatrix}
1\\
0\\
1\\
0\\
0
\end{pmatrix}$};
\draw (252,74) node [anchor=north west][inner sep=0.75pt]  [scale=0.6] [align=left] {CNOT};
\draw (292,62.75) node [anchor=north west][inner sep=0.75pt] [scale=0.75]   {$\pi $};
\draw (199,28) node [anchor=north west][inner sep=0.75pt]   [align=left] {State preparation};
\draw (225,68) node [anchor=north west][inner sep=0.75pt]  [scale=0.6] [align=left] {50:50};
\draw (193,46.4) node [anchor=north west][inner sep=0.75pt]    {$a_{11}$};
\draw (193,84.4) node [anchor=north west][inner sep=0.75pt]    {$a_{21}$};
\draw (193,60.4) node [anchor=north west][inner sep=0.75pt]    {$a_{12}$};
\draw (193,98.4) node [anchor=north west][inner sep=0.75pt]    {$a_{22}$};

\end{tikzpicture}
\end{center}
\caption{(a) A fermion sampling experiment with two input fermions in the state $ |\psi_\text{F}\rangle = a^\dagger_1 a^\dagger_2 |0\rangle$, represented by the solid circles, passed into a linear interferometer encoding some unitary. The clicks of the detectors map to a bit string. (b) The photonic simulation of this consists of a state preparation that transforms the state $a_{11}^\dagger a_{22}^\dagger |0\rangle$, shown by the solid circles, into the entangled state $|\psi_\text{B}\rangle = \frac{1}{\sqrt{2}}(a^\dagger_{11} a_{22}^\dagger - a_{12}^\dagger a_{21}^\dagger)|0\rangle$, where each pair of like-coloured circles represents a term in the superposition. The CNOT gate is constructed from linear optics using the unheralded design of Ref.~\cite{PhysRevA.65.062324} which has a success probability of $1/9$. This is then inserted into a pair of identical interferometers encoding the same unitary. The superimposed detector click distribution of both interferometers gives rise to the output bit string.  \label{fig:fermion_sampler}}
\end{figure*}

As an example let us encode a non-deterministic DVPS on the subspace of at most $N = 2$ photons. Following the recipe above, for each $x$ this requires us to introduce a three-mode linear interferometer encoding a unitary $u_x \in \mathrm{U}(3)$ that solves Eq.~\eqref{eq:non_linearity_condition} and postselect on the outputs of the ancillary modes. In Ref.~\cite{PhysRevA.68.032310} it was shown that there are infinitely many unitary solutions to this and the goal is to find the one maximising $p_x$. For this particular subspace, this is bounded from above as $p_x \leq 1/4$~\cite{PhysRevLett.95.040502,Scheel_2005}. Note that on this subspace this gate is closely related to the non-linear sign (NS) gate~\cite{PhysRevA.65.012314,KLM}, a well-known non-linear phase shifter used for constructing CNOT gates, defined via
\begin{equation}
U_\text{NS} : \alpha |0 \rangle + \beta |1\rangle + \gamma |2\rangle \mapsto \alpha |0 \rangle + \beta |1\rangle - \gamma |2\rangle.
\end{equation} 
The relationship between the NS gate and the DVPS is given by $U_\text{NS} = U_\text{PS}(\pi/2)U_\text{DVPS}(3\pi/2)$, where $U_\text{PS}$ is a standard phase shifter. Just like the NS gate, we can find a realisation of the DVPS with a single ancillary photon, so we take the ancillary state $|\mathbf{a}\rangle = |1,0\rangle$.

In Fig.~\ref{fig:non_deterministic}(b) we present the maximum value of $p_x$ by solving Eq.~\eqref{eq:non_linearity_condition} numerically for $x \in [0,2\pi)$. For $x \equiv 0, \pi \ (\text{mod $2\pi$})$ the probability is unity because for these values the DVPS is the identity and a $\pi$-phase shifter respectively, which are both deterministic gates with linear optics. For all other values of $x$ we see that the probability depends upon $x$ and is bounded as $1/6 \leq p_x \leq 1/4$. See Sec.~\ref{app:numerical_methods} of the appendix and the repository~\cite{data} for details on numerics, analytics, a set of solution unitaries, and the effect of different ancillary states.

\subsection{Photonic fermion sampling \label{sec:fermion_sampling}}

The use of measurement-induced DVPSs of the previous section will be impractical for large interferometers, as each phase shifter is non-deterministic and requires the same number of ancillary modes as input photons. Instead, let us try to directly simulate fermion sampling. 

As a linear fermion sampler is equivalent to time evolution under a non-interacting and particle-conserving fermionic Hamiltonian, this can be simulated with qubits via a Jordan-Wigner transformation. The particle-conserving unitaries are equivalent to Givens rotations on the subspace of fixed Hamming weight and circuits exist for this~\cite{PhysRevA.98.022322,Arrazola2022universalquantum,Anselmetti_2021,PhysRevResearch.5.023071}, allowing simulation of fermion sampling. Additionally, one could simulate fermion sampling with cold atoms in optical lattices~\cite{doi:10.1126/science.aal3837,schafer2020tools} by interpreting each column of an interferometer as one time step of nearest-neighbour hoppings.  

To provide a concrete realisation of fermion sampling with photonic linear optics, we use the method from Ref.~\cite{Matthews_2013}. We simulate $n$ fermions in an $N$-mode interferometer encoding a given unitary by first introducing $n$ identical copies of this interferometer. This system is described by a set of bosonic ladder operators $a_{\mu i}$ with two indices, where the ordered pair $(\mu,i)$ indexes the $i$th mode of the $\mu$th interferometer. These obey the bosonic algebra $[a_{\mu i},a^\dagger_{\nu j}]  = \delta_{\mu \nu}\delta_{ij}$ and $[a_{\mu i},a_{\nu j}] = [a^\dagger_{\mu i},a^\dagger_{\nu j}] = 0$. Then we construct the resource state
\begin{equation}
|\psi_\mathrm{in}\rangle = \frac{1}{\sqrt{n!}} \sum_{\sigma \in S_n} (-1)^\sigma \prod_{ \mu = 1}^n a_{\mu \sigma(\mu)}^{\dagger} |0\rangle, \label{eq:fermionic_input_state}
\end{equation}
where $S_n$ is the permutation group of $n$ elements. This is an entangled $n$-photon state, where each photon is inserted into its own interferometer, and is anti-symmetric upon exchange of any two photons' mode degree of freedom, hence the set of states of this form is isomorphic to the space $\wedge^n(\mathbb{C}^N)$. This state is a representation of the state of the first $n$ modes of a \textit{single} interferometer occupied with a fermion. Note that the complexity of this state scales with the number of particles, $n$, and not the number of modes, $N$. With linear optics alone, the success probability of producing this ansatz is upper bounded by $1/9$ as will be shown momentarily.

The way we sample bit strings from this system is slightly different. As each interferometer contains a single photon, upon measurement each one returns an $N$-bit string with a unit Hamming weight. Due to the anti-symmetric property of the state, no two interferometers will return the same bit string. We construct a new bit string from these by summing them into a single $N$-bit string of Hamming weight $n$. 

The probability to measure the output superimposed bit string $\mathbf{m}$ given the input state had the superimposed bit string $\mathbf{n}$ is given by
\begin{equation}
p(\mathbf{m}|\mathbf{n} ) = \left| \det u[\mathbf{m}|\mathbf{n}] \right|^2,
\end{equation}
where $u \in \mathrm{U}(N)$ is the unitary programmed into the interferometers. This probability displays fermionic statistics as it agrees precisely the with the determinant of Eq.~\eqref{eq:transition_amplitude}. When we simulate fermions using photons we refer to this as \textit{photonic fermion sampling} to contrast with fermion sampling with actual fermions. As the exchange anti-symmetry of the resource state of Eq.~\eqref{eq:fermionic_input_state} is what simulates fermionic behaviour, the set of interferometers must always encode identical unitaries otherwise this breaks down. In Sec.~\ref{app:boson_fermion_sampling} of the appendix we provide a proof of the transition amplitude.

As an example let us consider the case of simulating two fermions. From Eq.~\eqref{eq:fermionic_input_state} the required photonic resource state is given by
\begin{equation}
|\psi\rangle = \frac{1}{\sqrt{2}} \left( a_{11}^\dagger a_{22}^\dagger - a^{\dagger}_{12} a^{\dagger}_{21} \right) | 0 \rangle. \label{eq:2_femrions}
\end{equation} 
This represents state of a fermion sampler with the first two modes occupied, as shown in Fig.~\ref{fig:fermion_sampler}(a), and is equivalent to the Bell state $|\Psi^-\rangle$ if we use the dual-rail encoding of photonic qubits. In Fig.~\ref{fig:fermion_sampler}(b) we show how to prepare this state non-deterministically using linear optics alone which uses a circuit consisting of a 50:50 beamsplitter using the Hadamard convention, a CNOT gate, and a $\pi$ phase shifter. The CNOT gate is either the heralded~\cite{KLM} or the unheralded~\cite{PhysRevA.65.062324} version, working with a probability of $1/16$ and $1/9$ respectively. We opt for the latter as it has the benefit of both a higher success rate and no requirement for ancillary photons. As it is unheralded, we simply need to postselect on there being a single photon per intereferometer to see the fermionic statistics. For this reason, the success probability of preparing Eq.~\eqref{eq:fermionic_input_state} for $n > 1$ with linear optics alone is upper bounded by $1/9$ as at least one CNOT is required. Note that this two-fermion example can also be done by entangling the polarisation degree of freedom instead~\cite{PhysRevLett.108.010502}.

In addition to resulting in a simpler cost landscape, fermion sampling has two useful features. The first is that no many-to-one mapping is required to map the outputs to bit strings as multi-occupation states are forbidden. This removes the redundancy that bosonic linear optics suffers from as every output state is distinguishable even after threshold detection. The second is that the output states have a fixed Hamming weight and this will aid in solving Hamming-constrained QUBO problems to be seen in Sec.~\ref{sec:constrained_qubo}. Additionally, it has been argued that Hamming-constrained systems can help to minimise the barren plateau problem~\cite{Larocca2022diagnosingbarren,monbroussou2023trainability,Cherrat2023quantumdeephedging}.
 
Note that one could obtain fixed Hamming weight bit strings from a standard boson sampler by postselecting these outputs and discarding the rest, however the cost function will still take the general bosonic form and will not yield the simplified fermion cost landscape. This is because the Fourier modes of the cost landscape are determined by how many photons pass through each phase shifter of the interferometer, which can be greater than one for standard bosonic linear optics, and postselection on Hamming weight will not change this.

\section{Losses, sample complexity and resource scaling}

\begin{figure}
\begin{center}
\includegraphics[width=\columnwidth]{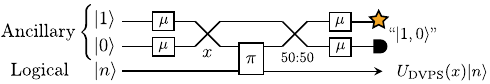}
\end{center}
\caption{The asymptotically deterministic DVPS of Sec.~\ref{sec:deterministic_DVPS} is modified with uniform losses. As uniform losses do not commute with non-linear operations, we insert losses before and after the cross Kerr non-linearity with loss probability of $\mu$. The total effect of the loss is that the ancillary photon has a loss probability of $\eta = \mu(2-\mu)$. \label{fig:lossy_deterministic_dvps}}
\end{figure}

\subsection{Effects of loss on the DVPS}\label{sec:losses}
To model losses in the system with a Hilbert space $\mathcal{H}$, we introduce a set of environment modes and extend the Hilbert space to $\mathcal{H} \otimes \mathcal{H}_\mathrm{E}$, where $\mathcal{H}_\mathrm{E}$ is the Hilbert space of the environment. We couple these spaces such that the photons have a probability of escaping into one of these modes as 
\begin{equation}
a_i^\dagger \rightarrow \sqrt{1 - \eta} a_i^\dagger + \sqrt{\eta} e_i^\dagger, \label{eq:single_mode_loss}
\end{equation}
where $e_i$ is the environment mode that the system mode $a_i$ is coupled to and $\eta \in [0,1]$ parametrises the loss. Now suppose we have an $N$-mode interferometer encoding a unitary that decomposes into a set of beamsplitters and phase shifters. We can insert losses into this interferometer at various points, modifying the action of the interferometer to
\begin{equation}
\mathbf{a}^\dagger \rightarrow A \mathbf{a}^\dagger + \sqrt{ 1 - A A^\dagger} \mathbf{e}^\dagger, \label{eq:multi_mode_loss}
\end{equation}
where $\mathbf{a} = (a_1,a_2,\ldots,a_N)$ and $\mathbf{e} = (e_1,e_2,\ldots,e_N)$ are vectors of ladder operators of the system and environment, respectively, and $A$ models the lossy interferometer, where $A A^\dagger \leq 1$~\cite{GarciaPatron2019simulatingboson}. 
This is a unitary transformation on the space of all modes, but on the system modes alone it is non-unitary. By tracing out the environment modes, this induces a quantum channel
\begin{equation}
\Phi_\eta(\rho) = \mathrm{Tr}_\mathrm{E}[U_\mathrm{L} (\rho \otimes |0\rangle \langle 0 |) U_\mathrm{L}^\dagger ], \label{eq:loss_channel}
\end{equation}
where $U_\mathrm{L}$ is the unitary implementing Eq.~\eqref{eq:multi_mode_loss} and $|0\rangle$ is the vacuum of the environment.

\begin{figure*}[t!]
\begin{center}
\includegraphics[width=\textwidth]{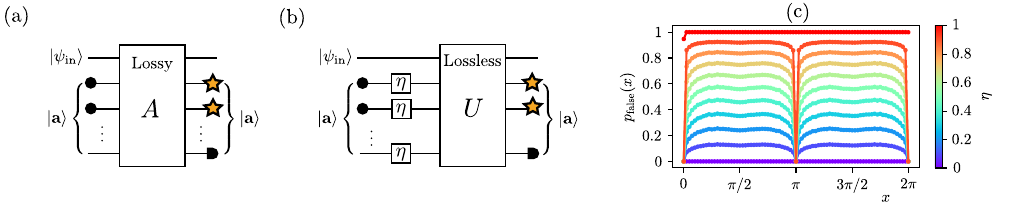}
\end{center}
\caption{(a) The lossy measurement-induced DVPS is modelled by taking the unitary $U$ for the DVPS and inserting losses throughout, yielding a non-unitary transformation $A$. (b) If we have uniform loss of $\eta$, then the lossy DVPS is equivalent to applying losses given by Eq.~\eqref{eq:single_mode_loss} to the input states (environment modes not shown). We ignore loss on the logical mode as it can be absorbed into the circuit before. (c) The Haar-averaged probability for a false herald for the DVPS on the subspace of at most $N = 2$ photons, c.f. Sec.~\ref{sec:non_deterministic_DVPS}, versus the phase of the DVPS $x$ for various losses $\eta$. The probability of false herald is $p_\text{false}(x) \approx \eta$.} \label{fig:lossy_DVPS}
\end{figure*} 

If we take the singular value decomposition $A = W \Sigma V^\dagger$, where $W$ and $V$ are unitary matrices and $\Sigma$ is a diagonal matrix of singular values $\Sigma_{ij} \equiv \sqrt{1 - \eta_i} \delta_{ij}$, then we see the lossy transformation on the interferometer modes decomposes into a set of single mode losses, represented by $\Sigma$ and given by Eq.~\eqref{eq:single_mode_loss}, sandwiched between two lossless unitaries. If the losses are uniform, so that $\eta_i \equiv \eta$ for all $i$, then we can commute the losses through the system and write $A = W V^\dagger \Sigma \equiv U \Sigma $ in which case the effect of uniform loss is equivalent to applying uniform loss to each input mode of the system, then applying the original lossless unitary $U$~\cite{GarciaPatron2019simulatingboson}. Below we assume uniform loss for simplicity.

Before we analyse the effects of loss in the three designs of the DVPS we note that two of the designs rely on heralding of ancillary photons in some way. We assume that only the ancillary modes have loss and the logical modes are lossless. This is because the successful operation is heralded by the measurement of ancillary photons, so the effects of loss on the logical operation of the gate can be understood by understanding how losses affect the ancillary photons and heralding.

\subsubsection{Non-linear design}
Let us consider the design using non-linear optics of  Sec.~\ref{sec:deterministic_DVPS}. In the presence of losses a single attempt of this gate now has two ways to fail: either no losses occur and it fails due to the non-determinism of the gate, or it fails due to a loss. As the ancillary modes do not mix with the logical modes, we can detect the type of failure when we measure the ancillary modes. If the gate fails in the first way we can correct via feed-forward as in the lossless case, but if the gate fails in the second way then we would discard this run as photon loss cannot be corrected for. Hence, by post selecting on runs with no loss, the effects of loss on this system would not affect the logical operation of the DVPS, yet would reduce the success probability from $1/2$ to $(1- \eta)/2$, where $\eta$ is the probability of loss in the ancillary modes. In Fig.~\ref{fig:lossy_deterministic_dvps} we show an example of how losses are modelled.

As the repeat until success protocol works only when no losses occur, the probability of success after at most $m$ repeats becomes
\begin{equation}
\begin{aligned}
\sum_{k = 1}^m \left(\frac{1 - \eta}{2}\right)^k &  =  \left[ 1 - \left( \frac{1 - \eta}{2} \right)^m \right] \frac{1 - \eta}{1 + \eta} \\
& \xrightarrow[m \to \infty]{} \frac{1 - \eta}{1 + \eta }, \label{eq:non_linear_dvps_p}
\end{aligned}
\end{equation}
where we used the fact that $\left(\frac{1- \eta}{2}\right)^k$ is the probability for $k-1$ failures with no loss followed by a success. We see that with losses it is no longer asymptotically deterministic. The form of the fermionic cost landscape generated by this method will be unaffected as long as we postselect on these successful runs.

\subsubsection{Linear design}

We now consider the non-deterministic design of  Sec.~\ref{sec:non_deterministic_DVPS}. This gate uses linear optics only, allowing us to commute uniform losses through the system to the input state, as shown in Figs.~\ref{fig:lossy_DVPS}(a)-(b). The probability of success of the gate is clearly the probability for no ancillary photons to be lost, so if the success of the lossless DVPS with a phase of $x$ is $p(x)$, then the success probability of the lossy DVPS is given by $(1-\eta)^k p(x)$, where $k$ is the number of ancillary photons.  

However, unlike the previous design, we are unable to detect this loss as ancillary photons can mix with the logical modes. This means that measuring the correct ancillary output state does not necessarily mean the gate has operated successfully now, resulting in a chance of false heralds. For example, if an ancillary photon is lost to the environment, it can be replenished with a photon from the logical mode resulting in the correct herald, despite the fact the the logical mode has now lost a photon. The probability of measuring the ancillary output $|\mathbf{a}\rangle$ given a loss $\eta$ is given by
\begin{equation}
p_\text{H}(x, \eta) = \mathrm{Tr}\left[ U(x) \Phi_\eta(\rho) U^\dagger(x) \Pi \right], \ \Pi = \mathbb{I} \otimes | \mathbf{a} \rangle \langle \mathbf{a} |
\end{equation}
where $\rho = |\psi_\text{in}\rangle \langle \psi_\text{in}| \otimes |\mathbf{a}\rangle \langle \mathbf{a}|$ is the input state, $\Phi_\eta$ is the loss channel of Eq.~\eqref{eq:loss_channel} modelling losses to each input ancillary mode, and $U(x)$ is the lossless DVPS unitary for a phase $x$. Note that $p_\text{H}(x,0) = p(x)$. The conditional probability that given a herald it is a false herald is given by
\begin{equation}
p_\text{false}(x, \eta) = \frac{p_\mathrm{H}(x,\eta) - (1- \eta)^k p(x)}{p_\mathrm{H}(x,\eta)}.
\end{equation}
In Fig.~\ref{fig:lossy_DVPS}(c) we calculate this numerically for the $N = 2$ DVPS with the ancilla state $|\mathbf{a}\rangle = |1,0\rangle$ that was discussed in Sec.~\ref{sec:non_deterministic_DVPS}. This probability depends upon the input state $|\psi_\text{in}\rangle$, so we present the average probability as we average over input states with respect to the Haar measure. We see that given a loss of $\eta$ then $p_\text{false}(x, \eta) \approx \eta$, except at multiples of $\pi$ where the gate is deterministic and no longer requires heralding. This shows that a measurement-induced embodiment is sensitive to losses. Details on this calculation are presented Sec.~D of the appendix.

Despite this, if we have photon counting detectors at the very end of our circuit, we can detect loss somewhere in the circuit if the number of particles is not conserved when all DVPSs were heralded correctly. If we have this capability, we simply postselect on runs that conserve particle number, eliminating the possibility of false heralds. In this case, the cost landscape would be fermionic. On the other hand, if we do not have this capability then the cost landscape will diverge from being fermionic.

\subsubsection{Photonic fermion sampling}

Finally, the third design was to simulate the effects of a DVPS by imposing fermionic statistics on our particles, as presented is Sec.~\ref{sec:fermion_sampling}. Here, the DVPS is a phase shifter now and has no inner workings with ancilla particles; it is the Pauli exclusion principle that enforces the two-eigenvalue property. If we are working with actual fermions, the Pauli exclusion principle persists even with losses, so the fermionic cost landscape persists.

If we perform photonic fermion sampling, then we are simulating fermionic statistics with an entangled photonic resource state. Let $\rho = |\psi_\text{in}\rangle \langle \psi_\text{in}|$ be the $n$-photon resource state of Eq.~\eqref{eq:fermionic_input_state} and recall that this state consists of a single photon per interferometer. Suppose that uniform losses are present in, say, the first interferometer. The state transforms to 
\begin{equation}
\rho \rightarrow \Phi_\eta(\rho) = (1- \eta)\rho + |0\rangle \langle 0| \otimes \frac{\eta}{n} \sum_{k = 1}^n \rho_k
\end{equation}
where $\Phi_\eta$ is the uniform loss channel of Eq.~\eqref{eq:loss_channel} applied to all modes of the first interferometer only, $|0\rangle \langle 0 |$ is the vacuum state of the first interferometer and $\rho_k$ is a pure $(n-1)$-photon fermionic resource state of the form of Eq.~\eqref{eq:fermionic_input_state} with support on the remaining $n-1$ interferometers only whose photons are distributed across the set of modes $\{ 1, 2, \ldots, n \} \setminus \{ k \}$. Each term is completely anti-symmetric so the fermionic nature and fermionic cost landscape is preserved. See Sec.~\ref{app:loss} of the appendix for the proof of this. This generalises to losses that affect multiple interferometers as we can repeatedly apply this transformation. Despite the fermionic cost landscape being preserved, Hamming weight will not be due to losses.

\subsection{Resources for the three designs of the DVPS}

\begin{table}
\begin{center}
\resizebox{\columnwidth}{!}{\begin{tabular}{|c |c | c | c| c |}
\hline
Method & Modes & Ancillas & $p_\text{ansatz}$ & $p_\text{gate}$ \\
\hline
\hline
NL  & $N(1 + (1+m)N)$  & $N^2$ & 1 & $\leq (1 - \eta)/(1 + \eta)$ \\
L  & $N(1 + nN)$& $N^2n $ & 1 & $\leq (1-\eta)^n/n^2$ \\
FS  & $N$ & $0$ & 1 & $1$ \\
PFS & $Nn +O(n^4)$ & $O(n^4)$ & $\leq 1/9$ & 1 \\
\hline
\end{tabular}}
\caption{The physical resources for simulating a sampling experiment of $n$ particles in an $N$-mode universal interferometer consisting of $N^2$ phase shifters when we use the non-linear DVPS with $ m $ repeats (NL), measurement-induced linear DVPS (L), fermion sampling (FS) and photonic fermion sampling (PFS)  For each case, we list how many modes the interferometer requires, the total number of ancilla particles required, the success probability to generate the ansatz, and the success probability for each parametrised phase shifter. Here $\eta$ is the probability for loss in the ancillary modes. \label{tab:scaling}}
\end{center}
\end{table}

We now compare the resources of the three methods of presented in Secs.~\ref{sec:deterministic_DVPS}, \ref{sec:non_deterministic_DVPS} and \ref{sec:fermion_sampling} for conducting a sampling experiment, which are shown in Table~\ref{tab:scaling}. For each method, we consider sampling from an $N$-mode universal interferometer, which contains $N^2$ phase shifters~\cite{Clements:16,PhysRevLett.73.58}, with $n$ input particles. The phase shifters are either standard or DVPS depending upon the method of choice. We additionally make two key assumptions. First, we assume that single-occupation number state inputs can be created deterministically for both fermionic and bosonic statistics; and second we assume that the physical hardware of each circuit, such as beamsplitters, standard phase shifters (not DVPS) and Kerr non-linearities, operate deterministically---the non-determinism of any method arises from preparing the input ansatz or the measurement and postselection. In 

First, the non-linear DVPS of Sec.~\ref{sec:deterministic_DVPS}. The input states are simple number states so are deterministically prepared. Each gate requires two ancillary modes and one ancillary photon, however if the gate is repeated $m$ times in the repeat-until-success protocol we require $1 + m$ ancillary modes, so the entire interferometer requires $(1+m)N^2$ ancillary modes, $N$ logical modes, $N^2$ ancillary photons, and $n$ logical photons. Each phase shifter operates with a success probability of $1 - 1/2^m$. If we take into account losses with loss probability $\eta$, then this becomes $(1-\eta)/(1+\eta)$ in the limit of large $m$, c.f. Eq.~\eqref{eq:non_linear_dvps_p}. However, this relies upon a strong cross Kerr non-linearity which makes this infeasible with photonics today.

Second, the linear DVPS of Sec.~\ref{sec:non_deterministic_DVPS}. The input states for this system are also simple number states so are deterministically prepared. Each gate requires $n$ ancillary modes with $n$ ancillary photons in general, see Sec.~\ref{app:numerical_methods} of the appendix for more information, where each gate must be a universal interferometer itself in order to encode any unitary required. The DVPS fall into the class of generalised NS gates and these have a success probability of at most $1/n^2$~\cite{Scheel_2005} which reduces to $(1-\eta)^n/n^2$ with losses. This gives a total of $N^2 n$ ancillary modes, $N$ logical modes, $N^2 n$ ancillary photons and $n$ logical photons.

Finally, if we perform fermion sampling of Sec.~\ref{sec:fermion_sampling} with, say cold atoms in an optical lattice, then this system has no ancillary modes, and each gate is deterministic as only standard phase shifters are required. On the other hand, if we were to perform photonic fermion sampling we require $n$ copies of the interferometer giving us $Nn$ modes in total. Simulating $n$ fermions also requires a state preparation stage with $O(n^4)$ controlled swap gates~\cite{Matthews_2013}. With linear optics alone, this is non-deterministic and requires $O(n^4)$ ancillary photons and $O(n^4)$ ancillary modes. The probability of success of state preparation is at most $1/9$ (see Sec.~\ref{sec:fermion_sampling}) however each phase shifter in the interferometer now works deterministically as they are linear

We see that both realisations of fermion sampling scale the best out of the methods presented. Whilst the non-linear DVPS is promising, it requires infeasible non-linearities. As we are interested in photonic linear optics, one would opt for photonic fermion sampling as it scales better than the DVPS for large problems because each gate is deterministic and the total number of photons required is independent of $N$, allowing for very large $N$ systems to be constructed at relatively low cost.

\subsection{Statistical estimation and sample complexity}
The two algorithms we focus on in this work are gradient descent using the parameter shift rule and Rotosolve. Both algorithms are iterative, whereby the parameters of the interferometer are updated successively using data obtained from the previous iteration. In this section, we investigate how many samples of the interferometer are required to evaluate each iteration to a given accuracy. Following Ref.~\cite{facelli2024exactgradientslinearoptics}, we assume that the variance of the observable $H$ present in the cost function is independent of the parameter, in other words we assume that
\begin{equation}
\langle H^2 \rangle_{\boldsymbol{\theta}} - \langle H \rangle^2_{\boldsymbol{\theta}}= \sigma^2,
\end{equation}
for all $\boldsymbol{\theta}$, where $\langle \cdot \rangle_{\boldsymbol{\theta}} = \mathrm{Tr}[U(\boldsymbol{\theta}) \rho U^\dagger(\boldsymbol{\theta}) (\cdot)]$. If this is not the case, we take $\sigma^2$ to be the maximum. Additionally, given an estimator $\hat{\theta}$ that we use to estimate a quantity $\theta$, we use the mean squared error as a measure of its error, defined as $\mathrm{MSE}[\hat{\theta}] = \mathbb{E}[ (\hat{\theta} - \theta)^2 ]$.

For gradient descent, the gradient is calculated using the parameter shift rule of Eqs.~\eqref{eq:parameter_shift_rule} or \eqref{eq:fermionic_parameter_shift}. As this gradient is a function of measured data, we must use an estimator to estimate the gradient from the data. The estimator for the $(n+1)$th iteration parameters is
\begin{equation}
\hat{\boldsymbol{\theta}}^{(n+1)} = \boldsymbol{\theta}^{(n)} - h \widehat{\nabla} E^{(n)},
\end{equation}
which is the usual gradient descent rule with estimators substituted in, where $\widehat{\nabla} E$ is the estimator for the gradient. If we have a bosonic cost landscape, it was shown in Ref.~\cite{facelli2024exactgradientslinearoptics} that if we wish to estimate this update to an error of $\epsilon^2$, then we require 
\begin{equation}
N_\text{S} = O \left( \frac{\sigma^2 n^2}{\epsilon^2} \right),
\end{equation}
samples. If we have a fermionic cost landscape, then this is equivalent to a bosonic cost landscape with $n = 1$ so the number of samples is
\begin{equation}
N_\text{S} = O \left( \frac{\sigma^2}{\epsilon^2} \right),
\end{equation}
and no longer scales with the number of particles, giving us an improvement.

For Rotosolve, we construct an estimator for the parameter update rule by replacing all quantities in Eq.~\eqref{eq:rotosolve_iteration} with estimators. For the $\mu$th parameter $x \equiv \theta_\mu$, the estimator is
\begin{subequations}
\label{eq:rotosolve_estimator}
\begin{align}
\hat{x}_0 & = - \pi/2 - \operatorname{atan2}(\hat{Y},\hat{X}), \\
\hat{X} & = \hat{f} (\pi/2) - \hat{f} (-\pi/2),  \\
\hat{Y} & = 2\hat{f} (0) - \hat{f} (\pi/2) - \hat{f} (-\pi/2), 
\end{align}
\end{subequations}
where $\hat{f}(x) \equiv \hat{E}(\theta_\mu = x) = \frac{1}{N_\text{S}} \sum_{i = 1}^{N_\text{S}} E_i$ is the estimator for the cost function given by an average of $N_\text{S}$ measurements $\{ E_i \}$ of the observable $H$. As the estimator $\hat{x}_0$ is a non-linear function of other estimators, it turns out it is biased and evaluating its mean squared error is non-trivial. In Sec.~\ref{app:estimation} of the appendix, we show that the number of samples required for this estimator is given by 
\begin{equation}
N_\text{S} = O \left( \frac{\sigma^2}{\epsilon^2 A^2} \right),
\end{equation}
where $A$ is the amplitude of the fermionic cost landscape of Eq.~\eqref{eq:fermionic_cost}. This will further reduce the number of samples required as long as the amplitude $A$ is larger than the standard deviation $\sigma$.

\begin{figure*}[t!]
\begin{center}
\includegraphics[width=1\textwidth]{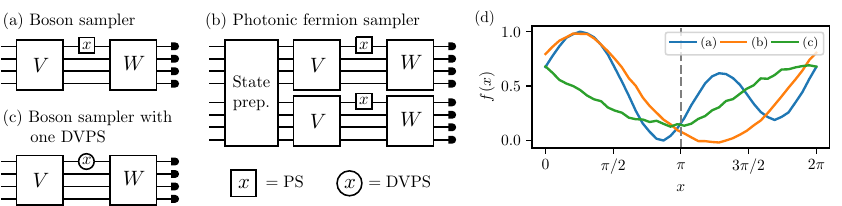}
\end{center}
\caption{(a)~The boson sampler is an $N$-mode interferometer constructed from a single phase shifter (PS) of phase $x$ sandwiched between two fixed linear unitaries. (b)~The photonic fermion sampler simulating $n$ fermions is an $nN$-mode interferometer consisting of an entangling state preparation state (shown in Fig.~\ref{fig:fermion_sampler}) followed by $n$ identical copies of the boson sampler. For this particular example, $n = 2$. (c)~The bosonic DVPS sampling experiment has the same hardware as the boson sampler, except the parametrised phase shifter has been replaced with a DVPS. (d)~The numerical simulation of the cost landscapes for the three cases in (a)-(c) as we vary $x$ with the input states $|\psi_\text{in}\rangle = |1100\rangle$ for the bosonic cases, and the two-photon entangled state from Eq.~\eqref{eq:fermionic_input_state} for the fermionic case, which represents the fermionic state $|\psi_\text{in}\rangle = |1100\rangle$. The same observable $H$ and fixed linear unitaries $V$ and $W$ were chosen for each case. The boson sampler produces two harmonics as expected, as it contains two photons. This reduces to a single harmonic if a photonic fermion sampler or DVPS is used instead which removes the local minimum. Numerical simulation parameters were taken to be $10^6$ shots per $x$, purity of 0.9, indistinguishability of 0.95, brightness of 0.5 (equivalently loss of $0.5$), dark counts rate of $5 \times 10^{-6}$ and threshold detection with efficiency of $0.9$. \label{fig:cost_function_comparison}}
\end{figure*} 

\section{Results \label{sec:results}}

\subsection{Cost landscapes}

In this section we numerically test the cost landscapes that can be achieved using the three methods of boson sampling of Sec.~\ref{sec:vqe}, the non-deterministic DVPS of Sec.~\ref{sec:non_deterministic_DVPS}, and photonic fermion sampling of Sec.~\ref{sec:fermion_sampling}. All three methods can be implemented today with linear optics, single photon sources and threshold detectors. 

For each method, we insert two photons into an $N$-mode linear interferometer encoding the parametrised unitary $W U(x) V$, where $V$ and $W$ are two randomly-chosen linear unitaries that are held fixed and $U(x)$ is either a parametrised phase shifter or DVPS depending on which of the three methods we choose. The circuits of these three methods are shown in Figs.~\ref{fig:cost_function_comparison}(a)-(c). If the DVPS is used, then for each $x$ the DVPS equation of Eq.~\eqref{eq:non_linearity_condition} must be solved numerically beforehand to find the required unitaries to encode. See table~\ref{tab:unitary_table} of Sec.~\ref{app:numerical_methods} of the appendix and the repository~\cite{data} for the algorithm to generate these. Without loss of generality, we choose an observable $H$ for the cost function that is diagonal in the number basis with eigenvalues chosen randomly, which is equivalent to assigning a randomly-chosen cost to each possible output bit string.

We numerically simulate these three methods using the Python package Lightworks~\cite{lightworks, data} which takes into account realistic effects of finite sampling, photon purity, photon indistinguishability, photon source brightness, photon losses, threshold detection, detector dark counts and detector efficiency. In Fig.~\ref{fig:cost_function_comparison}(d) we compare the cost landscape as we sweep over the phase $x$ for the three methods. The cost function for boson sampling has two harmonics due to the presence of two photons, as predicted by Eq.~\eqref{eq:bosonic_cost_landscape}, giving rise to a local minimum. If a DVPS or photonic fermion sampling is used instead, the cost function has a single harmonic as predicted by Eq.~\eqref{eq:fermionic_cost} which removes this local minimum, moreover this property persists with the presence of loss as discussed in Sec.~\ref{sec:losses}. Note that the cost function of the DVPS and phase shifter agree at integer multiples of $\pi$ as they are equivalent unitaries for these values. In Sec.~\ref{sec:fermion_sampling} we show how to prepare the two-fermion resource state and in Sec.~\ref{app:numerical_methods} of the appendix we show how to construct the non-deterministic DVPS.

\subsection{Application to constrained and unconstrained QUBO problems \label{sec:constrained_qubo}}

Using dual-valued phase shifters as our parametrised unitaries instead of phase shifters simplifies the cost landscape considerably. In this section we assess the consequences of this for variational quantum algorithms using gradient descent and Rotosolve as their optimisers. We focus on solving QUBO problems as many real-world problems take their form. As the DVPS simplifies the cost landscape in a manner that should be independent of the variational problem we are solving, we generate random QUBO instances to solve to avoid any bias towards any particular type of problem.

In order to simulate a DVPS, we opt for fermion sampling of Sec.~\ref{sec:fermion_sampling} as it displays the desired features whilst being numerically efficient to simulate. Any discrete variable linear optics simulator can be modified to simulate fermion sampling by replacing the permanent with the determinant when calculating transition amplitudes in Eq.~\eqref{eq:transition_amplitude} allowing us to simulate large systems. From this point onwards we will be comparing the performance of boson sampling and fermion sampling and all simulations were conducted using the fermion sampler of Lightworks~\cite{lightworks}.

We first test this for solving Hamming constrained QUBO problems. This constraint can be imposed by modifying the cost function of Eq.~\eqref{eq:QUBO_cost} with a penalty term as
\begin{equation}
C_w(\mathbf{x}) = C(\mathbf{x}) + \lambda \left(w - \sum_{i = 1}^N x_i  \right)^2, \label{eq:hamming_constrained_cost}
\end{equation}
where $\lambda$ is a Lagrange multiplier and $w$ is the Hamming constraint. As the outputs of a fermion sampler have a fixed Hamming weight, this will aid the optimiser as all outputs obey the constraint. One could argue that for an $N$-mode interferometer with $n$ bosons, where $N \gg n^2$, the bosonic birthday paradox will come to our aid as the chance of multiply-occupied modes at the outputs is low, so most outputs of a boson sampler will have a fixed Hamming weight anyway~\cite{arkhipov2012bosonic,10.1145/1993636.1993682}. However, for some problems $w$, and hence $n$, may be on the order of $N$ so this will not apply. An example problem is the travelling salesman problem~\cite{tsp}. In these cases, removing redundancy with a fermion sampler is desirable.

\begin{figure}[t]
\includegraphics[width=\columnwidth]{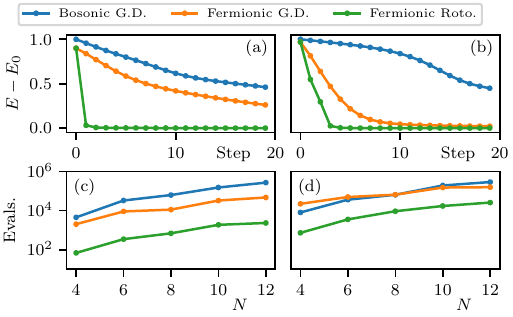}
\caption{(a) The cost function over time for a random QUBO problem of size $N = 8$ with Hamming constraint $w = 3$. The choice of algorithm was either boson sampling with gradient descent, fermion sampling with gradient descent, or fermion sampling with Rotosolve. (b) The same test for a random unconstrained problem of size $N = 4$. (c) The average number of cost function evaluations before termination for 20 random Hamming constrained QUBO problems with Hamming constraint $w = \lfloor N/2 \rfloor$ vs. the problem size $N$. For each $N$ the number of input particles was $w$. Termination occurs when the relative change $(E_n - E_{n-1})/E_{n-1} \leq 10^{-4}$. (d) The same test but for unconstrained problems, except now the photonic fermion sampler has $2N$ modes instead.  \label{fig:constrained_scaling}}
\end{figure}
In Fig.~\ref{fig:constrained_scaling}(a) we present the cost function over time for a random constrained QUBO problem with these three methods, where gradient descent uses the learning rate $h = 0.01$. Gradients were evaluated using the parameter shift rules of Eqs.~\eqref{eq:parameter_shift_rule}~and~\eqref{eq:fermionic_parameter_shift}. We see how effective Rotosolve is compared to gradient descent, as for this particular example with a solution space of 56 solutions (8-bit strings with Hamming weight 3), it finds the minimum in just two steps.

In Fig.~\ref{fig:constrained_scaling}(c) we compared the average number of cost function evaluations for solving constrained problems of various sizes $N$ with the Hamming constraint of $w = \lfloor N/2 \rfloor$ for the three methods. Each time, the QUBO problem was constructed by generating a random symmetric matrix $Q$ whose integer elements were randomly sampled in the range $[-10,10]$ and the Lagrange multiplier was taken to be $\lambda = 2 \operatorname{max} \{ Q_{ij} \}$. To solve these problems, we chose an $N$-mode system with $N/2$ layers of the Clements design~\cite{Clements:16} which has about $N^2/2$ variational parameters, together with the collision-free input state consisting of a single particle in the first $w$ modes. For each algorithm, the termination criteria was taken to be when $(E_n - E_{n-1})/E_{n-1} < 10^{-4}$, where $E_n$ is the cost function of the $n$th step and we define a step as one update of all variational parameters. Fermion sampling is seen to outperform boson sampling when using gradient descent and gains a further speed-up when using Rotosolve. 

\begin{figure}[t]
\begin{center}
\tikzset{every picture/.style={line width=0.75pt}} 

\begin{tikzpicture}[x=0.75pt,y=0.75pt,yscale=-1,xscale=1]

\draw [line width=0.75]    (74,70) -- (94,70) ;
\draw [line width=0.75]    (74,80) -- (94,80) ;
\draw [line width=0.75]    (74,90) -- (94,90) ;
\draw [line width=0.75]    (74,100) -- (94,100) ;
\draw [line width=0.75]    (74,110) -- (94,110) ;
\draw [line width=0.75]    (104,100) -- (114,100) ;
\draw [line width=0.75]    (104,90) -- (114,90) ;
\draw [line width=0.75]    (104,70) -- (114,70) ;
\draw [line width=0.75]    (74,60) -- (94,60) ;
\draw [line width=0.75]    (94,60) -- (104,70) ;
\draw [line width=0.75]    (94,70) -- (104,60) ;
\draw [line width=0.75]    (94,80) -- (104,90) ;
\draw [line width=0.75]    (94,90) -- (104,80) ;
\draw [line width=0.75]    (94,100) -- (104,110) ;
\draw [line width=0.75]    (94,110) -- (104,100) ;
\draw [line width=0.75]    (124,70) -- (134,80) ;
\draw [line width=0.75]    (124,80) -- (134,70) ;
\draw [line width=0.75]    (104,80) -- (114,80) ;
\draw [line width=0.75]    (124,90) -- (134,100) ;
\draw [line width=0.75]    (124,100) -- (134,90) ;
\draw [line width=0.75]    (104,60) -- (114,60) ;
\draw [line width=0.75]    (134,70) -- (154,70) ;
\draw [line width=0.75]    (134,80) -- (154,80) ;
\draw [line width=0.75]    (134,90) -- (154,90) ;
\draw [line width=0.75]    (134,100) -- (154,100) ;
\draw  [fill={rgb, 255:red, 0; green, 0; blue, 0 }  ,fill opacity=1, line width=1.5]  (154,58) -- (156,58) .. controls (157.1,58) and (158,58.9) .. (158,60) .. controls (158,61.1) and (157.1,62) .. (156,62) -- (154,62) -- cycle ;
\draw  [fill={rgb, 255:red, 0; green, 0; blue, 0 }  ,fill opacity=1, line width=1.5]  (154,78) -- (156,78) .. controls (157.1,78) and (158,78.9) .. (158,80) .. controls (158,81.1) and (157.1,82) .. (156,82) -- (154,82) -- cycle ;
\draw  [fill={rgb, 255:red, 0; green, 0; blue, 0 }  ,fill opacity=1, line width=1.5]  (154,88) -- (156,88) .. controls (157.1,88) and (158,88.9) .. (158,90) .. controls (158,91.1) and (157.1,92) .. (156,92) -- (154,92) -- cycle ;
\draw [line width=0.75]    (74,130) -- (94,130) ;
\draw [line width=0.75]    (104,130) -- (114,130) ;
\draw [line width=0.75]    (74,120) -- (94,120) ;
\draw [line width=0.75]    (74,130) -- (94,130) ;
\draw [line width=0.75]    (104,120) -- (114,120) ;
\draw [line width=0.75]    (104,110) -- (114,110) ;
\draw [line width=0.75]    (94,120) -- (104,130) ;
\draw [line width=0.75]    (94,130) -- (104,120) ;
\draw [line width=0.75]    (124,110) -- (134,120) ;
\draw [line width=0.75]    (124,120) -- (134,110) ;
\draw [line width=0.75]    (134,110) -- (154,110) ;
\draw [line width=0.75]    (134,120) -- (154,120) ;
\draw   (154,96) -- (155,96) -- (155,132) -- (154,132) -- cycle ;
\draw [line width=0.75]  [dash pattern={on 1pt off 1pt}]  (114,120) -- (124,120) ;
\draw [line width=0.75]  [dash pattern={on 1pt off 1pt}]  (114,130) -- (124,130) ;
\draw [line width=0.75]  [dash pattern={on 1pt off 1pt}]  (114,110) -- (124,110) ;
\draw [line width=0.75]  [dash pattern={on 1pt off 1pt}]  (114,100) -- (124,100) ;
\draw [line width=0.75]  [dash pattern={on 1pt off 1pt}]  (114,90) -- (124,90) ;
\draw [line width=0.75]  [dash pattern={on 1pt off 1pt}]  (114,80) -- (124,80) ;
\draw [line width=0.75]  [dash pattern={on 1pt off 1pt}]  (114,70) -- (124,70) ;
\draw [line width=0.75]  [dash pattern={on 1pt off 1pt}]  (114,60) -- (124,60) ;
\draw [line width=0.75]    (124,130) -- (154,130) ;
\draw [line width=0.75]    (124,60) -- (154,60) ;
\draw  [fill={rgb, 255:red, 0; green, 0; blue, 0 }  ,fill opacity=1 ] (71.5,60) .. controls (71.5,58.62) and (72.62,57.5) .. (74,57.5) .. controls (75.38,57.5) and (76.5,58.62) .. (76.5,60) .. controls (76.5,61.38) and (75.38,62.5) .. (74,62.5) .. controls (72.62,62.5) and (71.5,61.38) .. (71.5,60) -- cycle ;
\draw  [fill={rgb, 255:red, 0; green, 0; blue, 0 }  ,fill opacity=1 ] (71.5,70) .. controls (71.5,68.62) and (72.62,67.5) .. (74,67.5) .. controls (75.38,67.5) and (76.5,68.62) .. (76.5,70) .. controls (76.5,71.38) and (75.38,72.5) .. (74,72.5) .. controls (72.62,72.5) and (71.5,71.38) .. (71.5,70) -- cycle ;
\draw  [fill={rgb, 255:red, 245; green, 166; blue, 35 }  ,fill opacity=1 ] (156.34,64.58) -- (157.78,67.38) -- (161.02,67.83) -- (158.68,70) -- (159.23,73.07) -- (156.34,71.62) -- (153.45,73.07) -- (154,70) -- (151.66,67.83) -- (154.89,67.38) -- cycle ;
\draw  [fill={rgb, 255:red, 0; green, 0; blue, 0 }  ,fill opacity=1 ] (71,80.5) .. controls (71,79.12) and (72.12,78) .. (73.5,78) .. controls (74.88,78) and (76,79.12) .. (76,80.5) .. controls (76,81.88) and (74.88,83) .. (73.5,83) .. controls (72.12,83) and (71,81.88) .. (71,80.5) -- cycle ;
\draw  [fill={rgb, 255:red, 0; green, 0; blue, 0 }  ,fill opacity=1 ] (71,90.5) .. controls (71,89.12) and (72.12,88) .. (73.5,88) .. controls (74.88,88) and (76,89.12) .. (76,90.5) .. controls (76,91.88) and (74.88,93) .. (73.5,93) .. controls (72.12,93) and (71,91.88) .. (71,90.5) -- cycle ;

\draw (97,42.4) node [anchor=north west][inner sep=0.75pt]    {$U(\boldsymbol{\theta } )$};
\draw (24,54) node [anchor=north west,scale=1.2][inner sep=0.75pt]    {$|\psi _{\text{F}} \rangle \begin{cases}
 & \\
 & \\
 & \\
 & \\
\end{cases}$};
\draw (145,45) node [anchor=north west][inner sep=0.75pt]    {$\begin{drcases}
 & \\
 & 
\end{drcases}\begin{pmatrix}
0\\
1\\
0\\
0
\end{pmatrix}$};
\end{tikzpicture}

\end{center}
\caption{All possible $N$-bit strings are generated by inserting $N$ fermions into a $2N$-mode linear optical interferometer in the state $|\psi_\mathrm{F}\rangle$, represented by the black circles, and measuring only the top $N$ modes. In order to perform this with photonic fermion sampling, we must introduce $N$ copies of the interferometer as discussed in Sec.~\ref{sec:fermion_sampling} and shown in Fig.~\ref{fig:fermion_sampler} \label{fig:fermion_sampling_unconstrained}}
\end{figure}

We can also use a fermion sampler to solve unconstrained problems. To generate all possible $N$-bit strings with a fermion sampler, we double the number of modes by introducing an interferometer with $2N$ modes and measure only the top half, as shown in Fig.~\ref{fig:fermion_sampling_unconstrained}. In Fig.~\ref{fig:constrained_scaling}(b) we present the cost function over time for a random unconstrained QUBO problem with the same three methods as before, and in Fig.~\ref{fig:constrained_scaling}(d) we compare the number of cost function evaluations as the problem scales. In each case, we sample random QUBO matrices as before, except now for the fermion sampler we choose $N$ layers of the Clements design instead of $N/2$ as we have doubled the number of modes. We see that fermion sampling with gradient descent has little to no improvement over boson sampling, which is likely due to having double the number of variational parameters now. However, Rotosolve gives us a significant speed-up again.

Benchmarking the performance of variational quantum algorithms is difficult due to their heuristic nature and the inability to gauge the time scaling analytically. Additionally, direct comparison with classical algorithms for solving the same problems could be difficult as they may perform in completely different ways. However, here we have indirectly proposed a classical variational algorithm which can be directly compared with the original boson sampling variational quantum algorithm. This is because fermion and boson samplers operate in precisely the same way, with the same circuits, ansatzes, and classical optimisers, except that fermion sampling is classically efficient to simulate. This implies the existence a classical algorithm that outperforms the best-known boson sampling quantum variational algorithm, and for this reason the quantum advantage of boson sampling does not translate to a \textit{practical} quantum advantage for these algorithms.

\subsection{Rotosolve for boson sampling}

Rotosolve is an effective algorithm for minimising sinusoidal cost functions because there is an analytical expression for the unique minimum for a given parameter direction given just three samples of the cost function. This algorithm can be generalised to handle bosonic cost functions with multiple harmonics like Eq.~\eqref{eq:bosonic_cost_landscape} and is based upon trigonometric interpolation, see Refs.~\cite{rotosolve3,lai2025optimalinterpolationbasedcoordinatedescent} and Sec.~\ref{app:cost_landscape} of the appendix. This follows a similar route to Rotosolve in Algorithm~\ref{alg:rotosolve}, except that the update rule requires solving for the set of stationary points of the cost function and choosing the one with the lowest cost. As with evaluating the gradient, however, this does not scale well because for each iteration of the algorithm we need to sample the cost function $2n+1$ times for an $n$-photon bosonic cost landscape, solve for the roots of a polynomial and sort through  these to find the one corresponding to the minimum. We ask whether we can avoid this.

Consider an observable of the form
\begin{equation}
H = \prod_{i = 1}^N \hat{n}_i^{p_i}, \label{eq:reduced_harmonic_observable}
\end{equation}
where $p_i \in \mathbb{N}$ and $\hat{n}_i= a^\dagger_i a_i$ is the number operator for the $i$th mode. It was shown that the maximum number of harmonics in the bosonic cost landscape generated by this observable is given by $R = \min \{p,n \}$, where $p = \sum_i p_i$ and $n$ is the number of bosons inserted into the interferometer~\cite{facelli2024exactgradientslinearoptics}. This result implies certain observables will have a simpler cost function, allowing for a more efficient evaluation of the gradient as only the first $2R$ terms of the generalised parameter shift rule in Eq.~\eqref{eq:parameter_shift_rule} are required. For some observables $R = 1$ and we can apply Rotosolve even if the system is bosonic.

Unfortunately, bosonic cost landscape generated by the QUBO Hamiltonian of Eq.~\eqref{eq:qubit_hamiltonian} has $R = n$ giving rise to the maximum number of harmonics in general for a given number of photons, in contrast to its fermionic counterpart which has $R = 1$ always. We prove this result here.
\begin{corollary}
Given a linear optical interferometer with threshold detectors and $n$ photons, then the bosonic cost landscape generated by the QUBO observable of Eq.~\eqref{eq:qubit_hamiltonian} contains $n$ harmonics in general.
\begin{proof}
When restricted to the subspace of at most $n$ photons, we can use polynomial interpolation to represent the threshold number operator $\Theta(\hat{n})$ as a polynomial of degree $n$ in the number operator as 
\begin{equation}
\Theta(\hat{n}) \equiv \sum_{k = 1}^n c_k \hat{n}^k.
\end{equation}
For example, if $n = 3$ then we can write $\Theta(\hat{n}) \equiv \frac{1}{6}\hat{n}(\hat{n}^2 - 6 \hat{n} + 11)$. Therefore, each quadratic term of the QUBO Hamiltonian, $\Theta(\hat{n}_i)\Theta(\hat{n}_j)$, will contribute the term $\hat{n}_i^{n} \hat{n}_j^n$, so $p = 2n$ and $R = \min \{ p, n\} = n$. By linearity, the sum of these quadratic terms will have $R = n$.
\end{proof}
\end{corollary}

\begin{figure}
\begin{center}
\includegraphics[width=0.49\columnwidth]{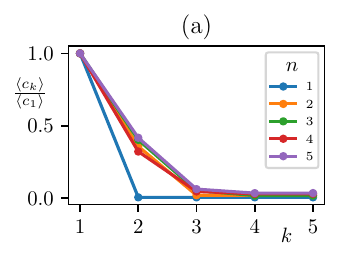}
\includegraphics[width=0.49\columnwidth]{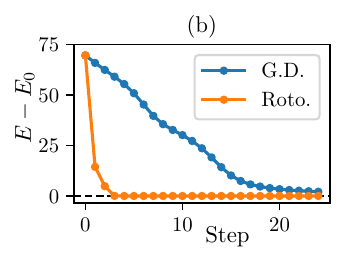}
\end{center}
\caption{(a) The ratio of the average Fourier coefficients $\langle c_k \rangle/\langle c_1 \rangle$ of the bosonic cost landscape of Eq.~\eqref{eq:bosonic_cost_landscape} for 20 randomly-chosen QUBO problems of size $N = 6$ for different boson numbers $n$. (b) An example of the cost function of a boson sampler over time using gradient descent (G.D.) or Rotosolve for a randomly-chosen QUBO problem of size $N = 6$ and $n = 4$ bosons.  \label{fig:bosonic_rotosolve}} 
\end{figure}

Despite this, it is seen numerically that the Fourier spectrum of the bosonic cost landscape for QUBO problems remains peaked around low frequencies, even as we scale up $n$. In Fig.~\ref{fig:bosonic_rotosolve}(a) we show the average magnitude of the Fourier coefficients of the cost function for 20 randomly chosen QUBO Hamiltonians versus the number of bosons in the system. We sample random symmetric QUBO matrices with integer matrix elements in the range $[-10,10]$ to insert into the cost function of Eq.~\eqref{eq:QUBO_cost}. The first harmonic is the dominant term, whilst the magnitude of the higher-order harmonics dies off quickly, even as we increase the number of bosons $n$ in the system. This suggests that Rotosolve, despite being designed for sinusoidal cost functions only, may work approximately for solving QUBO problems with boson sampling as the error will be small.

In Fig.~\ref{fig:bosonic_rotosolve}(b) we give an example of optimising the cost function for a QUBO problem size $N = 6$ with $n = 4$ photons inserted into a boson sampler. Despite having up to four harmonics, Rotosolve finds the minimum in a few steps compared to gradient descent which takes far longer. The improvement in performance over gradient descent could give boson sampling a faster and gradient-free way to get an approximate solution.

\section{Conclusion and discussion}

The central theme of this work was to investigate under what conditions we expect to see barren plateaus in linear optics, by providing both exact analytic results and numerics to confirm this. As we are free to vary the size of the interferometer and the number of particles independently, we do not always expect to see barren plateaus when scaling these variables. However, when both the number of modes and particles scales together, which is what one would typically define scaling up a linear optical quantum computer to mean, we do indeed see barren plateau scaling for both bosonic and fermionic linear optics, however the barren plateau decay for bosonic statistics is orders of magnitude faster than for fermionic statistics, implying that fermionic statistics is less susceptible.

Despite this, fermionic linear optics has some desirable features. Namely, it results in a simpler cost landscape consisting of a single Fourier mode of unit frequency per variational parameter in a way that is independent of the choice of input ansatz, circuit layout, and the observable to minimise, allowing us to bypass gradient descent entirely by using the gradient-free Rotosolve algorithm which has not been applied to interferometric systems until now. This results a significant speed-up over the best known interferometer-based photonic variational quantum algorithms, requires fewer cost function evaluations per iteration, and reduces the sample complexity. These improvements are obtained by exploiting quantum effects alone with no classical preprocessing required.

Motivated by this, we have shown that one can obtain fermionic cost landscapes with bosons by replacing each parametrised phase shifter of a linear optical interferometer with a non-linear phase shifter, the dual-valued phase shifter, which has two distinct eigenvalues. We provided three ways to construct this phase shifter, either by using non-linearities directly, measurement-induced non-linearities, or simulating fermion sampling with an entangled photonic resource state. The latter two designs require linear optics alone, allowing them to be constructed today. This allows one to exploit the benefits of fermionic cost landscapes using widely-available photonic linear optics.

On the other hand, it is doubtful whether the quantum advantage of boson sampling results in any \textit{practical} quantum advantage for solving problems with variational quantum algorithms in the way studied in this work. A boson sampler makes the cost landscape needlessly complicated, with barren plateaus and an algorithm whose complexity scales with the number of particles $n$. The fact that this algorithm can be defeated by simply changing the statistics of the particles from bosonic to fermionic whilst retaining identical hardware, resulting in a classically simulable algorithm with less severe barren plateaus and no scaling dependence on $n$, casts doubt on whether one should solve variational problems using this approach. This suggests that the quantum advantage of boson sampling is at odds with trainability in this setting. Despite this, we believe the techniques presented and the dual-valued phase shifter will be applicable beyond the variational algorithms setting.

We conclude with three interesting directions of further investigation. First, the advantage of our method is the ability to mould the cost landscape into something less barren in a way that is independent of the problem, ansatz or circuit design. This work applies to interferometric systems whose parametrised unitaries are phase shifters, but whether this can be generalised to systems whose parametrised unitaries are more general, such as those found in other variational algorithms such as QAOA~\cite{farhi2014quantumapproximateoptimizationalgorithm,BLEKOS20241}, the variational quantum eigensolver~\cite{peruzzo2014variational,TILLY20221} or quantum machine learning~\cite{PERALGARCIA2024100619}, remains an open question. 

Secondly, in section~\ref{sec:barren_plateau} we derived the barren plateau scaling of linear optics, but omitted an analytic derivation for linear optics with dual-valued phase shifters. If we extend linear optics to include these non-linear phase shifters, then the set of generators of the system is extended. In this work we derived the barren plateau scaling using the representation theory of $\mathrm{U}(N)$ by focussing on the Lie group only, as this was sufficient to get exact results. However, tackling the non-linear DVPS may require us to use Lie algebraic methods by looking at the dynamical Lie algebra instead, which is an alternative way to study the barren plateau problem~\cite{fontanacharacterizing, ragone2024lie}. This may be more fruitful for this endeavour as we would be looking at non-Gaussian unitaries instead.

Third, whilst we have shown that sampling bit strings from a boson sampler for the purpose of variational quantum algorithms does not have a practical advantage, it still has its \#P quantum advantage for sampling alone. Whether a boson sampler with all its phase shifters replaced by DVPSs retains this quantum advantage is unknown. It has been shown that introducing non-linear elements into a boson sampler can increase its complexity~\cite{spagnolo2023non}, however we leave these open problems to future work.

\acknowledgements
I thank Harry Bromley for helping me use Aegiq's Lightworks Python SDK~\cite{lightworks} for the numerical simulations of this work. I thank Callum W. Duncan for pointing out the two-eigenvalue generalisation beyond fermion sampling. I thank Pieter Kok for pointing out the martingale strategy to improve the non-linear DVPS and its link with the NS gate. Finally, I thank Scott Dufferwiel, Jiannis K. Pachos, Emma Horner, Alex Little, Maksym Sich and everyone mentioned above for valuable discussions.

\section*{Data availability}
The data that support the findings of this study were generated by numerical simulations. The source code and parameters used to generate the simulations is publicly available~\cite{data}.

\bibliography{refs.bib}

\onecolumngrid
\appendix

\section{Bosonic and fermionic linear optics \label{app:boson_fermion_sampling}} 

\subsection{Linear intereferometers}
Let us consider a set fermionic or bosonic creation and annihilation operators, $a_i^\dagger$ and $a_i$ respectively, acting on a Fock space $\mathcal{F}$, where $i$ labels the mode degree of freedom. These operators obey the algebra
\begin{equation}
[a_i,a_j]_\pm = [a_i^\dagger,a_j^\dagger]_\pm = 0, \quad [a_i,a_j^\dagger]_\pm = \delta_{ij},
\end{equation}
where $[A,B]_\pm = AB \pm BA$ is the anti-commutator for fermions $(+)$ and the commutator for bosons $(-)$. Given an $N$-mode system, multi-particle states are given by
\begin{equation}
|\mathbf{n}\rangle \equiv |n_1,n_2,\ldots,n_N\rangle = \prod_{i = 1}^N \frac{(a^\dagger_i)^{n_i}}{\sqrt{n_i!}} |0\rangle, \label{eq:number_state}
\end{equation}
where $n_i$ is the number of particles in the $i$th mode and $|0\rangle$ is the vacuum state. For bosons, the creation operators can be applied as many times as you like, so $n_i \in \mathbb{N}$. However, for fermions, the fermionic algebra implies that $(a_i^\dagger)^2 = 0$ hence $n_i \in \{ 0 , 1\} $, which is the Pauli exclusion principle. The set of states of this form are called number states or Fock states and form a canonical basis of the Fock space. The dimension of the subspaces $\mathcal{H}_\text{B}$ and $\mathcal{H}_\text{F} $ containing $n$ bosons or fermions and $N$ modes is respectively given by
\begin{equation}
\mathop{\mathrm{dim}}\mathcal{H}_\text{B} = \frac{(N+n-1)!}{(N-1)!n!},\quad \mathop{\mathrm{dim}} \mathcal{H}_\text{F} = \frac{N!}{(N-n)!n!}. \label{eq:fock_space_dimensions}
\end{equation}

In this study, we are interested in $N$-mode linear optical interferometers. These are devices constructed from a network of $N$ connected waveguides whose action on any input state is described by a unitary operator $U: \mathcal{F} \to \mathcal{F}$ acting linearly on the ladder operators in the Heisenberg picture as
\begin{equation}
U a_i^\dagger U^\dagger = \sum_{j = 1}^N u_{ji} a_j^\dagger, \label{eq:unitary}
\end{equation}
where $u \in \mathrm{U}(N)$ is an $N \times N$ unitary matrix which is in 1:1 correspondence with $U$. In fact, the unitary $U$ forms a reducible representation of the unitary matrix $u$ on the Fock space, where each irrep is an $n$-photon subspace. It is $u$ that is chosen, not $U$, when programming unitaries into an interferometers. We assume that the particles have identical internal state, such as frequency, polarisation and time bin, so the only degree of freedom available is the spatial degree of freedom of which waveguide the photons are on. From the Baker-Campbell-Hausdorff formula, this linear transformation implies that the unitaries are generated by quadratic Hamiltonians as $U = \exp(iHt)$, where
\begin{equation}
H = \sum_{i,j = 1}^N h_{ij} a_i^\dagger a_j.
\end{equation}
where the unitary matrix in Eq.~\eqref{eq:unitary} is given by exponentiating the single-particle Hamiltonian as $u = \exp(iht)$. Hence, a linear optical interferometer can be viewed as simulating time evolution under a non-interacting and particle conserving Hamiltonian.

For example, given a two-mode system with modes $a_0$ and $a_1$, a phase shifter that acts on the first mode is given by
\begin{equation}
U = \exp( - i \theta a_0^\dagger a_0 ) \quad \Leftrightarrow \quad u = \begin{pmatrix}
e^{-i \theta} & 0 \\ 0 & 1
\end{pmatrix},
\end{equation}
where $\theta \in [0,2 \pi)$ is the phase shift, whilst a beamsplitter is given by
\begin{equation}
U = \exp[ \theta (a^\dagger_0 a_1 - a_1^\dagger a_0) ] \quad \Leftrightarrow \quad u = \begin{pmatrix} \sqrt{T}  &  \sqrt{1-T} \\ -\sqrt{1-T} & \sqrt{T}  \end{pmatrix},
\end{equation}
where $T = \cos^2 \theta$ for $\theta \in [0,2 \pi)$ is the transmission of the beamsplitter. In both cases we have provided the Fock space unitary $U$ and the corresponding linear transformation $u$. 

In practice, these interferometers are constructed from a network of phase shifters and 50:50 beamsplitters only. There exist multiple universal interferometer layouts that can encode any $u \in \mathrm{U}(N)$, e.g.~\cite{PhysRevLett.73.58,Clements:16,Fldzhyan:20}, however this does not mean that the interferometer is a universal quantum computer, as only $U$ of the form in Eq.~\eqref{eq:unitary} can be generated and ones describing interactions between photons or particle non-conservation are not possible. For universal quantum computing with linear optics, a non-deterministic approached must be used~\cite{KLM,PhysRevA.65.012314,Kok_Lovett_2010}. 
 
\subsection{Sampling amplitudes}
Let us calculate the transition amplitude between the $n$-particle number states $|\mathbf{n}\rangle = |n_1,n_2,\ldots,n_N\rangle $ and $|\mathbf{m}\rangle = |m_1,m_2,\ldots,m_N\rangle$ after inserting $|\mathbf{n}\rangle$ into an $N$-mode linear interferometer, where $\sum_{i = 1}^N n_i = \sum_{i = 1}^N m_i = n$. We first rewrite these states in a slightly different way to how they are defined in Eq.~\eqref{eq:number_state} which makes calculation of the transition amplitudes easier. Let us define $\nu_i$ as the mode that the $i$th particle is on in the state $|\mathbf{n}\rangle$, for $i = 1,2,\ldots,n$, and similarly $\mu_i$ for $|\mathbf{m}\rangle$. Some of these indices may be identical if more than one particle is found on any of the modes. Then we can rewrite the number states as 
\begin{equation}
|\mathbf{n}\rangle = \frac{1}{\sqrt{\mathbf{n}!}} \prod_{i = 1}^n a_{\nu_i}^\dagger |0\rangle, \quad |\mathbf{m}\rangle = \frac{1}{\sqrt{\mathbf{m}!}} \prod_{i = 1}^n a_{\mu_i}^\dagger |0\rangle. \label{eq:number_state_2}
\end{equation}
where $\mathbf{n}! = \prod_{i = 1}^N n_i !$. We make sure the ladder operators are applied in the same order as in Eq.~\eqref{eq:number_state}, i.e., indices from largest to smallest from left to right, to avoid picking up any relative phases between the two ways of writing down the states. Let us apply the unitary $U$, whose action is defined in Eq.~\eqref{eq:unitary}, to $|\mathbf{n}\rangle$ which gives
\begin{equation}
\begin{aligned}
U |\mathbf{n}\rangle & = \frac{1}{\sqrt{\mathbf{n}!}} \prod_{i = 1}^n (U a_{\nu_i}^\dagger U^\dagger )|0\rangle \\ 
& = \frac{1}{\sqrt{\mathbf{n}!}} \prod_{i = 1}^n \left( \sum_{j = 1}^N u_{j \nu_i} a^\dagger_j \right) |0\rangle \\
& = \frac{1}{\sqrt{\mathbf{n}!}} \sum_{j_1 = 1}^N \sum_{j_2 = 1}^N \ldots \sum_{j_n = 1}^N u_{j_1 \nu_1}  u_{j_2 \nu_2}   \ldots u_{j_n \nu_n} a^\dagger_{j_1} a^\dagger_{j_2} \ldots a^\dagger_{j_n} |0\rangle
\end{aligned}
\end{equation}
Now if we take the inner product with $|\mathbf{m}\rangle$ we get
\begin{equation}
\langle \mathbf{m}| U |\mathbf{n}\rangle  = \frac{1}{\sqrt{\mathbf{n}!}} \sum_{j_1 = 1}^N \sum_{j_2 = 1}^N \ldots \sum_{j_n = 1}^N u_{j_1 \nu_1}  u_{j_2 \nu_2}   \ldots u_{j_n \nu_n} \langle \mathbf{m}| a^\dagger_{j_1} a^\dagger_{j_2} \ldots a^\dagger_{j_n} |0\rangle
\end{equation}
If we substitute in Eq.~\eqref{eq:number_state_2} and apply Wick's theorem we have
\begin{equation}
\begin{aligned}
\langle \mathbf{m}| a^\dagger_{j_1} a^\dagger_{j_2} \ldots a^\dagger_{j_n}|0\rangle 
& = \frac{1}{\sqrt{\mathbf{m}!}} \langle 0| a_{\mu_1} a_{\mu_2} \ldots a_{\mu_n} a^\dagger_{j_1} a^\dagger_{j_2} \ldots a^\dagger_{j_n}|0\rangle \\
& = \frac{1}{\sqrt{\mathbf{m}!}} \sum_{\sigma \in S_n} (\pm 1)^\sigma \prod_{i = 1}^n \delta_{\mu_{\sigma(i)},j_i}
\end{aligned}
\end{equation}
where Wick's theorem amounts to repeated use of the commutation relations and gives us a sum over all possible permutations, where $S_n$ is the permutation group of $n$ elements, and we additionally have a factor of $(\pm 1)^\sigma$ which arises from commuting operators past each and picking up a sign if they have fermionic anti-commutation relations. Substituting this back into the inner product yields 
\begin{equation}
\begin{aligned}
\langle \mathbf{m}| U |\mathbf{n}\rangle  & = \frac{1}{\sqrt{\mathbf{m}! \mathbf{n}! }} \sum_{j_1 = 1}^N \sum_{j_2 = 1}^N \ldots \sum_{j_n = 1}^N u_{j_1 \nu_1}  u_{j_2 \nu_2}   \ldots u_{j_n \nu_n} \sum_{\sigma \in S_n} (\pm 1)^\sigma \prod_{i = 1}^n \delta_{\mu_{\sigma(i)} j_i} \\
& =\frac{1}{\sqrt{\mathbf{m}! \mathbf{n}!}}  \sum_{\sigma \in S_n}  (\pm 1)^\sigma  \prod_{i = 1}^n \sum_{j_i = 1}^N u_{j_i \nu_i}    \delta_{\mu_{\sigma(i)} j_i} \\
& = \frac{1}{\sqrt{\mathbf{m}! \mathbf{n}! }} \sum_{ \sigma \in S_n} (\pm 1)^\sigma \prod_{i = 1}^n u_{\mu_{\sigma(i)} \nu_i} \\
& = \frac{1}{\sqrt{\mathbf{m}! \mathbf{n}!}} \begin{cases}
\operatorname{per} u[\mathbf{m}|\mathbf{n}] & \text{bosons ($+$)} \\
\det u[\mathbf{m}|\mathbf{n}] & \text{fermions ($-$)}
\end{cases},
\end{aligned} \label{eq:app_transition_amplitude}
\end{equation}
where $\operatorname{per}$ is the permanent, $\det$ is the determinant, and $u[\mathbf{m}|\mathbf{n}]$ is the matrix $u$ with the $i$th row index repeated $m_i$ times and the $j$th column index repeated $n_j$ times~\cite{10.1145/1993636.1993682,scheel2004permanents}. For example, 
\begin{equation}
u = \begin{pmatrix}
u_{11} & u_{12} & u_{13} \\
u_{21} & u_{22} & u_{23} \\
u_{31} & u_{32} & u_{33}
\end{pmatrix} \quad \Rightarrow \quad u[(3,1,0)|(1,1,2)] = \begin{pmatrix}
u_{11} & u_{12} & u_{13} & u_{13} \\
u_{11} & u_{12} & u_{13} & u_{13} \\
u_{11} & u_{12} & u_{13} & u_{13} \\
u_{21} & u_{22} & u_{23} & u_{23}
\end{pmatrix}.
\end{equation}
The process of inserting bosons/photons into a linear optical interferometer and sampling from the output states is known as boson sampling. As calculating the permanent is a \#P-hard problem classically, then simulating the output distribution of boson sampling affords a quantum advantage~\cite{10.1145/1993636.1993682}. On the other hand, the determinant can be calculated efficiently using an algorithm such as LU decomposition which has a complexity of $O(n^3)$ for an $n \times n$ matrix, hence simulation of fermions is classically efficient. This determinant is sometimes called the Slater determinant.

\subsection{Simulating fermions with single photons and linear optics \label{app:fermionic_simulation}}

The following reviews the results of Ref.~\cite{Matthews_2013}. We wish to simulate fermionic statistics with photons. As photons are bosons, we must use entangled states to simulate this. Suppose we have $n$ photons that we insert into $n$ separate interferometers, each with $N$ modes. Define the state
\begin{equation}
|\psi\rangle = \frac{1}{\sqrt{n!}}\sum_{\sigma \in S_n} (-1)^\sigma \prod_{\mu =1}^n  a^{ \dagger}_{ \mu \sigma(\mu) } |0\rangle, \label{eq:fermionic_photons}
\end{equation}
where $a_{\mu i} $ is the mode operator for the $i$th mode of the $\mu$th interferometer. These modes obey the commutation relations
\begin{equation}
[a_{\mu i}, a_{\nu j} ] = [a_{\mu i}^\dagger, a_{ \nu j}^\dagger ] = 0, \quad [a_{\mu i}, a_{\nu j}^\dagger ] = \delta_{\mu \nu}\delta_{ij},
\end{equation} 
so operators from different interferometers will always commute. The state defined in  Eq.~\eqref{eq:fermionic_photons} consists of a single photon in one of the first $n$ modes of each interferometer. This generates an entangled state. Now we pass this state through the system. Each photon enters a separate interferometer, so there will be no interference between photons and the only effects are due to the entanglement and interference of photons with themselves. 

We now define the unitary $U$ which acts on these modes as
\begin{equation}
U a^{ \dagger}_{\mu i} U^\dagger = \sum_{j = 1}^N u_{ji} a^{ \dagger }_{\mu j}, \quad \forall \mu.
\end{equation}
This applies a linear transformation to the mode degree of freedom $i$ and not the interferometer degree of freedom $\mu$ which remains unchanged. Therefore, this unitary $U$ describes describes the process of applying the same unitary operation to each photon within its respective interferometer. If we apply this to our input state we have
\begin{equation}
\begin{aligned}
U|\psi\rangle & = \frac{1}{\sqrt{n!}} \sum_{\sigma \in S_n} (-1)^\sigma \prod_{\mu=1}^n U a^{ \dagger}_{\mu \sigma(\mu)} U^\dagger |0\rangle \\
& = \frac{1}{\sqrt{n!}} \sum_{\sigma \in S_n} (-1)^\sigma \prod_{\mu=1}^n \sum_{i = 1}^N u_{i \sigma(\mu)} a^{ \dagger}_{\mu i} |0\rangle \\
& = \frac{1}{\sqrt{n!}} \sum_{\sigma \in S_n} (-1)^\sigma \sum_{i_1 = 1}^N \sum_{i_2 = 1}^N \ldots \sum_{i_n = 1}^N u_{i_1 \sigma(1)} u_{i_2 \sigma(2)} \ldots u_{i_n \sigma(n)} a^{ \dagger}_{1 i_1} a^{ \dagger}_{2 i_2} \ldots a^{ \dagger}_{n i_n} |0\rangle .
\end{aligned}
\end{equation}
Now introduce an $n$-photon state, where each interferometer has a single photon in it, given by
\begin{equation}
|\phi_\mathbf{k}\rangle = \prod_{\mu = 1}^n   a^{ \dagger}_{\mu k_\mu} |0\rangle ,
\end{equation}
which is the state for which the $\mu$th interferometer has a single photon in the $k_\mu$th mode. The amplitude for the input state to be found in the state $|\phi_\mathbf{k}\rangle$ is given by
\begin{equation}
\begin{aligned}
\langle \phi_\mathbf{k} | U |\psi\rangle & =\frac{1}{\sqrt{n!}} \sum_{\sigma \in S_n} (-1)^\sigma \sum_{i_1 = 1}^N \sum_{i_2 = 1}^N \ldots \sum_{i_n = 1}^N u_{i_1 \sigma(1)} u_{i_2 \sigma(2)} \ldots u_{i_n \sigma(n)} \langle \phi_\mathbf{k}| a^{ \dagger}_{1 i_1} a^{ \dagger}_{2 i_2} \ldots a^{ \dagger}_{n i_n} |0\rangle  \\
& = \frac{1}{\sqrt{n!}} \sum_{\sigma \in S_n} (-1)^\sigma \sum_{i_1 = 1}^N \sum_{i_2 = 1}^N \ldots \sum_{i_n = 1}^N u_{i_1 \sigma(1)} u_{i_2 \sigma(2)} \ldots u_{i_n \sigma(n)} \delta_{k_1 i_1} \delta_{k_2 i_2} \ldots \delta_{k_n i_n} \\
& = \frac{1}{\sqrt{n!}} \sum_{\sigma \in S_n} (-1)^\sigma  u_{k_1 \sigma(1)} u_{k_2 \sigma(2)} \ldots u_{k_n \sigma(n)} \\
& = \frac{1}{\sqrt{n!}} \det u_\mathbf{k},
\end{aligned}
\end{equation}
where $u_\mathbf{k}$ is an $n \times n$ matrix with components
\begin{equation}
[u_{\mathbf{k}}]_{ij} = u_{k_i j}, 
\end{equation}
where $u$ is the original unitary matrix. Suppose we ask for the probability that two interferometers have a photon on the same mode, say interferometers $\mu$ and $\nu$, in which case $k_\mu = k_\nu$. We see that this amplitude vanishes because the matrix $u_\mathbf{k}$ would have two identical rows and hence its determinant is zero.

Suppose instead we ask: what is the probability to get the distribution of photons $\mathbf{m} = (m_1,m_2,\ldots,m_N)$, where $m_i = 1$ if there is a photon in the $i$th mode \textit{somewhere} across the $N$ chips, in other words what is the probability that the superimposed output of all $n$ interferometers is $\mathbf{m}$? This means we now ignore the interferometer degree of freedom $\mu$ so multiple output distributions will be equivalent. In this case we sum up the probabilities for all permutations of $\mathbf{k}$ that yield this particular $\mathbf{m}$. We have
\begin{equation}
\begin{aligned}
P(\mathbf{m})  =  \frac{1}{n!}\sum_{\sigma \in S_n}  \left | \det u_{\sigma(\mathbf{k})} \right|^2  = \frac{1}{n!} \sum_{\sigma \in S_n}  \left | (-1)^\sigma \det u_{\mathbf{k}} \right|^2  = \left| \det u_{\mathbf{k}} \right|^2,
\end{aligned} 
\end{equation}
however we have
\begin{equation}
\left|\det u_{\mathbf{k}} \right| = \left|\det u[\mathbf{m}|\mathbf{n}]\right|
\end{equation}
which is because $k_\mu$ are simply the indices for which $m_{k_\mu} = 1$, which therefore allows us to rewrite this in terms of the matrix $u[\mathbf{m}|\mathbf{n}]$ as defined in Eq.~\eqref{eq:app_transition_amplitude} in the fermionic case. Therefore, the entangled photonic state gives rise to fermionic statistics. 

The success probability of producing the fermionic resource state of Eq.~\eqref{eq:fermionic_photons} with linear optics alone is upper bounded by $1/9$. This is because for the smallest non-trivial example of $n = 2$ fermions a single CNOT gate is required to produce the state which has a probability of $1/9$ with linear optics and postselection, see Sec.~\ref{sec:fermion_sampling} for this example. For larger $n$, more CNOTs are required, reducing the probability. In the supplementary material of Ref.~\cite{Matthews_2013} a three-fermion example is shown.

\section{Cost landscapes of linear optics \label{app:cost_landscape}}

\subsection{Derivation of the cost landscape \label{app:cost_landscape_proof}}

As discussed above in Sec.~\ref{sec:vqe} the interferometer is constructed from an array of parametrised phase shifters and fixed 50:50 beamsplitters. In this work, we introduce a \textit{non-linear} phase shifter which acts slightly differently to how standard phase shifters act, which we refer to a dual-valued phase shifter (DVPS), and we compare the the performance of variational quantum algorithms which use interferometers constructed from these types. The two types of phase shifters are described by the unitaries
\begin{equation}
U(x) = \begin{cases}
\exp(i \hat{n} x) & \text{standard} \\
\exp(i \hat{q} x) & \text{DVPS}
\end{cases} \label{eq:phase_shifters}
\end{equation}
where $\hat{n} = a^\dagger a$ is the number operator for the mode that the phase shifter acts upon, whilst the DVPS is generated by a new Hermitian operator $\hat{q}$ which is diagonal in the number basis, just like the number operator is, but instead has only two distinct eigenvalues. In other words, it acts as 
\begin{equation}
\hat{q}|n\rangle = q(n)|n\rangle,
\end{equation} 
where $q: \mathbb{N} \to \{ a, b \} $.

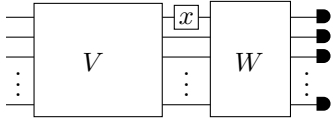
\begin{figure}
\begin{center}

\begin{tikzpicture}[x=0.75pt,y=0.75pt,yscale=-1,xscale=1]

\draw   (64,49) -- (128,49) -- (128,106) -- (64,106) -- cycle ;
\draw   (152,48) -- (192,48) -- (192,106) -- (152,106) -- cycle ;
\draw   (134,50) -- (146,50) -- (146,62) -- (134,62) -- cycle ;
\draw    (128,56) -- (134,56) ;
\draw    (146,56) -- (152,56) ;
\draw    (128,66) -- (152,66) ;
\draw    (128,76) -- (152,76) ;
\draw    (192,56) -- (206,56) ;
\draw    (192,66) -- (206,66) ;
\draw    (192,76) -- (206,76) ;
\draw  [fill={rgb, 255:red, 0; green, 0; blue, 0 }  ,fill opacity=1, line width=1.5] (206,53) -- (208.5,53) .. controls (209.88,53) and (211,54.12) .. (211,55.5) .. controls (211,56.88) and (209.88,58) .. (208.5,58) -- (206,58) -- cycle ;
\draw [fill={rgb, 255:red, 0; green, 0; blue, 0 }  ,fill opacity=1, line width=1.5]  (206,63) -- (208.5,63) .. controls (209.88,63) and (211,64.12) .. (211,65.5) .. controls (211,66.88) and (209.88,68) .. (208.5,68) -- (206,68) -- cycle ;
\draw [fill={rgb, 255:red, 0; green, 0; blue, 0 }  ,fill opacity=1, line width=1.5]  (206,73) -- (208.5,73) .. controls (209.88,73) and (211,74.12) .. (211,75.5) .. controls (211,76.88) and (209.88,78) .. (208.5,78) -- (206,78) -- cycle ;
\draw    (50,56) -- (64,56) ;
\draw    (50,66) -- (64,66) ;
\draw    (50,76) -- (64,76) ;
\draw    (50,100) -- (64,100) ;
\draw    (128,100) -- (152,100) ;
\draw    (192,100) -- (206,100) ;
\draw [fill={rgb, 255:red, 0; green, 0; blue, 0 }  ,fill opacity=1, line width=1.5]  (206,97) -- (208.5,97) .. controls (209.88,97) and (211,98.12) .. (211,99.5) .. controls (211,100.88) and (209.88,102) .. (208.5,102) -- (206,102) -- cycle ;

\draw (87,72.4) node [anchor=north west][inner sep=0.75pt]    {$V$};
\draw (163,72.4) node [anchor=north west][inner sep=0.75pt]    {$W$};
\draw (135,52.4) node [anchor=north west][inner sep=0.75pt]    {$x$};
\draw (137,74.4) node [anchor=north west][inner sep=0.75pt]    {$\vdots $};
\draw (197,74.4) node [anchor=north west][inner sep=0.75pt]    {$\vdots $};
\draw (53,74.4) node [anchor=north west][inner sep=0.75pt]    {$\vdots $};
\end{tikzpicture}

\end{center}
\caption{When we vary a single phase shifter in the interferometer, we effectively have an interferometer consisting of a single phase shifter with phase $x$ sandwiched between two (non-universal in general) interferometers encoding the fixed unitaries $V$ and $W$. \label{fig:varying_phase_shifter}}
\end{figure}

We now investigate how the choice of $\hat{n}$ or $\hat{q}$ modifies the cost landscape. This closely follows the calculatios of Ref.~\cite{gan2022fock}. Suppose we  introduce the Hamiltonian $H$ whose (possibly degenerate) ground state corresponds to the solution to a given problem. We introduce the cost function
\begin{equation}
E(\boldsymbol{\theta}) = \langle \psi(\boldsymbol{\theta})| H |\psi(\boldsymbol{\theta})\rangle,
\end{equation}
where $|\psi(\boldsymbol{\theta}) \rangle = U(\boldsymbol{\theta}) |\psi_\text{in}\rangle$ is the output of the interferometer described by the parametrised unitary $U(\boldsymbol{\theta})$. Now suppose we vary one of the parameters of the interferometer, say the $j$th one, whilst keeping the others fixed, then the cost function can be viewed as the function
\begin{equation}
f(x) :=  E(\theta_j = x) = \langle \psi_\mathrm{in} | U^\dagger(x) H U(x) |\psi_\mathrm{in}\rangle.
\end{equation}
Each variable of the unitary controls a single phase shifter in the interferometer, so here the unitary $U$ takes the form $
U(x) = WS(x)V$, where $S(x)$ is the phase shifter from Eq.~\eqref{eq:phase_shifters}; whilst $V$ and $W$ are the unitaries describing the remainder of the interferometer before and after this phase shifter respectively as shown in Fig.~\ref{fig:varying_phase_shifter}. 

As the unitaries conserve particle number, we can restrict ourselves to the $n$-particle subspace. Let us consider two $n$-particle states $|\mathbf{n}\rangle$ and $|\mathbf{m}\rangle$. The matrix elements on this subspace are given by
\begin{equation}
\begin{aligned}
U_{\mathbf{m} \mathbf{n}} \equiv \langle \mathbf{m}| U |\mathbf{n}\rangle   & = \sum_\mathbf{p,q} W_\mathbf{mp} S_\mathbf{pq}(x) V_{\mathbf{qn}} \\
& =  \sum_\mathbf{p,q} W_\mathbf{mp} V_\mathbf{qn} e^{i p_j x} \delta_\mathbf{pq}\\
& =  \sum_\mathbf{p} W_\mathbf{mp} V_\mathbf{pn} e^{i p_j x},
\end{aligned}
\end{equation}
where the sums are over the $n$-particle Fock number basis. 

Suppose we took our initial state as $|\psi_\mathrm{in}\rangle = |\mathbf{n}\rangle$ and pass it through the interferometer, the expectation value of the cost function gives
\begin{equation}
f(x)  = \langle \mathbf{n} |U^\dagger H U |\mathbf{n} \rangle
= \sum_\mathbf{p,q} U^\dagger_\mathbf{np} H_\mathbf{pq} U_\mathbf{qn}.
\end{equation}
We can safely assume that the Hamiltonian is diagonal in the Fock basis as $H_{\mathbf{pq}} = E_\mathbf{p} \delta_{\mathbf{pq}}$. Alternatively, we could diagonalise $H$ to this form with a unitary which we absorb into the definition of $U$. We can also encode the use of threshold detectors, which are unable to count the number of photons in each mode, by taking the eigenvalues of the Hamiltonian to be identical for all states that map to the same bit string under threshold detection, i.e., $E_\mathbf{p} = E_\mathbf{q}$ if $\mathbf{p} = \Theta(\mathbf{q})$ where $\Theta$ is the Heaviside step function that acts on each element of the vector. We could also map to bit strings by using parity photo-detectors that can detect whether there was an even or odd number of photons too~\cite{bradler2021certain} and an alternative parity encoding on the Hamiltonian is used be used for this case. 

 Using this and substituting in the matrix elements of $U$ we have the cost function
\begin{equation}
\begin{aligned}
f(x) & =\sum_\mathbf{p,q} \left( \sum_\mathbf{k} W^*_\mathbf{pk} V^*_\mathbf{kn} e^{-i k_j x}\right) E_\mathbf{p} \delta_\mathbf{pq} \left(\sum_\mathbf{l} W_\mathbf{ql} V_\mathbf{ln} e^{i l_j x} \right) \\
& = \sum_\mathbf{k,l}  \left(\sum_\mathbf{p} E_\mathbf{p} W^*_\mathbf{pk} V^*_\mathbf{kn} W_\mathbf{pl} V_\mathbf{ln} \right) e^{i(l_j-k_j)x} \\
& \equiv \sum_\mathbf{k,l} a_\mathbf{kl} e^{i(l_j-k_j)x},
\end{aligned}
\end{equation}
This is beginning to look like a Fourier series. Let us define the frequency $p = l_j - k_j$. If we work with standard phase shifters then on the $n$-particle subspace the set of frequencies is given by $p \in \{ -n,\ldots,n \}$. This is because $l_j$ and $k_j$, which are both the possible eigenvalues of the number operator, can take all integer values from $0$ to a maximum of $n$ as we sum over the number states, hence $p = l_j - k_j$ takes all values from $-n$ to $n$.

Now let us reorder the sum by combining all the coefficients of each exponential with the same frequency $p$ to give us
\begin{equation}
f(x) = \sum_{p = -n}^n c_p e^{i p x}, \quad c_p = \sum_{\substack{ \mathbf{k,l} \\ l_j - k_j = p}} a_\mathbf{kl}, \label{eq:fourier_truncated}
\end{equation}
which is our final result. This result was originally from Ref.~\cite{gan2022fock}.  This means that if we vary a single parameter of the cost function, the resultant function takes the general form of a Fourier series with a maximum frequency of $n$, where $n$ is the number of bosons inserted into the interferometer.

\begin{figure}
\begin{center}
\includegraphics[scale=0.75]{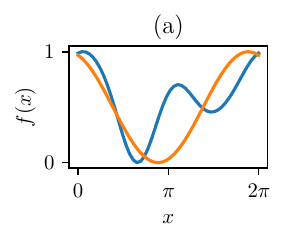}
\includegraphics[scale=0.75]{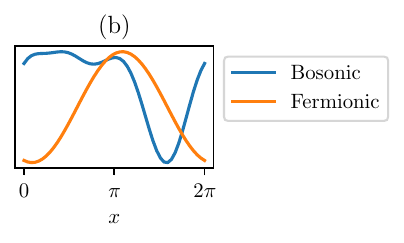}
\caption{The normalised cost function for a random Hamiltonian $H$ if we use bosonic or fermionic statistics as we vary a single phase shifter $x$ in the interferometer. We see that fermionic statistics can remove or reduce (a) local minima and (b) barren plateaus . \label{fig:cost_function_examples}}
\end{center}
\end{figure}

If instead we work with dual-valued phase shifters, then the only possible frequencies are given by $p \in \{0,a-b,b-a\}$. This is because the eigenvalues $l_j$ and $k_j$ can only take values $a$ or $b$. The actual values of the eigenvalues we assign each Fock state is not important, only that the Fock states are eigenstates with two possible eigenvalues. This means that we only have a single frequency in our resultant Fourier series, so Eq.~\eqref{eq:fourier_truncated} is truncated down further to
\begin{equation}
f(x) = A \sin(\omega x - \phi) + B, \label{eq:two_eigenvalues_cost_landscape}
\end{equation}
where $\omega = |a-b|$, and $A$, $B$ and $\phi$ are constants determined by the rest of the parameters in the interferometer that control $V$ and $W$ and the observable $H$. Similarly, if we worked with standard phase shifters and \textit{fermionic} states, for which $\hat{n} = 0,1$ automatically, we would obtain the same cost landscape for $\omega = 1$ without the need to use $\hat{q}$ explicitly. In Fig.~\ref{fig:cost_function_examples} we show how transforming from bosonic to fermionic statistics removes local minima and barren plateaus for the same cost function Hamiltonian and unitaries $V$ and $W$.

\subsection{Solving for the minima \label{app:stationary_points}}

In this appendix we solve for the stationary points of the bosonic cost landscape of Eq.~\eqref{eq:fourier_truncated} which forms the basis of the Optimal Interpolation-based Coordinate Descent (OICD)~\cite{lai2025optimalinterpolationbasedcoordinatedescent} algorithm that generalises Rotosolve to trigonometric cost functions with multiple harmonics. This closely follows the original works of Refs.~\cite{lai2025optimalinterpolationbasedcoordinatedescent,boyd2006computing}. 

Firstly, we can obtain the Fourier coefficients of an unknown trigonometric cost function of the form of Eq.~\eqref{eq:fourier_truncated} given a finite set of samples. Let us take $2n+1$ samples of the cost function at the set of equally-spaced points $x_j = 2j \pi /(2n+1)$ as
\begin{equation}
f(x_j) = \sum_{k = -n}^n c_k e^{ikx_j},
\end{equation}
for $j = 1,2, \ldots, 2n+1$. This takes the form of a discrete Fourier transform. By perfoming an inverse discrete Fourier transform we arrive at
\begin{equation}
c_k = \frac{1}{2n+1} \sum_{j = -n}^n f(x_j) e^{-ikx_j}
\end{equation}
which gives us the coefficients.

After solving for the Fourier coefficients $\{ c_k \}$ the exact form of the Fourier series is known and the the minima can be solved for. We have the gradient
\begin{equation}
f'(x) = \sum_{k = -n}^n d_k e^{ikx},
\end{equation}
where $d_k = ikc_k$, so the stationary points of the cost function $f(x)$ are given by the roots of a trigonometric polynomial as $f'(x) = 0$ which we solve for. Let us extend the domain of $f'(x)$ by defining the Laurent polynomial
\begin{equation}
p(z) = \sum_{k = -n}^n d_k z^k.
\end{equation}
The restriction of $p(z)$ to the unit circle returns the original trigonometric polynomial as $f'(x) = p(e^{ix})$. The roots of of $p(z)$ such that $|z| = 1$ correspond to the roots of $f'(x)$, where the relationship is given by $x = \mathrm{arg}(z)$. In order to solve for the roots, define a $2n$-degree polynomial $q(z) = z^n p(z)$. This polynomial has the same roots as $p(z)$ and can be solved for using simple root finding numerical methods. By the fundamental theorem of algebra, the polynomial $q(z)$ will have $2n$ roots that lie in the complex plane and a subset of these may lie on the unit circle corresponding to the roots of the original problem. In Fig.~\ref{fig:fourier_roots} we show an example of this in action. To find the global minima we sort through the list of stationary points on the unit circle to find the stationary point of $f(x)$ with the smallest cost.
 
This also upper bounds the number of minima of $f(x)$. Not all of the roots of $q(z)$ will lie on the unit circle in general, as seen in Fig.~\ref{fig:fourier_roots} for example, so all we can say is that there are at most $2n$ roots that lie on the unit circle and therefore at most $2n$ stationary points of $f(x)$. As the number of minima and maxima are equal, so the number of minima is half of this giving us the upper bound on the number of local minima of $n$.

\begin{figure}[t]
\begin{center}
\includegraphics[scale=0.5]{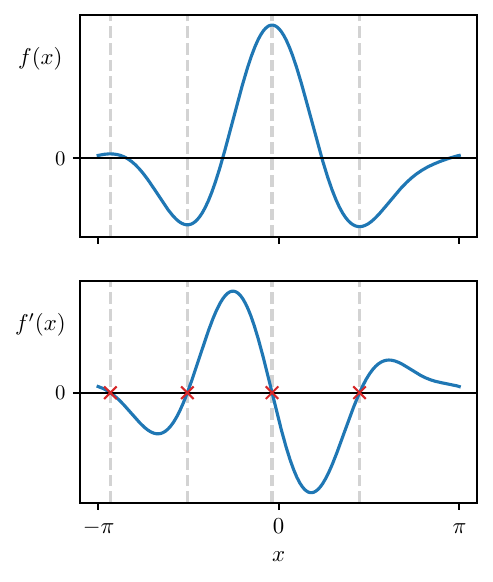}
\includegraphics[scale=0.75]{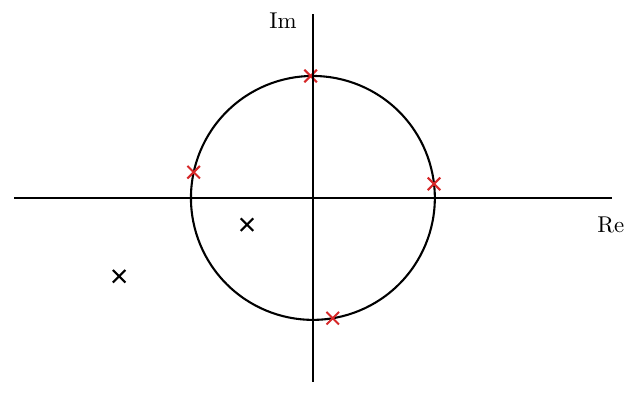}
\end{center}
\caption{The stationary points of a trigonometric polynomial $f(x)$ of order $n$ can be found by solving for the roots of a polynomial whose coefficients are equal to the Fourier coefficients of the derivative $f'(x)$. Roots $z$ of the polynomial that lie on the unit circle in the complex plane, shown by the red crosses, correspond to roots of the derivative and hence stationary points via $x = \mathrm{arg}(z)$. Other roots, shown by black crosses, do not correspond to stationary points. \label{fig:fourier_roots}}
\end{figure}

\section{Derivations for the dual-valued phase shifter}

\subsection{Deterministic design with non-linear optics \label{app:non_linear_DVPS}}

\begin{figure}[h]

\begin{center}

\includegraphics[scale=1]{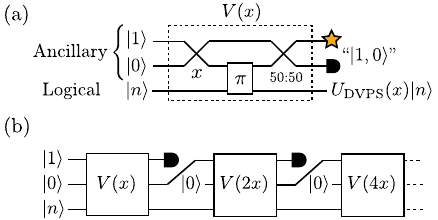}
\caption{(a) A non-deterministic DVPS. (b) A repeat-until-success variant. \label{fig:appendix_DVPS}}
\end{center}
\end{figure}

A possible realisation of the dual-valued phase shifter (DVPS) is to take 
\begin{equation}
U_\text{DVPS}(x) = e^{i \hat{q} x}, \quad \hat{q} = \frac{1}{2}\left( 1 - \hat{\pi} \right),
\end{equation}
where $\hat{\pi} = \exp(i\pi \hat{n})$ is the parity operator, equivalent to a $\pi$-phase shifter, and $\hat{n} = a^\dagger a$ is the number operator. This operator acts as
\begin{equation}
e^{i\hat{q}x} |n\rangle = e^{iq(n)x} |n\rangle, \quad q(n) = \frac{1}{2}(1 - (-1)^n),
\end{equation}
where $|n\rangle$ is a single-mode number state of $n$ photons. The goal is to construct this unitary.

Consider the circuit layout as given in Fig.~\ref{fig:appendix_DVPS} which consists of three modes $a_0$, $a_1$ and $a_2$ labelled from top to bottom, where $a_0$ and $a_1$ are ancillary modes and $a_2$ is the logical mode. The beamsplitters act on the space of modes linearly according to Eq.~\eqref{eq:unitary} as
\begin{equation}
u_\text{BS}(x) = \begin{pmatrix}
\cos \frac{x}{2} & -i \sin \frac{x}{2} & 0 \\ -i \sin \frac{x}{2} & \cos \frac{x}{2} & 0 \\ 0 & 0 & 1
\end{pmatrix}, \quad u_{\text{50:50}} = \begin{pmatrix}
\frac{1}{\sqrt{2}} & \frac{1}{\sqrt{2}} & 0 \\ \frac{1}{\sqrt{2}} & - \frac{1}{\sqrt{2}} & 0 \\ 0 & 0 & 1 
\end{pmatrix}
\end{equation}
whilst the cross Kerr non-linearity acts as $U_\mathrm{K} = e^{i \phi \hat{n}_1 \hat{n}_2}$ and has no matrix representation on the space of modes due to it being non-linear.

If we take the input state as $|\psi_\text{in}\rangle = |1,0,n\rangle =  a_0^\dagger (a^\dagger_2)^n |0\rangle/\sqrt{n!} $
then using the linear transformation rule from Eq.~\eqref{eq:unitary} for the ladder operators we see tunable beamsplitter mixes ancillary modes only to give us
\begin{equation}
|\psi_1\rangle =  \left(  \cos \frac{ x}{2}  |1,0\rangle - i \sin \frac{ x}{2}  |0,1\rangle \right) |n\rangle
\end{equation}
The cross Kerr non-linearity acts between $1$ and $2$ and applies a phase of $e^{in \phi}$ only if mode $1$ has a photon in it, so we get
\begin{equation}
|\psi_2\rangle  =  \left( \cos \frac{ x}{2}   |1,0 \rangle - i e^{i n \phi} \sin  \frac{ x}{2}  |0,1\rangle \right) |n\rangle.
\end{equation}
The second beamsplitter mixes modes $0$ and $1$ only to give
\begin{equation}
\begin{aligned}
|\psi_3\rangle & = \frac{1}{\sqrt{2}} \left[ \cos \frac{ x}{2}   \left( |1,0\rangle + |0,1\rangle \right) - i e^{i n \phi} \sin  \frac{ x}{2}  \left( |1,0\rangle - |0,1 \rangle \right) \right] |n\rangle \\
& = \frac{1}{\sqrt{2}} \left[ \left( \cos  \frac{ x}{2}  - i e^{in\phi} \sin  \frac{ x}{2}  \right) |1,0 \rangle +  \left( \cos \frac{ x}{2}  + i e^{i n \phi} \sin  \frac{ x}{2}  \right) |0,1\rangle \right] |n\rangle.
\end{aligned}
\end{equation}
If we postselect on the states for which the top two modes are in the state $|1,0\rangle$, then we have the total output state $P_{10}| \psi_3 \rangle/\sqrt{\langle \psi_3 | P_{10} |\psi_3 \rangle}$, where $P_{10} = |1,0\rangle \langle 1,0| \otimes \mathbb{I}$ is the projector. We have
\begin{equation}
\begin{aligned}
\langle \psi_3 | P_{10} |\psi_3 \rangle & = \frac{1}{2} \left| \cos \frac{ x}{2} - i e^{in\phi} \sin \frac{ x}{2} \right|^2  \\
& = \frac{1}{2} \left[1 + \sin(x)\sin(n\phi) \right],
\end{aligned}
\end{equation}
which is also the success probability of this non-deterministic gate. 

Now we are interested in the case where the Kerr non-linearity parameter is $\phi = \pi$ and $\theta = x/2$, in which case  we get $\langle \psi_4 |P_{10}|\psi_4 \rangle = 1/2$ so the success probability is $1/2$ and the output state of the logical mode is given by
\begin{equation}
\begin{aligned}
|\psi_\text{out}\rangle & =  \left( \cos \frac{ x}{2} - i(-1)^n \sin \frac{ x}{2} \right)|n\rangle \\
& =  e^{-i \frac{x}{2} (-1)^n} |n\rangle  \\
& \equiv e^{iq(n)x} |n\rangle 
\end{aligned}
\end{equation}
where the last line means equivalence up to a global phase (i.e. a phase that does not depend on $n$ and given by $e^{-ix/2}$), which is the desired phase of a DVPS. This concludes the proof of the heralded non-deterministic DVPS with a success probability of $1/2$.  

If the gate fails, as heralded by the ancillary output state $|0,1\rangle$, then the output state is given by 
\begin{equation}
|\psi_\text{out,fail}\rangle = e^{+i \frac{x}{2} (-1)^n} |n\rangle \equiv e^{- iq(n)x} |n\rangle
\end{equation}
which is the complex conjugate of what we want, up to a global phase. We can apply a second DVPS to the output of the failed gate, except with a phase of $2x$ now. If this second attempt is successful it will correct the incorrect phase of the first attempt. We repeat this procedure until success by taking the phase of the $n$th iteration as $x_n = 2^{n-1}x$ which gives a success probability after $n$ iterations of $1 - 1/2^n$. This is known as the martingale strategy. To do this, we only need to measure the top ancillary mode to tell whether the gate was a success or not, so if no photon is present we know it exited on the lower mode and this photon can be reused for the next iteration. This process is automatic, as the moment the gate fails the ancillary photon is routed into the next gate, whereas if the gate succeeds the ancillary photon is consumed and the remaining gates reduce to the identity as the Kerr interaction is no longer activated.

\subsection{Non-deterministic design with linear optics \label{app:numerical_methods}}

\begin{figure}

\begin{center}
\tikzset{every picture/.style={line width=0.75pt}} 

\begin{tikzpicture}[x=0.75pt,y=0.75pt,yscale=-1,xscale=1]

\draw  [fill={rgb, 255:red, 245; green, 166; blue, 35 }  ,fill opacity=1 ] (203.56,45.49) -- (204.92,48.11) -- (207.96,48.53) -- (205.76,50.58) -- (206.28,53.46) -- (203.56,52.1) -- (200.84,53.46) -- (201.36,50.58) -- (199.16,48.53) -- (202.2,48.11) -- cycle ;
\draw  [fill={rgb, 255:red, 245; green, 166; blue, 35 }  ,fill opacity=1 ] (203.56,55.78) -- (204.92,58.4) -- (207.96,58.82) -- (205.76,60.87) -- (206.28,63.75) -- (203.56,62.39) -- (200.84,63.75) -- (201.36,60.87) -- (199.16,58.82) -- (202.2,58.4) -- cycle ;
\draw  [fill={rgb, 255:red, 0; green, 0; blue, 0 }  ,fill opacity=1 ] (110.23,60.87) .. controls (110.23,59.78) and (111.11,58.91) .. (112.19,58.91) .. controls (113.28,58.91) and (114.15,59.78) .. (114.15,60.87) .. controls (114.15,61.95) and (113.28,62.83) .. (112.19,62.83) .. controls (111.11,62.83) and (110.23,61.95) .. (110.23,60.87) -- cycle ;
\draw  [fill={rgb, 255:red, 0; green, 0; blue, 0 }  ,fill opacity=1 ] (110.23,50.58) .. controls (110.23,49.5) and (111.11,48.62) .. (112.19,48.62) .. controls (113.28,48.62) and (114.15,49.5) .. (114.15,50.58) .. controls (114.15,51.66) and (113.28,52.54) .. (112.19,52.54) .. controls (111.11,52.54) and (110.23,51.66) .. (110.23,50.58) -- cycle ;
\draw  [fill={rgb, 255:red, 0; green, 0; blue, 0 }  ,fill opacity=1 ][line width=1.5]  (199.89,68.22) -- (202.83,68.22) .. controls (204.46,68.22) and (205.77,69.53) .. (205.77,71.16) .. controls (205.77,72.78) and (204.46,74.09) .. (202.83,74.09) -- (199.89,74.09) -- cycle ;
\draw [line width=0.75]    (114.15,50.58) -- (135.99,50.58) ;
\draw [line width=0.75]    (114.15,60.87) -- (135.99,60.87) ;
\draw [line width=0.75]    (111.7,71.16) -- (135.99,71.16) ;
\draw [line width=0.75]    (60,40) -- (136,40) ;
\draw [line width=0.75]    (135.99,40.29) -- (146.27,50.58) ;
\draw [line width=0.75]    (135.99,50.58) -- (146.27,40.29) ;
\draw [line width=0.75]    (135.99,60.87) -- (146.27,71.16) ;
\draw [line width=0.75]    (135.99,71.16) -- (146.27,60.87) ;
\draw [line width=0.75]    (166.85,50.58) -- (177.14,60.87) ;
\draw [line width=0.75]    (166.85,60.87) -- (177.14,50.58) ;
\draw [line width=0.75]    (166.85,71.16) -- (177.14,81.44) ;
\draw [line width=0.75]    (166.85,81.44) -- (177.14,71.16) ;
\draw [line width=0.75,-stealth]    (168,40) -- (259,40) ;
\draw [line width=0.75]    (177.14,50.58) -- (201.36,50.58) ;
\draw [line width=0.75]    (177.14,60.87) -- (201.36,60.87) ;
\draw [line width=0.75]    (177.14,71.16) -- (201.36,71.16) ;
\draw [line width=0.75]    (177.14,81.44) -- (199.89,81.44) ;
\draw    (146.27,60.87) -- (154.5,60.87) ;
\draw    (146.27,71.16) -- (154.5,71.16) ;
\draw  [dash pattern={on 1pt off 1pt}]  (154.5,50.58) -- (166.85,50.58) ;
\draw  [dash pattern={on 1pt off 1pt}]  (154.23,40) -- (166.58,40) ;
\draw    (146.27,50.58) -- (154.5,50.58) ;
\draw    (146,40) -- (154.23,40) ;
\draw  [dash pattern={on 1pt off 1pt}]  (154.5,60.87) -- (166.85,60.87) ;
\draw  [dash pattern={on 1pt off 1pt}]  (154.5,71.16) -- (166.85,71.16) ;
\draw  [dash pattern={on 1pt off 1pt}]  (154.5,81.44) -- (166.85,81.44) ;
\draw    (111.7,81.44) -- (154.5,81.44) ;
\draw  [fill={rgb, 255:red, 0; green, 0; blue, 0 }  ,fill opacity=1 ][line width=1.5]  (199.89,78.5) -- (202.83,78.5) .. controls (204.46,78.5) and (205.77,79.82) .. (205.77,81.44) .. controls (205.77,83.07) and (204.46,84.38) .. (202.83,84.38) -- (199.89,84.38) -- cycle ;
\draw  [color={rgb, 255:red, 155; green, 155; blue, 155}  ,draw opacity=1 ][dash pattern={on 1.5pt off 1.5pt}] (77,34) -- (243,34) -- (243,100) -- (77,100) -- cycle ;

\draw (27,32) node [anchor=north west][inner sep=0.75pt]    {$|\psi _{\text{in}} \rangle $};
\draw (263,32) node [anchor=north west][inner sep=0.75pt]    {$|\psi _{\text{out}} \rangle $};
\draw (80,44) node [anchor=north west][inner sep=0.75pt]   [align=left] {$\displaystyle |\mathbf{a} \rangle \begin{cases}
 & \\
 & 
\end{cases}$};
\draw (130,82.4) node [anchor=north west][inner sep=0.75pt]    {$\mathrm{U}( N+1)$};
\draw (192,44) node [anchor=north west][inner sep=0.75pt]   [align=left] {$\displaystyle \begin{drcases}
 & \\
 & 
\end{drcases} |\mathbf{a} \rangle $};
\draw (119,16) node [anchor=north west][inner sep=0.75pt]   [align=left] {Non-linearity};

\end{tikzpicture}
\end{center}
\caption{The measurement-induced non-linearity \label{fig:measurement_induced_nonlinearity}}
\end{figure}
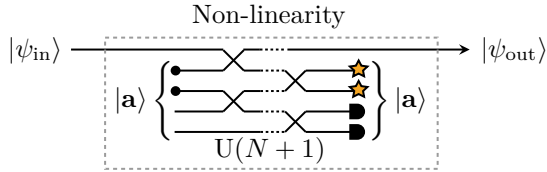

In this appendix we give a detailed overview of how to realise the non-deterministic dual-valued phase shifter with linear optics. 

On the subspace of at most $N$ photons, sometimes called the $N$th Fock layer, the dual-valued phase shifter (DVPS) can be represented as an $N$th degree polynomial in the number operator $\hat{n}$ by using polynomial interpolation: given a function $f(x)$ and a set of $N+1$ points $\{ x_i \}_{i = 0}^{N}$ then there exists a unique $K$th degree polynomial $P_K(x)$ that intersects the function at these points, i.e., $P_K(x_i) = f(x_i)$ for all $i = 0,1,\ldots,N$, where $K \leq N$. We shall apply this in order to find the polynomial representation of the DVPS. This operator acts as
\begin{equation}
e^{i\hat{q}x} |n\rangle = e^{iq(n)x} |n\rangle, \quad q(n) = \frac{1}{2}(1 - (-1)^n), 
\end{equation}
where $|n\rangle$ is a single-mode number state of $n$ photons. Let us first represent the eigenvalues on the $N$th Fock layer as an $N$th degree polynomial in $n$ as
\begin{equation}
 e^{iq(n)x} \equiv \sum_{m = 0}^N a_m(x) n^m,
\end{equation}
where $n = 0,1,\ldots,N$. This can be written as a linear equation by interpreting $T_{nm} \equiv n^m$ as the elements of a matrix $T$, known as a Vandermonde matrix, and the coefficients $a_n(x)$ and values $\exp[ixq(n)]$ as the elements of two vectors. In matrix notation we have
\begin{equation}
\begin{pmatrix}
1 \\ e^{ix} \\ 1 \\ \vdots
\end{pmatrix} = \begin{pmatrix}
1 & 0 & 0 &  \cdots \\
1 & 1 & 1  &\cdots \\
1 & 2 & 4 & \cdots \\
\vdots & \vdots & \vdots & \ddots
\end{pmatrix} \begin{pmatrix}
a_0(x) \\ a_1(x) \\ a_2(x) \\ \vdots 
\end{pmatrix}.
\end{equation}
Inverting the matrix $T$ allows us to solve for the vector of coefficients $a_m(x)$ giving us the sought after polynomial. Then by replacing the integer $n$ in the polynomial with the number operator $\hat{n}$ we arrive at the alternative expression for the DVPS when restricted to the $N$th Fock layer. 

As an example consider $N = 2$. The DVPS can be written as 
\begin{equation}
e^{i \hat{q} x} \equiv 1 - (e^{ix} - 1) \hat{n}(\hat{n}-2). \label{eq:dvps_polynomial}
\end{equation}
We stress that this equivalence is true on the subspace with at most $N = 2$ photons only. On the $N$th Fock layer we need a polynomial of degree $N$.

In order to encode this operation on the $N$th Fock layer, we use the following theorem from Ref.~\cite{PhysRevA.68.032310}.
\begin{theorem}[Measurement-induced non-linearities \cite{PhysRevA.68.032310}]
Consider an $(N+1)$-mode linear optical interferometer encoding a unitary $u \in \mathrm{U}(N+1)$ consisting of a single logical mode and $N$ ancillary modes, where the logical mode is prepared in the state $|\psi\rangle$ on the $n$th Fock layer and the ancillary modes are prepared in the single-occupation state $ |1,1,\ldots,1\rangle$. If the output of the ancillary modes is projected onto its input state, then the unnormalised output of the logical mode is $u_{00}^{\hat{n}} P_K(\hat{n})|\psi\rangle$, where $ P_K(\hat{n})$ is a polynomial of degree $K = \min\{ N,n \}$. 
\label{thm:1}
\end{theorem}
\begin{proof}
We use the index convention that the logical mode is indexed as $0$ and the ancillary modes are indexed with $\{1,2,\ldots,N\}$. Consider the input state $|\psi_\text{in}\rangle|\mathbf{a}\rangle $, where the logical mode is in an arbitrary state $|\psi_\text{in}\rangle$ and the ancillary modes contain a single photon each as $|\mathbf{a}\rangle = |1,1,\ldots,1\rangle$. Inserting this into an interferometer encoding a unitary $U:\mathcal{F} \rightarrow \mathcal{F}$ and projecting the ancillary modes onto their input state induces an operator $M : \mathcal{F}_0 \rightarrow \mathcal{F}_0$, where $\mathcal{F}_0$ is the Fock space of the logical mode. This mapping is given by
\begin{equation}
|\psi_\text{in}\rangle \mapsto \frac{M |\psi_\text{in}\rangle }{ \Vert M |\psi_\text{in}\rangle \Vert}, \quad
M  =  \langle \mathbf{a}| U | \mathbf{a}\rangle , \label{eq:M_definition}
\end{equation}
where the inner product is a partial inner product on the ancillary mode indices only. The magnitude of $M |\psi_\text{in}\rangle$ is the success probability. 

The induced operation $M$ conserves particle number. This is because the unitary $U$ encoded by the interferometer is particle-conserving, so if the number of photons in the ancillary modes is conserved, which is the case when projecting onto $|\mathbf{a}\rangle$, then the number of logical photons must not change either. Therefore $M$ is diagonal in the number basis as
\begin{equation}
M = \sum_n \alpha(n) |n\rangle \langle n |  = \alpha(\hat{n}), \label{eq:M_diagonal}
\end{equation}
where $\alpha$ is some function. From Eqs.~\eqref{eq:M_definition} and \eqref{eq:M_diagonal} the function $\alpha(n)$ is given by 
\begin{equation}
\alpha(n) = \langle n | M | n \rangle  = \langle n | \langle \mathbf{a} | U | n \rangle | \mathbf{a}\rangle. \label{eq:M_eigenvalue}
\end{equation}
The $n$th eigenvalue of $M$ is then the amplitude for the interferometer to leave the state $|n \rangle |\mathbf{a}\rangle$ invariant which of course is given by the permanent from Eq.~\eqref{eq:app_transition_amplitude}. We would like to find the functional dependence of this on $n$ so we must expand this out explicitly. We have
\begin{equation}
\begin{aligned}
U |n\rangle |\mathbf{a}\rangle & = \frac{1}{\sqrt{n!}} U (a_0^\dagger)^n \prod_{i= 1}^N a_i^\dagger |0\rangle \\
& = \frac{1}{\sqrt{n!}} \left( \sum_{i = 0}^N u_{i0} a_i^\dagger \right)^n \prod_{i= 1}^N \left(  \sum_{j=0}^N u_{ji}a_j^\dagger \right) |0\rangle ,
\end{aligned}
\end{equation}
where we have inserted in the linear transformation of the ladder operators of Eq.~\eqref{eq:unitary}. 
We now apply the multinomial theorem to the sum raised to the power of $n$, to give us 
\begin{equation}
\begin{aligned}
U |n\rangle |\mathbf{a}\rangle & = \frac{1}{\sqrt{n!}} \sum_{\substack{k_0+\ldots+k_N = n \\  k_i \geq 0}} \frac{n!}{k_0! k_1! \ldots k_N!} ( u_{00} a^\dagger_0)^{k_0} ( u_{10} a^\dagger_1)^{k_1} \ldots ( u_{N0} a^\dagger_N)^{k_N} \prod_{i= 1}^N \left(  \sum_{j=0}^N u_{ji}a_j^\dagger \right) |0\rangle , \label{eq:multinomial}
\end{aligned}
\end{equation} 
This is a complicated mess of terms, however from Eq.~\eqref{eq:M_eigenvalue} we are interested in the amplitude for $|n\rangle |\mathbf{a}\rangle$ at the output only, so we look for the coefficient of this term in the expansion. To proceed note that we have two sets of modes: a single logical mode that must be occupied by $n$ photons and $N$ ancillary modes that must each be occupied by a single photon. We can divide up the amplitude into cases where $m$ photons swap between these two subsets, where the number of swaps is upper bounded as $m \leq \min \{ N,n \}$. 

Consider the amplitude for $m$ swaps. In Eq.~(\ref{eq:multinomial}) the index $k_0$ in the sum corresponds to the number of logical photons that remain in the logical mode after the transformation, so we can change variables as $k_0 = n - m$. The indices $k_1,\ldots,k_N$ correspond to the logical photons that \textit{did} swap and tell us what ancillary modes they ended up in. As we can have at most one logical photon swapping with an ancillary photon, we have $k_i! = 1$ for $i = 1,\ldots,N$. Therefore, we can write
\begin{equation}
\begin{aligned}
U |n\rangle |\mathbf{a}\rangle & = \frac{1}{\sqrt{n!}} \sum_{m = 0}^{\min \{ N,n \}} \frac{n!}{(n-m)!} u_{00}^{n-m} A_m (a_0^\dagger)^n a_1^\dagger a_2^\dagger \ldots a_N^\dagger |0\rangle + \ldots \\
& = \sum_{m = 0}^{\min \{ N,n \}} \frac{n!}{(n-m)!} u_{00}^{n-m} A_m |n\rangle |\mathbf{a}\rangle + \ldots,
\end{aligned}
\end{equation}
where $A_m$ is the remainder of the amplitude that corresponds to the transitions made by the remaining photons. If $m$ logical photons and $m$ ancillary photons are swapped, then we have many choices of ancillary photons to swap with. This amplitude is then a sum over all possible subsets of ancillary photons of size $m$ as
\begin{equation}
A_m = \sum_{ \substack{T \subseteq \{ 1,\dots,N\} \\ |T| = m}} \Bigg( \prod_{i \in T} u_{0i}u_{i0}  \Bigg) \Bigg( \sum_{\sigma \in S_{N-m}} \prod_{j \not\in T} u_{j \sigma(j)} \Bigg).
\end{equation}
The first term in parentheses is the amplitude for $m$ logical photons to swap with $m$ ancillary photons from the subset $T$, where $u_{0i}$ is the amplitude for the logical photon to move to the $i$th ancillary mode, and $u_{i0}$ is the ampltiude for the $i$th ancillary photon to move to the logical mode. The second term in parentheses corresponds to the amplitude for the ancillary photons that did not swap. These photons transition between the modes $T^C =  \{1,\ldots,N\} \setminus T$, with only one per mode at the output, and includes all possible permutations of this. This term can be rewritten as a permanent of the submatrix of $u$ with only rows and columns selected from $T^C$. 

Pulling everything together, we have the eigenvalues
\begin{equation}
\alpha(n) =  u^n_{00} \sum_{m = 0}^{\min \{N,n \}}  \frac{n!}{(n-m)!}   \frac{1}{u_{00}^m} A_m. \label{eq:induced_amplitudes}
\end{equation}
This summation is a polynomial in $n$ of degree $ \min \{ N,n\}$ as the last term in the sum has the combinatorial factor $n(n-1)\ldots(n- \min \{ N,n \} + 1)$ which is the highest-order term. By replacing all $n$ with the number operator $\hat{n}$, we arrive at the result that on the $n$th Fock layer $M$ is given by $u_{00}^{ \hat{n} } P_K(\hat{n})$ for $K = \min \{N,n \}$, where $P_K$ is an $K$th order polynomial. See Fig.~\ref{fig:measurement_induced_nonlinearity} for an example of this in action.
\end{proof}

In general it will not be possible to solve for the generated polynomial analytically, so we resort to numerical methods to find the unitary $u$ that results in the polynomial we seek. In order to do this we reframe the problem into something easier to encode and solve numerically. Suppose we prepare the logical mode in the most general photonic state for a single mode on the $N$th Fock layer as
\begin{equation}
|\psi_\text{in} \rangle = \sum_{n = 0}^N  c_n |n\rangle,
\end{equation}
where $|n\rangle$ is the state consisting of $n$ photons and $\sum_{n= 0}^N |c_n|^2 = 1$. We then prepare an $N$-mode ancillary state $|\mathbf{a}\rangle = |a_1,a_2,\ldots,a_N \rangle$. The total state of the logical mode and the ancillary modes is
\begin{equation}
|\psi_\text{in} \rangle |\mathbf{a}\rangle  = \sum_{n = 0}^N  c_n |n\rangle |\mathbf{a}\rangle  \equiv \sum_{n = 0}^N  c_n |\mathbf{v}_n\rangle,
\end{equation}
where $\mathbf{v}_n = (n,a_1,a_2,\ldots,a_N)$. We take the ancillary modes have occupation numbers $a_i \in \{0,1\}$ only. We now insert this state into an $(N+1)$-mode linear interferometer described by a unitary $U$ which acts as Eq.~\eqref{eq:unitary}. We have
\begin{equation}
\begin{aligned}
U |\psi_\text{in} \rangle |\mathbf{a}\rangle & = \sum_{n = 0}^N  c_n U |\mathbf{v}_n\rangle \\
& = \sum_{n = 0}^N  \sum_{ \mathbf{m} \in \mathbb{N}^{N + 1}} c_n \langle \mathbf{m} | U |\mathbf{v}_n\rangle |\mathbf{m}\rangle \\
& = \sum_{n = 0}^N  \sum_{\mathbf{m} \in \mathbb{N}^{N+1}} \frac{c_n}{\sqrt{\mathbf{m}! n!}}  \operatorname{per}(u[\mathbf{m}|\mathbf{v}_n]) |\mathbf{m}\rangle \\
& =  \sum_{n = 0}^N  \frac{c_n}{n!} \operatorname{per}( u[\mathbf{v}_n|\mathbf{v}_n]) |n\rangle |\mathbf{a}\rangle + \ldots,
\end{aligned}
\end{equation}
where in the second line we multiplied by the resolution of the identity in the Fock basis, in the third line we used the permanent formula for the transition amplitudes from Eq.~\eqref{eq:app_transition_amplitude} and used the fact that $\mathbf{v}_n! = n!$ if $a_i \in \{0,1\}$, and in the final line we pulled out a common factor of $|\mathbf{a}\rangle$. We now measure the ancillary modes and keep the logical state if the ancillary modes are measured in the state $|\mathbf{a}\rangle$. This gives us the unnormalised output state of the logical mode 
\begin{equation}
|\psi_\text{out}\rangle = \sum_{n = 0}^N  \frac{c_n}{n!} \operatorname{per}(u[\mathbf{v}_n|\mathbf{v}_n]) |n\rangle, \label{eq:post_selected_output}
\end{equation}
where the success probability is given by $p = \langle \psi_\text{out}|\psi_\text{out}\rangle$. On the other hand, the desired action of the dual-valued phase shifter on the logical mode is given by
\begin{equation}
U_\text{DVPS}(x)|\psi_\text{in}\rangle = \sum_{n = 0}^N  c_n e^{iq(n)x} |n\rangle, \quad q(n) = \frac{1}{2}(1 - (-1)^n), \label{eq:DVPS_definition}
\end{equation}
where we have chosen $q(n)$ to take this particular form, however any two-valued real function $q : \mathbb{N} \to \{ a,b \}$ could be used.

Let us assume that we can achieve this operation up to a global phase with measurement-induced non-linearities with a probability of $p_x \in [0,1]$, so we can write
\begin{equation}
|\psi_\text{out}\rangle = \sqrt{p_x} e^{i \alpha_x} U_\text{DVPS}(x) |\psi_\text{in}\rangle, \label{eq:measurement_induced_DVPS}
\end{equation}
where $\alpha_x \in [0,2\pi)$ is the global phase that in general depends upon $x$. We see that $\langle\psi_\text{out}| \psi_\text{out}\rangle = p_x$ confirming that $p_x$ is indeed the success probability. If we combine Eqs.~\eqref{eq:post_selected_output}, \eqref{eq:DVPS_definition} and \eqref{eq:measurement_induced_DVPS} then we must solve for the matrix $u_x \in \mathrm{U}(N+1)$ that solves the equation
\begin{equation}
\frac{1}{n!} \operatorname{per}(u_x[\mathbf{v}_n|\mathbf{v}_n]) = \sqrt{p_x} e^{i\alpha_x} e^{ixq(n)}, \label{eq:condition}
\end{equation}
for all $n = 0,1,\ldots,N$ and $x \in [0,2\pi)$. Eq.~\eqref{eq:condition} provides us with $N$ constraints in the form of polynomial equations in the elements of the unitary matrix $u_x$. Within the space of solutions that satisfy the constraints, we additionally need to find the particular unitary that maximises the success probability $p_x \in [0,1]$. By setting $n = 0$ in Eq.~\eqref{eq:condition}, the success probability is given by
\begin{equation}
\sqrt{p_x} e^{i \alpha_x} = \mathop{\mathrm{per}}(u_x[\mathbf{v}_0 |\mathbf{v}_0]) \quad \Rightarrow \quad  p_x = \left|\operatorname{per}(u_x[\mathbf{v}_0|\mathbf{v}_0])\right|^2.
\end{equation}
The problem can then be recast as an optimisation problem for each $x \in [0,2\pi)$ as
\begin{equation}
\begin{aligned}
\mathop{\mathrm{maximise}}_{u_x \in \mathrm{U}(N+1)} \quad &  p_x = \left|\operatorname{per}(u_x[\mathbf{v}_0|\mathbf{v}_0])\right|^2,  \\
\text{subject to} \quad & \left| \frac{1}{n!} \operatorname{per}(u_x[\mathbf{v}_n|\mathbf{v}_n]) -  \mathop{\mathrm{per}}(u_x [\mathbf{v}_0 |\mathbf{v}_0]) e^{ixq(n)} \right|^2 = 0 , \quad \forall n = 1,\ldots,N.
\end{aligned}  \label{eq:optimisation_problem}
\end{equation}
This is a complicated optimisation problem and many methods have been presented in the literature, e.g., Refs.~\cite{PhysRevLett.95.040502,wu2007optimizingopticalquantumlogic,sparrow2018simulating,PhysRevA.79.042326}. However, for small $N$ this can be solved in a simple manner as we now show.

\begin{figure}
\begin{center}
\includegraphics[scale=0.7]{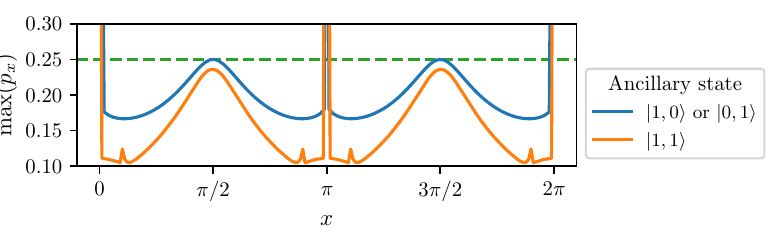}
\end{center}
\caption{The maximum success probability $p_x$ for the non-deterministic DVPS versus the phase $x$ for various ancillary input states on the $N = 2$ subspace. We see that the ancillary state $|1,0\rangle$ or $|0,1\rangle$ obtains the upper bound of $1/4$. For $x = 0,\pi$ the gate is deterministic with $p_x = 1$, as here the DVPS is simply the identity or $\pi$ phase shifter, respectively. Note that the discontinuity in $|1,1\rangle$ happens at $x \approx \pi/10$. \label{fig:ancillary_comparison}}
\end{figure}

In the main text we looked at an example for $N = 2$ for the ancillary state $|\mathbf{a}\rangle = |1,0\rangle$. In order to solve the optimisation problem of Eq.~\eqref{eq:optimisation_problem} we parametrised our $\mathrm{U}(3)$ unitaries by writing them in the exponential form as
\begin{equation}
u(\boldsymbol{\theta}) = \exp \left(i \sum_{\mu = 1}^9 \theta_\mu T^\mu \right),
\end{equation}
where $ \{ T^\mu \}_{\mu = 1}^9$ is a basis of nine three-dimensional Hermitian matrices that spans the Lie algebra $\mathfrak{u}(3)$ and $\theta_\mu \in \mathbb{R}$ parametrise the unitaries. We chose the basis
\begin{equation}
\begin{aligned}
T^1 & = 
\begin{pmatrix}
1 & 0 & 0 \\ 
0 & 0 & 0 \\ 
0 & 0 & 0
\end{pmatrix}, 
\quad 
T^2 = 
\begin{pmatrix}
0 & 0 & 0 \\ 
0 & 1 & 0 \\ 
0 & 0 & 0
\end{pmatrix}, 
\quad 
T^3 = 
\begin{pmatrix}
0 & 0 & 0 \\ 
0 & 0 & 0 \\ 
0 & 0 & 1
\end{pmatrix} \\
T^4  & = 
\begin{pmatrix}
0 & 1 & 0 \\ 
1 & 0 & 0 \\ 
0 & 0 & 0
\end{pmatrix}, 
\quad 
T^5 = 
\begin{pmatrix}
0 & 0 & 1 \\ 
0 & 0 & 0 \\ 
1 & 0 & 0
\end{pmatrix}, 
\quad
T^6 = 
\begin{pmatrix}
0 & 0 & 0 \\ 
0 & 0 & 1 \\ 
0 & 1 & 0
\end{pmatrix},
\\
T^7  & = 
\begin{pmatrix}
0 & -i & 0 \\ 
i & 0 & 0 \\ 
0 & 0 & 0
\end{pmatrix}, 
\quad
T^8 = 
\begin{pmatrix} 
0 & 0 & -i \\ 
0 & 0 & 0 \\ 
i & 0 & 0 
\end{pmatrix}, 
\quad 
T^9 = 
\begin{pmatrix}
0 & 0 & 0 \\ 
0 & 0 & -i \\ 
0 & i & 0 
\end{pmatrix}.
\end{aligned}
\end{equation}
We then solved the constrained optimisation problem of Eq.~\eqref{eq:optimisation_problem} by optimising over the parameters $\theta_\mu$ using the Sequential Least Squares Programming (SLSQP) of SciPy's optimize package on Python. Due to local minima, the optimiser would often return a suboptimal solution, i.e., one which solved the permanent equation constraints but did not maximise $p_x$, so we repeated the optimisation $500$ times per $x$ and returned the solution with the largest value of $p_x$.

 In Fig.~\ref{fig:ancillary_comparison} we show the maximum probability vs. $x$ for three possible ancillary states $|1,1\rangle$, $|1,0\rangle$ and $|0,1\rangle$. Note that the performance of $|1,0\rangle$ and $|0,1\rangle$ is identical as they can be transformed between with a deterministic linear transformation which does not change the success probability. We see that $|1,0\rangle$ actually outperforms $|1,1\rangle$ and is able to obtain the upper bound of $1/4$ as predicted analytically. In Table~\ref{tab:unitary_table} we give a few example unitaries which encode the desired DVPS with phase $x$ on the subspace of at most $N = 2$ photons using the ancillary state $|\mathbf{a}\rangle = |1,0\rangle$, together with the success probability $p_x$.
 
We may wonder why the ancillary state $|\mathbf{a}\rangle = |1,0\rangle$ is able to induce the action of the second order polynomial of Eq.~\eqref{eq:dvps_polynomial} when theorem~\ref{thm:1} states that it should induce a polynomial $P_1(\hat{n})$ of degree $1$ in the number operator. The reason is that the additional factor of $u_{00}^{\hat{n}}$ boosts the degree of the polynomial by one, giving rise to a polynomial of degree $2$ on this Fock layer. More precisely, with a single ancillary photon we generate an operator of the form 
\begin{equation}
M = u_{00}^{\hat{n}} (a \hat{n} + b)
\end{equation}
which, via polynomial interpolation, is equivalent to a polynomial of degree two on the second Fock layer. If $u_{00} = 1$ then $M$ is just a first order polynomial. 

Now whether the polynomial we generate is what we would like is another matter, which is where the additional empty ancillary mode aids us. This empty ancillary mode allows the submatrix of the encoded unitary $u$ corresponding to the occupied input and output modes to have more freedom. If we have one ancillary photon, then Eq.~\eqref{eq:induced_amplitudes} gives us the amplitudes
\begin{align}
\alpha_0 & = u_{00},  \\
\alpha_1 & = u_{00} u_{11} + u_{01} u_{10}, \\
\alpha_2  &= u^2_{00} u_{11} + 2u_{00} u_{01} u_{10}, 
\end{align}
and there is enough degrees of freedom in the unitary $u \in \mathrm{U}(3)$ to find one such that $\alpha_n \propto \exp[i q(n) x]$. On the other hand if $u \in \mathrm{U}(2)$, which would be the case if we did not have this additional empty ancillary mode, then we would have the constraint that the columns are orthonormal, so $|u_{00}|^2 + |u_{10}|^2 = |u_{01}|^2 + |u_{11}|^2 = 1$ and $u_{00} u_{11} + u_{10}u_{01} = 0$, and similarly for the rows, which heavily constrains the matrix elements. However, if $u \in \mathrm{U}(N)$ for $N > 2$ then submatrices of the unitaries are not always unitary themselves, so the constraints on the matrix elements are relaxed and we therefore have more freedom to choose the matrix elements to solve these equations. This also makes it clear that once the optimal solution has been found, then additional ancillary modes will not increase the probability as it is only the 2D submatrix that appears in these equations.

\begin{table}
\begin{center}
\begin{tabular}{|c | c |c |} 
 \hline
 $x$ & $u_x$ & $p_x$  \\ [0.5ex] 
 \hline\hline 
 $ 0 $ & $\begin{pmatrix}
 1 & 0 & 0 \\ 0 & 1 & 0 \\ 0 & 0 & 1
 \end{pmatrix}$ & 1 \\ 
 \hline
 $ \pi/8 $ & 
 $ 
 \begin{pmatrix}
 0.4376-0.3446i & -0.4479+0.3954i & 0.0255+0.5763i \\
       -0.4538-0.39i  &  0.2225-0.3428i & -0.43  +0.5385i \\
        0.0155-0.5757i & -0.4392-0.5321i & 0.2717-0.3442i
 \end{pmatrix}
 $
  & $1/6$ \\
 \hline
 $\pi/4$ & $\begin{pmatrix} 0.2517-0.3920i & -0.6332+0.3146i &  0.5144-0.1359i \\
       -0.7080 +0.1018i & -0.0896-0.4160i &  0.5120 -0.2123i \\
       -0.3028-0.4238i & -0.5023-0.2584i & -0.6398-0.0181i  \end{pmatrix}$ & $ 2/11 $\\
 \hline 
 $ 3\pi/8 $ & 
 $
 \begin{pmatrix}
 0.1171-0.4097i & -0.7931-0.0695i & -0.1944+0.3832i \\
       -0.0499-0.7917i &  0.4535+0.1013i & -0.3871-0.0706i \\
       -0.3224-0.2919i & -0.3174-0.2225i &  0.3428-0.7369i
 \end{pmatrix}
 $
  & $8/37$ \\
 \hline
 $\pi/2$ & 
$
\begin{pmatrix}
0.0001-0.4146i &  0.7130+0.339i & 0.3426-0.0321i \\
0.8046-0.2382i &  0.0631-0.4962i &  0.0792+0.1985i \\
0.0680+0.3455i & -0.1952-0.044i &  0.8769-0.2588i 
\end{pmatrix}
$ 
  & $1/4$  \\
 \hline
 $\pi$ & $\begin{pmatrix}
 -1 & 0 & 0 \\ 0 & 1 & 0 \\ 0 & 0 & 1
\end{pmatrix}$  & $1$  \\ [1ex] 
 \hline
\end{tabular}
\caption{Example unitaries $u_x$ and success probabilities $p_x$ for encoding the dual-valued phase shifter for various $x$ on the Fock layer of $N = 2$, given the ancillary state $|\mathbf{a}\rangle = |1,0\rangle$. Note that these unitaries have been rounded to four decimal places for presentational purposes. \label{tab:unitary_table}}
\end{center}
\end{table}

\section{Effects of loss \label{app:loss}}

\subsection{Measurement induced DVPS}
As dicussed in the main text, if we have a lossy linear optical interferometer described by the transformation
\begin{equation}
\mathbf{a}^\dagger \rightarrow A \mathbf{a}^\dagger + \sqrt{1 - A A^\dagger} \mathbf{b}^\dagger
\end{equation}
where $\mathbf{a}^\dagger = (a_1^\dagger, a_2^\dagger ,\ldots , a_N^\dagger)$ and $\mathbf{b}^\dagger = (b_1^\dagger, b_2^\dagger, \ldots, b_N^\dagger)$ are the vectors of mode operators for the interferometer and environment, respectively. By taking the singular value decomposition of the matrix $A$, we can write this as 
\begin{equation}
A = W \Sigma V^\dagger
\end{equation} 
where $V$ and $W$ are unitary matrices and the diagonal matrix of singular values
\begin{equation}
\Sigma = \begin{pmatrix}
\sqrt{1 - \eta_1} & & & \\
& \sqrt{1 - \eta_2} & & \\
& & \ddots & 
\end{pmatrix}
\end{equation}
Hence, a lossy interferometer is equivalent to a series of single-mode loss channels acting on each of the $N$ modes individually, sandwiched between two lossless interferometers. In the case of zero loss, the transformation is $A = WV^\dagger \equiv U$ which is the original lossless unitary transformation. In the simple case of uniform loss we have $\eta_i = \eta$ for all $i$, so that $\Sigma = \sqrt{1 - \eta} \mathbb{I}$ and therefore the total transformation is
\begin{equation}
A = \sqrt{1-\eta} U 
\end{equation}
where $U$ is the lossless unitary transformation. This is equivalent to applying a uniform loss to each input state.  

We now investigate how uniform loss affects the measurement-induced DVPS. As the DVPS is constructed from linear optics alone, then the effects of uniform loss can be modelled by applying uniform losses to the input states instead. Let us ignore the effect of the loss channel on the logical mode, as this loss can be absorbed into the circuit leading up to the DVPS, so the only losses occur at the input of the ancillary modes. The input state to the logical and ancillary modes is
\begin{equation}
\rho = \rho_\text{in} \otimes \rho_\text{A}, \quad \rho_\text{in} = |\psi_\text{in}\rangle \langle \psi_\text{in}|, \ \rho_\text{A} = |\mathbf{a}\rangle \langle \mathbf{a}|
\end{equation}
as before. If the lossless DVPS uses the input ancillary state $|\mathbf{a}\rangle$ consisting of the first $k$ modes filled with a single photon, then after applying losses to this state we get the mixed state
\begin{equation}
\rho_A \rightarrow \sigma^{\otimes k} \otimes (|0\rangle \langle 0 |)^{\otimes (N - k)}, \quad  \sigma = (1-\eta) |1\rangle \langle 1| + \eta |0\rangle \langle 0 |.
\end{equation}
Successful operation of the DVPS is clearly when no ancillary photons are lost: if there is no loss then the input ancillary state to $U_\text{DVPS}$ is the correct one and the gate will operate as intended. If the lossless gate had a success probability of $p(x)$ to apply a DVPS with phase of $x$, see Fig.~\ref{fig:ancillary_comparison}, then with losses the probability becomes
\begin{equation}
p_\eta(x) = (1 - \eta)^k p(x). \label{eq:lossy_success}
\end{equation}

As the operation of the gate is heralded by measuring the output ancillary state $|\mathbf{a}\rangle$, with losses it is possible for false heralds to occur. These are situations where an ancillary photon is lost, but is replenished by a photon from the logical mode. This can yield the correct herald state, but the gate will clearly have failed as a photon in the logical mode has now been lost. We can calculate the probability for a false herald by calculating the probability to herald the output ancilla state $|\mathbf{a}\rangle$ and taking the difference with EQ.~\eqref{eq:lossy_success}. The probability for a herald is given by
\begin{equation}
p_\mathrm{H}(x) = \mathrm{Tr}[ U(x) \Phi(\rho) U^\dagger(x) \Pi], \quad \Pi = \mathbb{I} \otimes |\mathbf{a}\rangle \langle \mathbf{a}|.
\end{equation}
To evaluate this, we can first rewrite the input ancilla state as
\begin{equation}
\begin{aligned}
\rho_A  & = [ (1- \eta) |1\rangle \langle 1 | + \eta |0\rangle \langle 0 |]^{\otimes k} \otimes (|0\rangle \langle 0 |)^{\otimes (N-k)} \\
& = \sum_{\boldsymbol{\sigma}} (1-\eta)^{w(\boldsymbol{\sigma})} \eta^{k - w(\boldsymbol{\sigma})} |\boldsymbol{\sigma}\rangle \langle \boldsymbol{\sigma} |,
\end{aligned}
\end{equation}
where $w(\boldsymbol{\sigma})$ is the Hamming weight of the bit string $\boldsymbol{\sigma}$ and the sum is over all possible bit string where the last $N-k$ bits are zero. Inserting this in we get
\begin{equation}
p_\mathrm{H}(x) = \sum_{\boldsymbol{\sigma} } (1-\eta)^{w(\boldsymbol{\sigma})} \eta^{k - w(\boldsymbol{\sigma})} \mathrm{Tr}\left[ U(x) (\rho_\text{in} \otimes  |\boldsymbol{\sigma}\rangle \langle \boldsymbol{\sigma} | ) U^\dagger(x) \Pi \right].
\end{equation}
The trace now involves pure state alone, so it is easier to work with state vectors. This trace is asking: what is the probability to measure the ancillary modes in the state $|\mathbf{a}\rangle$ when the input state to the interferometer is $|\psi_\text{in}\rangle |\boldsymbol{\sigma}\rangle$? To answer this, we follow the similar steps we followed for deriving the DVPS. We have
\begin{equation}
\begin{aligned}
U(x) |\psi_\text{in}\rangle |\boldsymbol{\sigma}\rangle & = \sum_{n =0}^N c_n U |n ,\boldsymbol{\sigma}\rangle \\
& = \sum_{n = 0}^N \sum_{m = 0}^N \sum_{\mathbf{m}} c_n \langle m, \mathbf{m} | U |n,\boldsymbol{\sigma} \rangle |m, \mathbf{m}\rangle \\
& =\sum_{n = 0}^N \sum_{m = 0}^N \sum_{\mathbf{m}}  \frac{c_n}{\sqrt{m! n! \mathbf{m}!}} \mathop{\mathrm{per}} \left( u[m,\mathbf{m} | n, \boldsymbol{\sigma}]  \right)|m, \mathbf{m}\rangle \\
& = \underbrace{\left( \sum_{m, n = 0}^N \frac{c_n}{\sqrt{m! n!}}  \mathop{\mathrm{per}} \left( u[m,\mathbf{a} | n, \boldsymbol{\sigma}]  \right) |m\rangle \right)}_{|\phi\rangle} |\mathbf{a}\rangle  + \ldots.
\end{aligned}
\end{equation}
The trace is simply the squared norm of $|\phi\rangle$, so 
\begin{equation}
\begin{aligned}
\mathrm{Tr}\left[ U(x) (\rho_\text{in} \otimes  |\boldsymbol{\sigma}\rangle \langle \boldsymbol{\sigma} | ) U^\dagger(x)  \Pi \right] & = \left( \sum_{m, n = 0}^N \frac{c^*_n}{\sqrt{m! n!}}  \mathop{\mathrm{per}} \left( u^*[m,\mathbf{a} | n, \boldsymbol{\sigma}]  \right) \langle m | \right) \left( \sum_{p, q = 0}^N \frac{c_q}{\sqrt{p! q!}}  \mathop{\mathrm{per}} \left( u[p,\mathbf{a} | q, \boldsymbol{\sigma}]  \right) |p\rangle \right) \\
& = \sum_{m,n,q = 0}^N \frac{c^*_n c_q}{m! \sqrt{n!q!}} \mathop{\mathrm{per}} \left( u^*[m,\mathbf{a} | n, \boldsymbol{\sigma}] \right) \mathop{\mathrm{per}} \left( u[m,\mathbf{a} | q, \boldsymbol{\sigma}] \right) \\
& =  \sum_{m,n = 0}^N \frac{|c_n|^2}{m! n!}  |\mathop{\mathrm{per}} \left( u[m,\mathbf{a} | n, \boldsymbol{\sigma}]  \right)|^2,
\end{aligned}
\end{equation}
where going to the last line we used the fact that the only terms in the sum that survive are when $q = n$, which is because the permanents are transition amplitudes under a particle conserving unitary, so they are only non-zero when particle number is conserved. Conservation of particle number for the first amplitude means that 
\begin{equation}
m + \sum_i a_i = n + \sum_i \sigma_i,
\end{equation}
and the second amplitude requires
\begin{equation}
m + \sum_i a_i = q + \sum_i \sigma_i.
\end{equation}
These two equations imply $n = q$. Pulling everything together, we find the probability to herald given the loss $\eta$ is given by
\begin{equation}
p_\mathrm{H}(x, \eta) =  \sum_{m,n = 0}^N \frac{|c_n|^2}{m! n!} \sum_{\boldsymbol{\sigma} } (1-\eta)^{w(\boldsymbol{\sigma})} \eta^{k - w(\boldsymbol{\sigma})}   |\mathop{\mathrm{per}} \left( u[m,\mathbf{a} | n, \boldsymbol{\sigma}]  \right)|^2.
\end{equation}
This formula depends upon the input state as the coefficients $c_n$ are present. To understand how this formula behaves on average, we evaluate it by averaging over all possible input states $|\psi_\text{in}\rangle$ by sampling them uniformly with respect to the Haar measure. The most general input state for a single mode photonic state on the subspace of at most $N$ photons is 
\begin{equation}
|\psi_\text{in}\rangle = \sum_{n = 0}^N c_n |n\rangle \equiv U |0\rangle,
\end{equation}
which we have written as a unitary rotation of some reference state which we take to be the vacuum state $|0\rangle$. The subspace these states live within is a Hilbert space of dimension $N + 1$. If we can sample the coefficients $c_n$ uniformly with respect to the Haar measure by sampling $U \in \mathrm{U}(N+ 1)$ with respect to the Haar measure. Using the fact that $|c_n|^2 = \mathrm{Tr}( \rho |n\rangle \langle n |)$ we get 
\begin{equation}
\begin{aligned}
\langle |c_n|^2 \rangle & = \int_{\mathrm{U}(N+1)} dU \mathrm{Tr}\left[ U |0\rangle \langle 0 | U^\dagger |n\rangle \langle n | \right] \\
& =   \mathrm{Tr}\left[ \left( \int_{\mathrm{U}(N+1)} dU U |0\rangle \langle 0 | U^\dagger \right)  |n\rangle \langle n | \right] \\
& = \mathrm{Tr}\left[ \frac{1}{N+1} \mathrm{Tr}( |0\rangle \langle 0 |) |n\rangle \langle n | \right] \\
& = \frac{1}{N+1},
\end{aligned}
\end{equation}
where $dU$ is the Haar measure on $\mathrm{U}(N+1)$ and we used Lemma~\ref{thm:haar_twirl} to evaluate the integral, bringing the false herald probability to
\begin{equation}
p_\mathrm{H}(x, \eta) =  \frac{1}{N+1} \sum_{m,n = 0}^N \frac{1}{m! n!} \sum_{\boldsymbol{\sigma} } (1-\eta)^{w(\boldsymbol{\sigma})} \eta^{k - w(\boldsymbol{\sigma})}   |\mathop{\mathrm{per}} \left( u[m,\mathbf{a} | n, \boldsymbol{\sigma}]  \right)|^2.
\end{equation}
In the case of zero loss, so when $\eta = 0$, the ancillary input state reduces to $\rho_A = |\mathbf{a}\rangle \langle \mathbf{a}|$ and the only term that contributes in this sum is when $\boldsymbol{\sigma} = \mathbf{a}$  which then by particle number conservation forces $m = n$, reducing $p_\text{H}(x)$ to the formula for the success probability in the lossless case. 

To apply this to our simple example for $N = 2$ with the ancilla state $|\mathbf{a}\rangle = |1,0\rangle$ we get $k=1$, this formula reduces to 
\begin{equation}
p_\text{H}(x, \eta) = (1 - \eta)p(x) +  \frac{\eta}{3} \sum_{m,n = 0}^2 \frac{1}{m! n!} \left| \mathop{\mathrm{per}} u[m,1,0|n,0,0] \right|^2.
\end{equation}

As the probability of a successful herald when there is no loss is given by $(1- \eta)^k p(x)$ where $k$ is the number of ancillary photons, then the probability that given a herald in a lossy system it is a false herald is given by
\begin{equation}
p_\text{false}(x, \eta) = \frac{p_\mathrm{H}(x, \eta) - (1- \eta)^k p(x)}{p_\mathrm{H}(x, \eta)}.
\end{equation}

\subsection{Photonic fermion sampling}
Let us now investigate the effects of uniform loss in photonic fermion sampling. The photonic fermion sampler consists of a set of $n$ copies of an $N$-mode interferometer, where each interferometer contains a single photon. The Hilbert space is given by $\mathcal{H}^{\otimes n}$, where $\mathcal{H} = \mathrm{span}\{ |i\rangle \}$ is the Hilbert space a single photon in $N$ modes and $|i\rangle$ is the state for the photon in the $i$th mode. We assume that $i = 1,2,\ldots,n$ and reserve the notation $|0\rangle$ to the vacuum state. The input state is the totally-antisymmetric $n$-photon state given by
\begin{equation}
|\psi\rangle = \frac{1}{\sqrt{n!}} \sum_{\sigma \in S_n} (-1)^\sigma \prod_{\mu = 1}^n a_{\mu \sigma(\mu)}^\dagger |0\rangle.
\end{equation}
where $a^\dagger_{\mu i}$ is the ladder operator that creates a photon in the $i$th mode of the $\mu$th interferometer.

Now as before, we extend the Hilbert space to $\mathcal{H} \otimes \mathcal{H}_\text{E}$, where $\mathcal{H}_\text{E}$ is the Hilbert space of the environment, by introducing a single loss mode per physical mode. The loss channel is modelled as
\begin{equation}
a_{\mu \sigma(\mu)}^\dagger \rightarrow U_\mathrm{L} a_{\mu \sigma(\mu)}^\dagger U^\dagger_\mathrm{L} =  \sqrt{1 - \eta} a_{\mu \sigma(\mu)}^\dagger + \sqrt{\eta} e_{\mu \sigma(\mu)}^\dagger. \label{eq:fermionic_loss_channel}
\end{equation}
Let us consider uniform loss on the first interferometer only. It is easier to work in first quantisation for this calculation, so let us rewrite our input state as
\begin{equation}
|\psi\rangle = \frac{1}{\sqrt{n!}} \sum_{\sigma \in S_n} (-1)^\sigma \bigotimes_{\mu = 1}^n |\sigma(\mu)\rangle
\end{equation}
where $a_{\mu \sigma(\mu)}^\dagger |0\rangle$ generates the state $|\sigma(\mu)\rangle$ in the $\mu$th tensor slot. The density matrix is given by
\begin{equation}
\rho = \frac{1}{n!}\sum_{\sigma, \tau \in S_n} (-1)^{\sigma + \tau} \bigotimes_{\mu = 1}^n |\sigma(\mu)\rangle \langle \tau(\mu) |
\end{equation}
We now apply losses to the $\mu = 1$ interferometer. To do this, we pair up each mode of this interferometer with an environment mode that is initially in the vacuum. As we apply loss to the $\mu = 1$ interferometer only, we can peel off this term from the tensor product and write 
\begin{equation}
\rho_\text{total}  = \frac{1}{n!}\sum_{\sigma, \tau \in S_n} (-1)^{\sigma + \tau} |\sigma(1),0\rangle \langle \tau(1), 0| \otimes   \bigotimes_{\mu = 2}^n |\sigma(\mu)\rangle \langle \tau(\mu)|,
\end{equation}
where the second label in each state is the associated environment mode. We then apply the loss unitary to the interferometer to mix its modes with the environment modes.
\begin{equation}
\rho_\text{total} \rightarrow \frac{1}{n!}\sum_{\sigma, \tau \in S_n} (-1)^{\sigma + \tau}   \left( U_\mathrm{L} |\sigma(1), 0\rangle \langle \tau(1), 0| U_\mathrm{L}^\dagger \right) \otimes \bigotimes_{\mu = 2}^n |\sigma(\mu)\rangle \langle \tau(\mu)|.
\end{equation}
Now let us trace out the environment to understand how losses affect the photons. We get
\begin{equation}
\rho_\text{out} = \mathrm{Tr}_\mathrm{E}(\rho_\text{total}) =  \frac{1}{n!}\sum_{\sigma, \tau \in S_n}  (-1)^{\sigma + \tau} \mathrm{Tr}_\mathrm{E} \left[U_\mathrm{L} |\sigma(1), 0\rangle \langle \tau(1), 0| U_\mathrm{L}^\dagger  \right]  \otimes \bigotimes_{\mu = 2}^n |\sigma(\mu)\rangle \langle \tau(\mu) |.
\end{equation}
As $|\sigma(\mu)\rangle$ corresponds to a single photon in the mode $\sigma(\mu)$, using Eq.~\eqref{eq:fermionic_loss_channel} we see losses act on this quite simply as 
\begin{equation}
U_\mathrm{L} |\sigma(1),0\rangle = \sqrt{1 - \eta} |\sigma(1),0\rangle + \sqrt{\eta} |0,\sigma(1)\rangle,
\end{equation}
where $|0,\sigma(\mu)\rangle = e_{\mu \sigma(\mu)}^\dagger |0\rangle$. This gives us 
\begin{equation}
U_\mathrm{L} |\sigma(1), 0\rangle \langle \tau(1), 0| U_\mathrm{L}^\dagger = (1 - \eta) |\sigma(1) ,0\rangle \langle \tau(1),0 |  + \eta |0, \sigma(1) \rangle  \langle 0, \tau(1)| + \ldots
\end{equation}
where $\ldots$ represents the terms that will vanish when we trace out the environment, to give us
\begin{equation}
\mathrm{Tr}_\mathrm{E} \left[ U_\mathrm{L} |\sigma(1)\rangle \langle \tau(1)| U^\dagger_\mathrm{L} \right] = (1- \eta)|\sigma(1)\rangle \langle \tau(1)| + \eta |0\rangle \langle 0 | \delta_{\sigma(1) \tau(1)} 
\end{equation}
Plugging this back in gives us the output state
\begin{equation}
\begin{aligned}
 \rho_\text{out} & = \frac{1 - \eta}{n!}\sum_{\sigma, \tau \in S_n} (-1)^{\sigma + \tau} \bigotimes_{\mu = 1}^n |\sigma(\mu)\rangle \langle \tau(\mu)| \\
& + \frac{\eta}{n!} |0\rangle \langle 0 | \otimes \underbrace{\sum_{\sigma, \tau \in S_n} (-1)^{\sigma + \tau} \delta_{\sigma(1) \tau(1)}  \bigotimes_{\mu = 2}^n |\sigma(\mu)\rangle \langle \tau(\mu)|}_{\rho'}
\end{aligned}
\end{equation}
The first term is proportional to the input state. In the second term, the Kronecker delta $\delta_{\sigma(1) \tau(1)}$ restricts us to permutations $\sigma$ and $\tau$ such that $\sigma(1) = \tau(1)$. Let us focus on this term: we can rewrite this as 
\begin{equation}
\rho'  = \sum_{k = 1}^n \sum_{\sigma, \tau \in S_n^k} (-1)^{\sigma + \tau} \bigotimes_{\mu = 2}^n |\sigma(\mu)\rangle \langle \tau(\mu)|
\end{equation}
where $S_n^k = \{ \sigma \in S_n : \sigma(1) = k \}$. Given any permutation $\sigma \in S_n^k$ we can relate it to an element of $S_{n-1}$ as
\begin{equation}
\sigma = \tilde{\sigma} \sigma_k
\end{equation}
where $\sigma_k \in S_n$ is some permutation such that $\sigma_k(1) = k$ that we take to be fixed and $\tilde{\sigma} \in S_{n-1} \subset S_n$ is a permutation of the remaining $n - 1$ terms that leaves the first element alone. For example, if we have $n = 4$ terms and $k = 2$ the permutation $\sigma \in S^k_n$ such that $\sigma(1,2,3,4) = (2,4,1,3)$ can be written as 
\begin{equation}
(1, 2, 3, 4) \xrightarrow[\sigma_2]{}  (2, 1, 3, 4) \xrightarrow[\tilde{\sigma}]{} (2, \tilde{\sigma}(1,3 ,4 ) ) = (2, 4, 1, 3) 
\end{equation}
Therefore, a sum over $S_n^k$ is equivalent to a sum over $S_{n-1}$, except now the permutations are permuting the elements from the set $I_k = \{1,2,\ldots, n\} \setminus \{k \}$. We can rewrite the sum as
\begin{equation}
\rho' = \sum_{k = 1}^n \sum_{\tilde{\sigma}, \tilde{\tau} \in S_{n-1}} (-1)^{\tilde{\sigma} + \tilde{\tau}} \bigotimes_{\mu \in I_k } |\tilde{\sigma}(\mu)\rangle \langle \tilde{\tau}(\mu)|
\end{equation}
where we used the fact that the sign term is
\begin{equation}
(-1)^{\sigma + \tau} = (-1)^{\sigma} (-1)^{\tau} = (-1)^{\tilde{\sigma} \sigma_k} (-1)^{\tilde{\tau} \sigma_k} = (-1)^{2 \sigma_k} (-1)^{\tilde{\sigma}} (-1)^{\tilde{\tau}} = (-1)^{\tilde{\sigma} + \tilde{\tau}}
\end{equation}
This state is a sum over fermionic resource states, instead constructed from $n-1$ photons distributed across modes $I_k$, namely
\begin{equation}
|\psi_k\rangle = \frac{1}{\sqrt{(n-1)!}} \sum_{\sigma \in S_{n-1}} (-1)^\sigma \bigotimes_{\mu \in I_k} |\sigma(\mu)\rangle
\end{equation}
We can therefore write
\begin{equation}
\rho' = (n-1)! \sum_{k = 1}^n |\psi_k\rangle \langle \psi_k|
\end{equation}
Pulling everything together, we get
\begin{equation}
\rho_\text{out} = (1- \eta) \rho_\text{in} + |0\rangle \langle 0| \otimes \frac{\eta}{n} \sum_{k = 1}^n |\psi_k \rangle \langle \psi_k|
\end{equation}
As each term in this sum is an anti-symmetric fermionic state, we see that a single loss does not affect the fermionic nature of the state. As multiple losses is equivalent to multiple single losses like this, we see that with losses across all modes the fermionic nature will be preserved.

\section{Samples estimation for Rotosolve \label{app:estimation}}
\label{supp:samples}
In this section we evaluate how many samples are required to perform a single iteration of Rotosolve. Before we begin, let us define an estimator $\hat{\theta}$ to be a function of measured data to estimate the true unknown $\theta$. We say that $\hat{\theta}$ is unbiased if $\mathbb{E}[\hat{\theta}] = \theta$, otherwise it is biased. We define the bias as $b = \mathbb{E}[\hat{\theta}] - \theta$. 

The parameter update rule for Rotosolve is 
\begin{align}
x_0 & = -\pi/2 + \tan^{-1}(Y/X)  \\
X & = f(\pi/2) - f(-\pi/2) \\
Y & = 2f(0) - f(\pi/2) - f(-\pi/2)
\end{align}
where $f(x) = \langle \psi_\text{in} | U^\dagger(x) H U(x) |\psi_\text{in}\rangle$ is the cost function, defined as the expectation value of the observable $H$ given the parametrised unitary $U(x)$. First we introduce the estimator
\begin{equation}
\hat{f}(x) = \frac{1}{N_S} \sum_{i = 1}^{N_S} E_i.
\end{equation}
This represents how we calculate $f(x)$ in practice  by taking the mean of the set of measurements we make of the cost function, where $E_i$ is the $i$th measurement outcome of the observable $H$ and $N_S$ is the number of samples we take. This is an unbiased estimator, as 
\begin{equation}
\mathbb{E}[\hat{f}(x)] = \frac{1}{N_S} \sum_{i = 1}^{N_S} \mathbb{E}[E_i] = \frac{1}{N_S} \sum_{i = 1}^{N_S} f(x) = f(x)
\end{equation}
Additionally, the variance of this estimator is
\begin{equation}
\mathrm{Var}[\hat{f}(x)] = \frac{1}{N_S^2} \sum_{i = 1}^{N_S}  \mathrm{Var}[E_i] = \frac{1}{N_S} \sum_{i = 1}^{N_S} \sigma^2(x) = \frac{\sigma^2(x)}{N_S} \label{eq:estimator_variance}
\end{equation}
where $\sigma^2(x)$ is defined as the variance of the observable $H$ given the system has parameter $x$ and we assumed that our measurements $\{ E_i \}$ are uncorrelated. We now make the assumption that this quantity does not depend upon $x$ by taking $\sigma^2(x) \equiv \sigma^2$; alternatively we can define this quantity as $\sigma^2 = \max_x \sigma^2(x)$ to give us the worst case scenario.

Next we define the estimators used in Rotosolve. We define them quite simply as
\begin{align}
\hat{x}_0 & = -\pi/2 + \tan^{-1}(\hat{Y}/\hat{X})  \\
\hat{X} & = \hat{f}(\pi/2) - \hat{f}(-\pi/2) \\
\hat{Y} & = 2\hat{f}(0) - \hat{f}(\pi/2) - \hat{f}(-\pi/2)
\end{align}
The estimators $\hat{X}$ and $\hat{Y}$ are clearly unbiased, however $\hat{x}_0$ will be biased because it is a non-linear function of $\hat{X}$ and $\hat{Y}$. However, if the number of samples $N_S$ is large, then according to Eq.~\eqref{eq:estimator_variance} the variance of $\hat{X}$ and $\hat{Y}$ will be small so the values of the estimator must lie close to the mean. We approximate the function by Taylor expanding about the means to second order. If we have a function of a vector of random variables $\hat{\mathbf{X}}$ then we have
\begin{equation}
F(\hat{\mathbf{X}}) \approx F(\boldsymbol{\mu}) + (\hat{\mathbf{X}} - \boldsymbol{\mu}) \nabla F + \frac{1}{2} (\hat{\mathbf{X}} - \boldsymbol{\mu})^T \mathbf{H} (\hat{\mathbf{X}} - \boldsymbol{\mu}) \label{eq:taylor_series}
\end{equation}
where $\mathbf{H}_{ij} = \partial_i \partial_j F$ is the Hessian matrix. If the vector of estimators is unbiased, then when we take the expectation value of the Taylor expansion then the first order term will vanish giving us
\begin{equation}
\begin{aligned}
\mathbb{E}[F(\hat{\mathbf{X}})] & = F(\boldsymbol{\mu}) + \frac{1}{2} \mathbb{E} \left[ (\hat{\mathbf{X}} - \boldsymbol{\mu})^T \mathbf{H} (\hat{\mathbf{X}} - \boldsymbol{\mu}) \right]\\
& = F(\boldsymbol{\mu}) + \frac{1}{2} \sum_{i,j} \mathbf{H}_{ij} \mathbb{E}[ (\hat{X}_i - \mu_i) (\hat{X}_j - \mu_j)] \\
& = F(\boldsymbol{\mu}) + \frac{1}{2} \mathrm{Tr}(C \mathbf{H}) \\
& \equiv F(\boldsymbol{\mu}) + b
\end{aligned}
\end{equation}
where $C_{ij}$ is the covariance matrix. This gives us the bias $b$ to second order. The Hessian matrix for the function $\tan^{-1}(\hat{Y}/\hat{X})$, after evaluating it on $\hat{X} = \mu_X$ and $\hat{Y} = \mu_Y$ is given by 
\begin{equation}
\mathbf{H} = \frac{1}{(\mu_X^2 + \mu_Y^2)^2} \begin{pmatrix}
2 \mu_X \mu_Y & \mu_Y^2 - \mu_X^2 \\
\mu_Y^2 - \mu_X^2 & - 2 \mu_X \mu_Y
\end{pmatrix}
\end{equation}
Now us calculate the covariance matrix of these random variables. We have
\begin{align}
& \begin{aligned}
\mathrm{Var}[\hat{X}] & = \mathrm{Var}[\hat{f}(\pi/2) - \hat{f}(-\pi/2)] \\
&  = \mathrm{Var}[\hat{f}(\pi/2)] + \mathrm{Var}[\hat{f}(-\pi/2)] \\
&= \frac{2 \sigma^2}{N_S} 
\end{aligned} \\
& \begin{aligned}
\mathrm{Var}[\hat{Y}] & = \mathrm{Var}[2\hat{f}(0) - \hat{f}(\pi/2) - \hat{f}(-\pi/2)] \\
&  = 4 \mathrm{Var}[\hat{f}(0)] + \mathrm{Var}[\hat{f}(\pi/2)]  + \mathrm{Var}[\hat{f}(- \pi/2)]  \\
&=\frac{6 \sigma^2}{N_S} 
\end{aligned}
\end{align}
where we used the fact that $\hat{f}(x)$ and $\hat{f}(y)$ are uncorrelated for $x \neq y$ as they correspond to different experiments. The covariance is similarly calculated as
\begin{equation}
\begin{aligned}
\mathrm{Cov}[\hat{X}, \hat{Y}] & = \mathrm{Cov}[\hat{f}(\pi/2) - \hat{f}(-\pi/2) , 2\hat{f}(0) - \hat{f}(\pi/2) - \hat{f}(-\pi/2)] \\
& = -\mathrm{Var}[\hat{f}(\pi/2)] + \mathrm{Var}[\hat{f}(-\pi/2)] \\
& = 0
\end{aligned} 
\end{equation}
to give us the covariance matrix 
\begin{equation}
C = \begin{pmatrix}
\mathrm{Var}[\hat{X}] & \mathrm{Cov}[\hat{X}, \hat{Y}] \\
\mathrm{Cov}[\hat{X}, \hat{Y}] & \mathrm{Var}[\hat{Y}]
\end{pmatrix} = \frac{2 \sigma^2}{N_S} \begin{pmatrix} 1 & 0 \\ 0 & 3 \end{pmatrix} \label{eq:covariance_matrix}
\end{equation}
Pulling this together, we find the bias of the estimator $\hat{x}_0$ to second order is given by
\begin{equation}
b = \frac{1}{2} \mathrm{Tr}(C \mathbf{H}) = \frac{4 \sigma^2 \mu_X \mu_Y}{N_S(\mu_X^2 + \mu_Y^2)^2}
\end{equation}
and we see that the bias decreases as with increase the number of samples $N_S$.

Now we evaluate the variance of our estimator $\hat{x}_0$. As this is a non-linear function of the random variables $\hat{X}$ and $\hat{Y}$, let us use the Delta method~\cite{Vaart_1998} by Taylor expanding the function about its mean. Referring back to Eq.~\eqref{eq:taylor_series} for the Taylor expansion, let us consider the first order term only, then when we take the variance we get
\begin{equation}
\mathrm{Var}[F(\hat{\mathbf{X}})] \approx \left(  \frac{\partial F}{\partial \hat{X}}  \right)^2 \mathrm{Var}[\hat{X}] +  \left(  \frac{\partial F}{\partial \hat{Y}}  \right)^2 \mathrm{Var}[\hat{Y}] + 2 \frac{\partial F}{\partial \hat{X}}  \frac{\partial F}{\partial \hat{Y}} \mathrm{Cov}[\hat{X}, \hat{Y}] 
\end{equation}
Applying this to the function $F(\hat{X}, \hat{Y}) = \tan^{-1}(\hat{Y}/\hat{X})$ we get
\begin{equation}
\mathrm{Var}[\hat{x}_0] \approx \frac{\mu_X^2 \mathrm{Var}[\hat{Y}] + \mu_Y^2 \mathrm{Var}[\hat{X}] - 2 \mu_X \mu_Y \mathrm{Cov}[\hat{X}, \hat{Y}]}{(\mu_X^2 + \mu_Y^2)^2}
\end{equation}
Now insert in the covariance matrix of Eq.~\eqref{eq:covariance_matrix} to get
\begin{equation}
\begin{aligned}
\mathrm{Var}[\hat{x}_0] & = \frac{2\sigma^2 (3 \mu_X^2 + \mu_Y^2)}{N_S(\mu_X^2 + \mu_Y^2)^2} 
\end{aligned}
\end{equation}
Using the fact that $\hat{f}$ is an unbiased estimator, we have
\begin{align}
\mu_X & = \mathbb{E}[\hat{X}] = f(\pi/2) - f(-\pi/2) \\
\mu_Y & =\mathbb{E}[\hat{Y}] = 2f(0) - f(\pi/2) - f(-\pi/2).
\end{align}
The final step is to use the fact that we know the functional form of $f(x)$ is given by $f(x) = A \sin(x - \phi) + B$, which is the functional form of the cost landscape that Rotosolve assumes from the beginning, therefore we have
\begin{align}
f(\pi/2) & = A \cos \phi + B\\
f(-\pi/2) & = - A\cos \phi + B \\
f(0) & = -A \sin \phi + B
\end{align}
we get the values
\begin{align*}
\mu_X & = 2A \cos \phi \\
\mu_Y & = -2A \sin \phi 
\end{align*}
Plugging these values into the variance gives us 
\begin{equation}
\mathrm{Var}[\hat{x}_0] = \frac{\sigma^2}{2N_SA^2}( 1+ 2 \cos^2 \phi)
\end{equation}
The amplitude $A$ of the cost function $f(x)$ is determined by the eigenvalues of the observable, which is because the cost function is defined as $f(x) = \langle H \rangle$. Therefore, the largest value that $A$ could be is $A = (E_\text{max} - E_\text{min})/2$ which is equal to the range the eigenvalues are found in. The phase shift $\phi$ is a function of the parameters of the interferometer that are held fixed when we perform the Rotosolve iteration. Therefore, $\phi$ can be any value in the range $[0, 2 \pi)$. Therefore, the maximum variance is given by 
\begin{equation}
\mathrm{Var}[\hat{x}_0] = \frac{3\sigma^2}{2N_S A^2}
\end{equation}
If we want to estimate $x_0$ to a mean squared error of $\epsilon^2$, then this means
\begin{equation}
\mathrm{MSE}[\hat{x}_0] = \mathbb{E}[(\hat{x}_0 - x_0)^2] = \mathrm{Var}[\hat{x}_0] + b^2 \equiv \epsilon^2.
\end{equation}
The MSE is given by
\begin{equation}
\mathrm{MSE}[\hat{x}_0] = \frac{3 \sigma^2}{2N_S A^2} + \frac{\sigma^4}{4 N_S^2 A^4}  \approx \frac{3 \sigma^2}{2N_S A^2}
\end{equation}
where we keep the $O(1/N_S)$ term only. If we rearrange this, we get the number of samples given by
\begin{equation}
N_S = \frac{3 \sigma^2}{2 \epsilon^2 A^2}
\end{equation} 
This tells us that as the amplitude $A$ increases we need less samples as long as $A > \sigma$.

\section{Proof of Theorem~\ref{thm:barren_plateaus} \label{app:barren_plateaus}}

\subsection{Haar measure moments}
In this section we derive some useful identities to be used later. 

\begin{lemma}
\label{thm:haar_measure_moments}
Let $U : \mathrm{U}(N) \rightarrow \mathrm{U}(d)$ be a reducible unitary representation of $\mathrm{U}(N)$ on the vector space $V$ with dimension $\mathop{\mathrm{dim}}(V) = d$ and orthonormal basis $\{ |i \rangle \}_{i = 1}^d$, where each irreducible representation has a multiplicity of one, then 
\begin{equation}
\int_{\mathrm{U}(N)} du U_{ij}(u) U^*_{kl}(u) = \sum_{\alpha} \frac{1}{d_\alpha} P^\alpha_{ik} P^\alpha_{lj} ,
\end{equation}
where $P^\alpha$ is the projector onto the irrep $\alpha$, $U_{ij}(u) = \langle i | U(u)|j\rangle$ and $P^\alpha_{ij} = \langle i | P^\alpha | j \rangle$ are the matrix elements with respect to the orthonormal basis, $d_\alpha$ is the dimension of the irrep $\alpha$ and $du$ is the Haar measure.
\end{lemma}
\begin{proof}
If the representation is reducible there exists a basis that block-diagonalises all $U(u)$ for all $u \in \mathrm{U}(N)$. Let this basis be denoted by $\{ |\alpha, x, n \rangle \}$, where $\alpha$ labels the irrep, $x$ labels the any multiplicity of the irreps, and $n$ labels the states within this basis are given by
\begin{equation}
\langle \alpha , x , m | U(u) | \beta , y , n \rangle = \delta_{\alpha \beta} \delta_{xy} U^{\alpha}_{mn}(u),
\end{equation}
where $U^\alpha_{mn}$ are the matrix elements of the irrep $\alpha$. If this irrep has non-trivial multiplicity, we assume that the matrix elements for each one are equal.

Next we introduce a unitary change of basis $C_{i(\alpha,x ,m)} = \langle i | \alpha, x , m \rangle$ where the triplet $(\alpha , x , m)$ is viewed as a multi-index. These components are sometimes called Clebsch-Gordon coefficients. We can then rewrite the matrix elements with respect to the original basis as
\begin{equation}
U_{ij}(u)  = \langle i | U(u) |j\rangle = \sum_{ \alpha} \sum_{x} \sum_{m, n} C_{i (\alpha, x , m)} C^*_{j (\alpha, x , n) }U^{\alpha}_{mn}(u).
\end{equation}
If we insert this into the integral, we get
\begin{equation}
\begin{aligned}
\int_{\mathrm{U}(N)} du U_{ij}(u)U_{kl}^*(u) & = \int du \left( \sum_{ \alpha} \sum_{x} \sum_{m, n}C_{i (\alpha, x , m)} C^*_{j (\alpha, x , n) }U^{\alpha}_{mn}(u) \right) \left( \sum_{ \beta} \sum_{y} \sum_{p, q}C^*_{k (\beta, y , p)} C_{l (\beta, y , q) }U^{\beta *}_{pq}(u) \right) \\
& =  \sum_{ \alpha , \beta} \sum_{x,y} \sum_{m,n,p,q} C_{i (\alpha, x , m)} C^*_{j (\alpha, x , n) } C^*_{k (\beta, y , p)} C_{l (\beta, y , q) }  \underbrace{\int_{\mathrm{U}(N)} du U^{\alpha}_{mn}(u)  U^{\beta *}_{pq}(u)}_{\frac{1}{d_\alpha} \delta^{\alpha \beta} \delta_{mp} \delta{nq}} \\
& = \sum_{\alpha} \frac{1}{d_\alpha} \sum_{x,y} \sum_{m,n} C_{i (\alpha, x , m)} C^*_{j (\alpha, x , n) } C^*_{k (\alpha, y , m)} C_{l (\alpha, y , n) }, 
\end{aligned}
\end{equation}
where we use the orthogonality relationship of matrix elements of irreps~\cite{georgi} to evaluate the integral in the second line and $d_\alpha$ is the dimension of the irrep $\alpha$. Define the operator
\begin{equation}
P^\alpha_{xy} = \sum_{m} |\alpha, x,m\rangle \langle \alpha, y ,m |,
\end{equation}
so we can rewrite the integral as
\begin{equation}
\int_{\mathrm{U}(N)} du U_{ij}(u) U^*_{kl}(u) = \sum_\alpha \sum_{x, y} \frac{1}{d_\alpha} (P^\alpha_{xy})_{ik} (P^\alpha_{yx})_{lj}. 
\end{equation}
Additionally, if we assume that the multiplicity of each irrep is unity then we can drop the indices $x$ and $y$, we then get
\begin{equation}
\int_{\mathrm{U}(N)} du U_{ij}(u) U^*_{kl}(u) = \sum_\alpha \frac{1}{d_\alpha} P^\alpha_{ik} P^\alpha_{lj},
\end{equation}
where $P^\alpha_{ij}$ are the matrix elements of the projector onto irrep $\alpha$.
\end{proof}

\begin{lemma}
\label{thm:haar_twirl}
Let $U : \mathrm{U}(N) \rightarrow \mathrm{U}(d)$ be a reducible representation of $\mathrm{U}(N)$ on the vector space $V$ with dimension $\mathop{\mathrm{dim}}(V) = d$ and orthonormal basis $\{ |i \rangle \}_{i = 1}^d$, where each irreducible representation has a multiplicity of one, and let $A$ be a linear operator on $V$ then
\begin{equation}
I = \int_{\mathrm{U}(N)} du  U(u) A U^\dagger(u)  = \sum_\alpha \frac{1}{d_\alpha} \mathrm{Tr}(P^\alpha A) P^\alpha.
\end{equation}
\begin{proof}
Let us calculate the matrix elements of $I$. We have
\begin{equation}
\begin{aligned}
I_{ij} & = \int_{\mathrm{U}(N)} du \sum_{k,l}   U_{ik}(u) A_{kl} U_{lj}^\dagger(u) \\
& =  \sum_{k,l}  A_{kl} \int_{\mathrm{U}(N)} du   U_{ik}(u)  U_{jl}^*(u) \\
& = \sum_{k,l} A_{kl} \sum_\alpha \frac{1}{d_\alpha}  P_{lk}^\alpha  P_{ij}^\alpha  \\
& = \sum_{\alpha} \frac{1}{d_\alpha}  \mathrm{Tr}(P^\alpha A) P_{ij}^\alpha.
\end{aligned}
\end{equation} 
where in the third equality we used Lemma~\ref{thm:haar_measure_moments}.
\end{proof}
\end{lemma}

\subsection{Calculating the variance}
Given a state $\rho$ and an observable $H$, the cost function is defined as
\begin{equation}
E(\boldsymbol{\theta}) = \mathrm{Tr}\left[ U(\boldsymbol{\theta}) \rho U^\dagger(\boldsymbol{\theta}) H \right].
\end{equation}
We calculate the variance of this quantity by assuming that our parametrised unitaries $U(\boldsymbol{\theta})$ form a $2$-design, so averaging over the parameters is equivalent to averaging the unitaries with respect to the Haar measure~\cite{mcclean2018barren}. In other words, given some operator $A$ on a $t$-fold tensor product space and an ensemble of parametrised unitaries $\{ p(\boldsymbol{\theta}), U(\boldsymbol{\theta}) \}$ then
\begin{equation}
\int d \boldsymbol{\theta}  p(\boldsymbol{\theta}) U^{\otimes t}(\boldsymbol{\theta}) A U^{\dagger \otimes t}(\boldsymbol{\theta}) = \int_{\mathrm{U}(N)} du U^{\otimes t}(u) A U^{\dagger \otimes t}(u).
\end{equation}
where $du$ is the Haar measure. This assumption is perfectly valid to make in linear optics, because we can use a universal linear interferometer which can encode any linear unitary, so we can sample with respect to the Haar measure. Let us assume that our unitaries form a representation of $\mathrm{U}(N)$ as $U = U(u)$, where $u \in \mathrm{U}(N)$. Then the first moment of the cost function is given by
\begin{equation}
\begin{aligned}
\mathbb{E}_{\boldsymbol{\theta}}[E(\boldsymbol{\theta})] & = \int d \boldsymbol{\theta} p(\boldsymbol{\theta}) \mathrm{Tr}[ U(\boldsymbol{\theta}) \rho U^\dagger(\boldsymbol{\theta}) H ] \\
 & = \int_{\mathrm{U}(N)} du \mathrm{Tr}[ U(u) \rho U^\dagger(u) H ]  \\
& = \mathrm{Tr} \left[ \left( \int_{\mathrm{U}(N)} du U(u) \rho U^\dagger(u) \right) H \right] \\
& = \mathrm{Tr} \left[ \frac{1}{D} \mathrm{Tr}(\rho) H \right] \\
& = \frac{1}{D} \mathrm{Tr}(H)
\end{aligned}
\end{equation} 
where going from the third to the fourth line we used Lemma~\ref{thm:haar_twirl} and the fact that the projectors are the identity as $U(u)$ forms an irrep and $D$ is the dimension of the Hilbert space, and going to the final line we used the fact that $\mathrm{Tr}(\rho) = 1$.

The second moment gives us 
\begin{equation}
\begin{aligned}
\mathbb{E}_{\boldsymbol{\theta}}[E^2(\boldsymbol{\theta})] & = \int d \boldsymbol{\theta} p(\boldsymbol{\theta}) \mathrm{Tr}\left[ U(\boldsymbol{\theta}) \rho U^\dagger(\boldsymbol{\theta}) H \right]^2 \\
& =  \int d \boldsymbol{\theta} p(\boldsymbol{\theta}) \mathrm{Tr}\left[ U(\boldsymbol{\theta}) \rho U^\dagger(\boldsymbol{\theta}) H \otimes U(\boldsymbol{\theta}) \rho U^\dagger(\boldsymbol{\theta}) H \right]  \\
& =  \mathrm{Tr}\left[ \left(\int d \boldsymbol{\theta} p(\boldsymbol{\theta}) U^{\otimes 2}(\boldsymbol{\theta}) \rho^{\otimes 2} U^{\dagger \otimes 2}(\boldsymbol{\theta}) \right) H^{\otimes 2}  \right] \\
& = \mathrm{Tr} \left[ \left( \int_{\mathrm{U}(N)} du   U^{\otimes 2}(u) \rho^{\otimes 2} U^{\dagger  \otimes 2}(u)  \right) H^{\otimes 2} \right] \\
& = \mathrm{Tr} \left[ \sum_\alpha \frac{1}{d_\alpha} \mathrm{Tr}\left(P^\alpha \rho^{\otimes 2} \right) P^\alpha H^{\otimes 2} \right] \\
& = \sum_{\alpha} \frac{1}{d_\alpha} \mathrm{Tr}\left( P^\alpha \rho^{\otimes 2} \right) \mathrm{Tr} \left( P^\alpha H^{\otimes 2} \right)
\end{aligned}
\end{equation}
where going from the first to the second line we used the fact that $\mathrm{Tr}(A)^2 = \mathrm{Tr}(A \otimes A)$, going from the third to fourth lines we used the $2$-design assumption and going from the fourth to the fifth line we used Lemma~\ref{thm:haar_twirl} and noted that $U^{\otimes 2}(u) \equiv U(u) \otimes U(u)$ forms a reducible tensor product representation and we made the assumption that each irrep has unit multiplicity in anticipation of applying this to linear optics. 

Pulling everything together, the variance is given by
\begin{equation}
\mathrm{Var}_{\boldsymbol{\theta}}[E(\boldsymbol{\theta})] = \sum_{\alpha} \frac{1}{d_\alpha} \mathrm{Tr}\left( P^\alpha \rho^{\otimes 2} \right) \mathrm{Tr} \left( P^\alpha H^{\otimes 2} \right) - \frac{ \mathrm{Tr}^2(H)}{D^2}. \label{eq:variance}
\end{equation}
which concludes the proof of Theorem~\ref{thm:barren_plateaus}.

\begin{theorem}
\label{thm:invariance_of_variance}
The variance of the cost function is invariant under $H \rightarrow \tilde{H} =  U(u) H U^\dagger(u)$ where $u \in \mathrm{U}(N)$.
\begin{proof}
Let us conjugate the Hamiltonian with $U(v)$ for $v \in \mathrm{U}(N)$. The first moment is
\begin{equation}
\begin{aligned}
\mathbb{E}_{\boldsymbol{\theta}}[\tilde{E}(\boldsymbol{\theta})] & = \int_{\mathrm{U}(N)} du \mathrm{Tr}[ U(u) \rho U^\dagger(u) \tilde{H} ] \\
& = \int_{\mathrm{U}(N)} du \mathrm{Tr}[ U(u) \rho U^\dagger(u) U(v) H U^\dagger(v) ] \\
& = \int_{\mathrm{U}(N)} du \mathrm{Tr}[ U(v^\dagger u) \rho  U^\dagger(v^\dagger u) H  ]
\end{aligned}
\end{equation}
Let $w = v^\dagger u$ and use the fact that the Haar measure is left invariant, so $dw = d(v^\dagger u) = du$. We then get
\begin{equation}
\mathbb{E}_{\boldsymbol{\theta}}[\tilde{E}(\boldsymbol{\theta})]  = \int_{\mathrm{U}(N)} dw \mathrm{Tr}[ U(w) \rho U^\dagger(w) H] = \mathbb{E}_{\boldsymbol{\theta}}[E(\boldsymbol{\theta})]
\end{equation}
A similar result applies to the second moment, hence the variance is invariant.
\end{proof} 
\end{theorem}

\subsection{Application to linear optics}
\subsubsection{Representations of linear optics}
In linear optics, if we have an $N$-mode interferometer and $n$ particles, then we have a representation is defined of $\mathrm{U}(N)$ on the Fock space $\mathcal{F}$. This representation acts on the ladder operators in the Heisenberg picture as 
\begin{equation}
U(u) a_i^\dagger U^\dagger(u) = \sum_{j = 1}^N u_{ji} a_j^\dagger
\end{equation}
To see this is a representation, we have explicitly 
\begin{equation}
\begin{aligned}
U(v)U(u) a_i^\dagger U^\dagger(u) U^\dagger(v) & = \sum_{j = 1}^N u_{ji} U(v) a_j^\dagger U^\dagger(v)  \\
& = \sum_{j = 1}^N \sum_{k = 1}^N u_{ji} v_{kj} a_k^\dagger \\
& = \sum_{k = 1}^N (v u)_{ki} a_k^\dagger \\
& = U(vu) a_i^\dagger U^\dagger(vu)
\end{aligned}
\end{equation}
therefore, $U(v) U(u) = U(vu)$ and $U(\mathbb{I}_N) = \mathbb{I}_d$, where $d$ is the dimension of the representation. Moreover, this representation commutes with the total number operator as
\begin{equation}
\begin{aligned}
U \left( \sum_{i = 1}^N a^\dagger_i a_i \right) U^\dagger & = \sum_{i =1}^N \sum_{k = 1}^N \sum_{l = 1}^N u_{ji} u^*_{ki} a^\dagger_j a_k  \\
& = \sum_{j = 1}^N \sum_{k = 1}^N (u u^\dagger)_{jk} a_j^\dagger a_k \\
& = \sum_{k = 1}^N a_k^\dagger a_k
\end{aligned}
\end{equation}
Therefore, the subspace of fixed particle number is an invariant subspace under the action of $U$. 

Let $V = \mathbb{C}^N$ be the Hilbert space of a single particle in an interferometer of $N$ modes, the Fock space $\mathcal{F}$, and the representation of $\mathrm{U}(N)$, decomposes as
\begin{equation}
\mathcal{F} = \bigoplus_n \mathcal{H}_n, \quad
\mathcal{H}_n = \begin{cases}
\mathrm{Sym}^n (V) & \text{bosons} \\
\wedge^n(V) & \text{fermions}
\end{cases}, \quad \mathrm{dim}(\mathcal{H}_n) = \begin{cases} { N + n - 1 \choose n } & \text{bosons} \\
{ N \choose n } & \text{fermions}
\end{cases},
\end{equation}
where $\mathcal{H}_n$ is the $n$-particle subspace and ($\wedge^n$) $\mathrm{Sym}^n$ is the (anti-)symmetrisation $n$-fold tensor product operation. Importantly, these subspaces forms \textit{irreducible} representations. This can be seen quite simply: the general form of the linear optical transformation is given by the exponential of a quadratic Hamiltonian that can be rewritten as
\begin{equation}
\begin{aligned}
U & = \exp\left( i \sum_{i,j} h_{ij} a_i^\dagger a_j \right) \\
&  = \exp\left( i \sum_i h_{ii} a_i^\dagger a_i + \frac{i}{2} \sum_{i \neq j} R_{ij} (a_i^\dagger a_j + a_j^\dagger a_i) + i I_{ij} (a_i^\dagger a_j - a_j^\dagger a_i ) \right) \\
& \equiv \exp \left( i \sum_i h_{ii} \hat{n}_i + \frac{i}{2} \sum_{i \neq j} (R_{ij} X_{ij} + I_{ij} Y_{ij}) \right)
\end{aligned}
\end{equation}
where $h_{ij} = R_{ij} + iI_{ij}$ for $i \neq j$ and $h_{ii} \in \mathbb{R}$. The corresponding Lie algebra representation $\mathfrak{g}$ is given by the tangent vectors to the identity which is $\mathfrak{g} = \mathrm{span}_\mathbb{R} \{ \hat{n}_i, X_{ij}, Y_{ij} \}$. There are $N^2$ generators in total and these obey the $\mathrm{U}(N)$ algebra.

We now complexify the Lie algebra by taking complex linear combinations of the generators. We can divide the generators into the Cartan subalgebra of mutually commuting generators $\mathfrak{h} = \mathrm{span}_\mathbb{C}\{ \hat{n}_i \}_{i=1}^N$, and the remaining generators we take linear combinations of to form ladder operators as
\begin{equation}
E_{ij} = \frac{1}{2}(X_{ij} + i Y_{ij}) = a_i^\dagger a_j,
\end{equation}
of which there are $N(N-1)/2$. These obey the algebra
\begin{align}
[H_i, E_{jk} ] & =( \delta_{ij} - \delta_{ik}) E_{jk} \\
[E_{ij}, E_{kl}] & = \delta_{jk} E_{il} - \delta_{il} E_{kj} \\
[H_i, H_j] & = 0
\end{align}
We divide the operators $E_{ij}$ up into a set of raising and lowering operators by defining $E_{ij}$ to be a raising operator if $i < j$ and a lowering for $j > i$ which is its hermitian conjugate. The highest-weight state is defined as the state $|\lambda\rangle$ that is annihilated by all raising operators. For $n$ particles in $N$ modes this is given by
\begin{equation}
|\lambda \rangle = \begin{cases}
|n,0,\ldots,0,0, \ldots , 0 \rangle & \text{bosons} \\
|\underbrace{1,1,\ldots,1}_{n},0, \ldots, 0\rangle & \text{fermions}
\end{cases}
\end{equation}
as $E_{ij} |\lambda\rangle = 0$ for all $i < j$. This representation is irreducible, as every state in $\mathcal{H}_n$ can be reached by applying lowering operators $E_{ij}$ to $|\lambda\rangle$. This is because the lowering operators allow us to shuffle the particles around the modes as they take the form of hopping operators. In conclusion, there is no basis that block diagonalises the generators, hence it is an irrep. 

Without loss of generality we can study the barren plateau problem by restricting ourselves to the $n$-particle subspace $\mathcal{H}_n$ as in a linear optics experiment the number of particles will be fixed anyway. Referring back to Eq.~\eqref{eq:variance}, we see that calculation of the variance requires us to perform a \textit{fourth} moment integral in the matrix elements of the representation, which is equivalent to a \textit{second} moment calculation in the tensor product representation instead and is given in terms of a sum over irreps of this tensor product. Therefore, we must understand how $\mathcal{H}_n \otimes \mathcal{H}_n$ decomposes into irreps under $\mathrm{U}(N)$.
\subsubsection{Young Tableaux}
This section follows Ref.~\cite{georgi}. The $n$-particle subspace $\mathcal{H}_n$ is irreducible and the irreps of its tensor products can be found quite easily using Young's tableaux. We introduce the notation $(n_1,n_2,\ldots,n_N)$ to denote a Young tableau with $n_1$ boxes in its first row, $n_2$ in its second row, and so on, where $n_1 \geq n_2 \geq \ldots \geq n_N$ and $n_1 + n_2+ \ldots n_N = n$ always, in which case each tableau corresponds to a partition of the integer $n$. If we have $k$ consecutive rows containing $n$ elements we can use the notation $n^k$. For example 
\begin{equation}
\ytableausetup{centertableaux}
(2) = \ydiagram{2} \quad 
(3,2) = \ydiagram{3,2} \quad \ytableausetup{centertableaux}
(4,2^2,1) = \ydiagram{4,2,2,1} \quad
(2^2,1^2) = \ydiagram{2,2,1,1}
\end{equation}
Given a Young tableau, its dimension is given by $d = F/H$, where
\begin{equation}
F = \prod_{i,j} F_{ij}, \quad H = \prod_{i,j} H_{ij},
\end{equation}
where $F_{ij}$ is a value assigned to each box $(i,j)$ in the diagram by first inserting an $N$ in the top left box and then increase the value by one as we move one step to the right or decreasing by one as we move one step down, and $H_{ij}$ is the hook factor of the box $(i,j)$ defined by first drawing a hook-shaped line that starts directly below it at the bottom of the column, turning right turn on the box at $(i,j)$, and moving out to the right and then counting number of boxes this line passes through. For example for $(m,n)$ we have 
\begin{align}
\ytableausetup
{mathmode, boxframe=normal, boxsize=3em}
\begin{ytableau}
\scriptstyle N & \scriptstyle N+1 & \scriptstyle \cdots & \scriptstyle \scriptstyle N + n - 1 &\scriptstyle N+n &  \scriptstyle \cdots & \scriptstyle N+m-1 \\
\scriptstyle N-1 & \scriptstyle N & \scriptstyle \cdots & \scriptstyle N+n-2  
\end{ytableau} & \quad F = \frac{(N+m-1)!}{(N-1)!} \cdot \frac{(N+n-2)!}{(N-2)!} \\
\begin{ytableau}
\scriptstyle m+1 & \scriptstyle m &  \scriptstyle \cdots & \scriptstyle m-n+2 & \scriptstyle m-\scriptstyle n & \scriptstyle \cdots & \scriptstyle 1 \\
\scriptstyle n & \scriptstyle n-1 &  \scriptstyle \cdots & \scriptstyle 1  
\end{ytableau} & \quad H = \frac{n!(m+1)!}{ (m-n + 1)}
\end{align}
so the dimension is given by
\begin{equation}
d(m,n)   = { N + m - 1 \choose m} {N + n - 2 \choose n} \frac{m-n+1}{m+1}. \label{eq:bosonic_tableau_dimension}
\end{equation}
If we repeat this calculation for Tableau with two columns instead, the dimension is given by
\begin{equation}
d(2^m,1^n) =  {N \choose m + n} {N+1 \choose m} \frac{n + 1}{m + n + 1} . \label{eq:fermionic_tableaudimension}
\end{equation}

\subsubsection{Derivation of examples \label{app:example_derivation}}
The Young tableaux corresponding to the bosonic and fermionic representation spaces $\mathcal{H}_n$ have the tableaux $\mathrm{Sym}^n(V) = (n)$ and $\wedge^n(V) = (1^n)$ respectively. Diagrammatically, these are a single row or column of $n$ boxes respectively. As required for evaluating the variance of the cost function in Eq.~\eqref{eq:variance}, we need to understand how the tensor product representations  $\mathcal{H}_n \otimes \mathcal{H}_n$ decompose into irreps. This can be calculated straightforwardly using the rules for tensoring Young tableaux, see Ref.~\cite{georgi}, which give
\begin{equation}
\mathrm{Sym}^n(V) \otimes \mathrm{Sym}^n(V) = \bigoplus_{\alpha = 0}^n (2n - \alpha, \alpha), \label{eq:bosonic_tensor_product}
\end{equation}
so the irreps are given by $\mathcal{H}_\alpha = (2n - \alpha, \alpha)$. For fermions we have 
\begin{equation}
\wedge^n(V) \otimes \wedge^n(V) = \bigoplus_{\alpha = 0}^{n} (2^\alpha, 1^{2n - 2\alpha}),
\end{equation}
so the irreps are given by $\mathcal{H}_\alpha = (2^\alpha, 1^{2n - 2\alpha})$, where the upper limit on the direct sum assumes that $n \leq N$ because a tableau with more than $N$ boxes is not allowed in a column as they are anti-symmetrised over. Each term in the direct sum is an irrep and has multiplicity one. 

We evaluate the variance of Eq.~\eqref{eq:variance} for three examples used in Fig.~\ref{fig:barren_plateau_scaling} of the main text. For all examples we take our measurement observable $H$ to be a projector onto the input state as $H = |\psi_\text{in}\rangle \langle \psi_\text{in}|$. From Theorem~\ref{thm:invariance_of_variance} the results below will hold for projectors onto any output state that differs to the input state up to a linear optical unitary transformation as $|\psi_\text{out}\rangle = U(u) |\psi_\text{in}\rangle$ for $u \in \mathrm{U}(N)$.
\begin{enumerate}
\item 
Let us consider fermionic statistics and prepare the system in the $n$-fermion pure state where we place a single particle in the first $n$ modes as 
\begin{equation}
|\psi_\text{in}\rangle = a_1^\dagger a_2^\dagger \ldots a_n^\dagger |0\rangle \in \wedge^n(V)
\end{equation}
This state has the weight $\lambda = (1,1,\ldots,1,0,\ldots,0)$ containing $n$ ones which is also the highest-weight state of $\wedge^n(V)$. The tensor product state $|\psi_\text{in}\rangle \otimes |\psi_\text{in}\rangle $ is the highest weight state of the tensor product representation $\wedge^n(V) \otimes \wedge^n(V)$ with weight $2\lambda$. Therefore, this state has support on only one irrep, namely the irrep corresponding to the Young Tableau $(2^n)$ which contains two columns of length $n$. A rigorous proof of this statement can be found in Supplementary Lemma 3 of Ref.~\cite{fontanacharacterizing}. This means that the sum over the irreps in Eq.~\eqref{eq:variance} reduces to a single term and the projector acts on the state as the identity, leaving us with the trace of a rank-$1$ projector which is unity. This gives us 
\begin{equation}
\begin{aligned}
\mathrm{Var}_{\boldsymbol{\theta}}[E(\boldsymbol{\theta})] & = \frac{1}{d(2^n)} - \frac{1}{D_n^2} \\
& = \frac{(N-n + 1)(n+1)}{N+1} \frac{1}{ { N \choose n }^2} - \frac{1}{D_n^2} \\
& = \frac{1}{D_n^2} \frac{n(N - n)}{N+1}
\end{aligned}
\end{equation}
where we used the fact that ${ N \choose n } = D_n$ which is the dimension of the Hilbert space of $n$ fermions and $N$ modes and used the dimension of the irreps from Eq.~\eqref{eq:fermionic_tableaudimension}. This result applies to projections onto any fermionic number state, as all number states are related via a permutation which can be implemented with a linear optical unitary. 
\item 
Now let us consider bosonic statistics and prepare the system in the $n$-boson state where we place all $n$ bosons into the first mode as
\begin{equation}
|\psi_\text{in}\rangle = \frac{1}{\sqrt{n}!} (a_1^\dagger)^n |0\rangle \in \mathrm{Sym}^n(V).
\end{equation}
Just like the in example 1, this state is the highest weight state of $\mathrm{Sym}^n(V)$ so the only irrep contributing anything non-trivial is the one with the Young Tableau $(2n)$ which contains $2n$ boxes in a single row. This is because the state $|\psi_\text{in}\rangle \otimes |\psi_\text{in}\rangle$ is completely symmetric on permutation of all $2n$ bosons, so it must be an element of $\mathrm{Sym}^{2n}(V)$ which has the tableau $(2n)$ and any anti-symmetrisation will annihilate the state. In this case, the variance is
\begin{equation}
\begin{aligned}
\mathrm{Var}_{\boldsymbol{\theta}}[E(\boldsymbol{\theta})] & = \frac{1}{d(2n)} - \frac{1}{D_n^2} \\
& = \frac{1}{{ N + 2n - 1 \choose 2n }} - \frac{1}{D_n^2} \\
& = \frac{1}{D_{2n}} - \frac{1}{D_n^2}
\end{aligned}
\end{equation}
where we used the bosonic dimension formula from Eq.~\eqref{eq:bosonic_tableau_dimension} fact that $D_{2n} = { N + 2n - 1 \choose 2n}$.
\item The third example is where we take the same state as in example 1, but with bosonic statistics as 
\begin{equation}
|\psi_\text{in}\rangle = a_1^\dagger a_2^\dagger \ldots a_n^\dagger |0\rangle \in \mathrm{Sym}^n(V).
\end{equation}
This is the most relevant bosonic input state for linear optics in a real experiment, as this state is produced by a single-photon source together with a demultiplexer. The difference now is that because of the bosonic statistics and the way the tensor product decomposes into irreps, the state $|\psi_\text{in}\rangle \otimes |\psi_\text{in}\rangle$ has overlap with multiple irreps. The state $|\psi_\text{in}\rangle \otimes |\psi_\text{in}\rangle$ is an element of the tensor product representation that decomposes as
\begin{equation}
\mathrm{Sym}^n(V) \otimes \mathrm{Sym}^n(V) = \mathrm{Sym}^2 ( \mathrm{Sym}^n(V) ) \oplus \wedge^2 (\mathrm{Sym}^n(V)).
\end{equation}
As $|\psi_\text{in}\rangle \otimes |\psi_\text{in}\rangle$ is symmetric under swaps across the tensor product, it must be an element of $\mathrm{Sym}^2(\mathrm{Sym}^n(V))$ which decomposes into irreps as
\begin{equation}
\mathrm{Sym}^2(\mathrm{Sym}^n(V)) = \bigoplus_{\text{even}  \ \alpha} (2n - \alpha, \alpha).
\end{equation}
In other words, it is the irreps of Eq.~\eqref{eq:bosonic_tensor_product} with even $\alpha$. Therefore, the variance is given by
\begin{equation}
\mathrm{Var}_{\boldsymbol{\theta}}[E(\boldsymbol{\theta})]  = \sum_{\text{even} \ \alpha} \frac{1}{d_\alpha} \mathrm{Tr}\left( P^\alpha \rho^{\otimes 2} \right) \mathrm{Tr} \left( P^\alpha H^{\otimes 2} \right) - \frac{ \mathrm{Tr}^2(H)}{D_n^2} ,
\end{equation}
where $\rho = H = |\psi_\text{in}\rangle \langle \psi_\text{in}|$. As shown in Ref.~\cite{kolarovszki2026generalframeworkanticoncentrationlinear}  this projection can be evaluated to give us
\begin{equation}
\mathrm{Tr}(P^\alpha \rho^{\otimes 2}) = \mathrm{Tr}(P^\alpha H^{\otimes 2}) =  2^{n - \alpha} \frac{2n - 2\alpha + 1}{2n - \alpha + 1} \frac{{ n \choose \alpha /2} }{ { 2n - \alpha \choose n -\alpha /2}}. 
\end{equation}
Therefore, inserting this into the expression for the variance gives us
\begin{equation}
\mathrm{Var}_{\boldsymbol{\theta}} [ E(\boldsymbol{\theta})] = \sum_{\text{even $\alpha$}} \frac{4^{n - \alpha}}{d_\alpha} \left( \frac{2n - 2\alpha + 1}{2n - \alpha + 1} \frac{{ n \choose \alpha /2} }{ { 2n - \alpha \choose n -\alpha /2}} \right)^2 - \frac{1}{D_n^2} .
\end{equation}

\end{enumerate}

\end{document}